\documentclass[reqno,a4paper]{amsart}
\usepackage{amssymb}
\usepackage{amsfonts}
\usepackage{geometry}

\newtheorem{theorem}{Theorem}
\theoremstyle{plain}

\newtheorem{corollary}{Corollary}

\newtheorem{definition}{Definition}
\newtheorem{example}{Example}

\newtheorem{lemma}{Lemma}
\newtheorem{notation}{Notation}
\newtheorem{problem}{Problem}
\newtheorem{proposition}{Proposition}
\newtheorem{remark}{Remark}

\numberwithin{equation}{section}
\input{tcilatex}
\begin{document}
\title[Space-Time Bound Derivation of 3D NLS with A Quadratic Trap]{On the
Rigorous Derivation of the 3D Cubic Nonlinear Schr\"{o}dinger Equation with
A Quadratic Trap}
\author{Xuwen Chen}
\address{Department of Mathematics, Brown University, 151 Thayer Street,
Providence, RI 02912}
\email{chenxuwen@math.umd.edu / chenxuwen@math.brown.edu}
\urladdr{http://www.math.brown.edu/\symbol{126}chenxuwen/}
\date{V3 for ARMA, 04/04/2013}
\dedicatory{Dedicated to Xuqing.}

\begin{abstract}
We consider the dynamics of the 3D $N$-body Schr\"{o}dinger equation in the
presence of a quadratic trap. We assume the pair interaction potential is $%
N^{3\beta -1}V(N^{\beta }x).$ We justify the mean-field approximation and
offer a rigorous derivation of the 3D cubic NLS with a quadratic trap. We
establish the space-time bound conjectured by Klainerman and Machedon \cite%
{KlainermanAndMachedon} for $\beta \in \left( 0,2/7\right] $ by adapting and
simplifying an argument in Chen and Pavlovi\'{c} \cite{TChenAndNPSpace-Time}
which solves the problem for $\beta \in \left( 0,1/4\right) $ in the absence
of a trap.
\end{abstract}

\maketitle
\tableofcontents

\section{Introduction}

It is widely believed that the cubic nonlinear Schr\"{o}dinger equation (NLS)%
\begin{equation*}
i\partial _{t}\phi =L\phi +\left\vert \phi \right\vert ^{2}\phi ,
\end{equation*}%
where $L$ is the Laplacian $-\triangle $ or the Hermite operator $-\triangle
+\omega ^{2}\left\vert x\right\vert ^{2},$ describes the physical phenomenon
of Bose-Einstein condensation (BEC). This belief is one of the main
motivations for studying the cubic NLS. BEC is the phenomenon that particles
of integer spin (Bosons) occupy a macroscopic quantum state. This unusual
state of matter was first predicted theoretically by Einstein for
non-interacting particles. The first experimental observation of BEC in an
interacting atomic gas did not occur until 1995 using laser cooling
techniques \cite{Anderson, Davis}. E. A. Cornell, W. Ketterle, and C. E.
Wieman were awarded the 2001 Nobel Prize in Physics for observing BEC. Many
similar successful experiments were performed later on \cite{Cornish,
Philips, Ketterle, Stamper}.

Let $t\in \mathbb{R}$ be the time variable and $\mathbf{x}_{N}=\left(
x_{1},x_{2},...,x_{N}\right) \in \mathbb{R}^{3N}$ be the position vector of $%
N$ particles in $\mathbb{R}^{3}$. Then BEC naively means that the N-body
wave function $\psi _{N}(t,\mathbf{x}_{N})$ satisfies 
\begin{equation}
\psi _{N}(t,\mathbf{x}_{N})=\dprod\limits_{j=1}^{N}\phi (t,x_{j})
\label{formula:BEC state}
\end{equation}%
up to a phase factor solely depending on $t$, for some one particle state $%
\phi .$ In other words, every particle is in the same quantum state. It is
believed that the one particle wave function $\phi $ which models the
condensate satisfies the cubic NLS. Gross \cite{Gr1,Gr2} and Pitaevskii \cite%
{Pitaevskii} proposed such a description. However, the cubic NLS is a
phenomenological mean field type equation and its validity needs to be
established rigorously from the many body system which it is supposed to
characterize. As a result, we investigate the procedure of laboratory
experiments of BEC according to \cite{Anderson, Davis}.

\begin{itemize}
\item[Step A.] Confine a large number of Bosons inside a trap e.g., the
magnetic fields in \cite{Anderson, Davis}. Cool it down so that the many
body system reaches its ground state. It is expected that this ground state
is a BEC state / factorized state. This step corresponds to the mathematical
problem.

\begin{problem}
\label{Problem:Lieb}Show that the ground state of the N-body Hamiltonian%
\begin{equation*}
\sum_{j=1}^{N}\left( -\frac{1}{2}\triangle _{x_{j}}+\frac{\omega _{0}^{2}}{2}%
\left\vert x_{j}\right\vert ^{2}\right) +\frac{1}{N}\sum_{1\leqslant
i<j\leqslant N}N^{3\beta }V\left( N^{\beta }\left( x_{i}-x_{j}\right) \right)
\end{equation*}%
is a factorized state.
\end{problem}

We use the quadratic potential $\left\vert x\right\vert ^{2}$ to represent
the trap. This simplified yet reasonably general model is expected to
capture the salient features of the actual trap: on the one hand the
quadratic potential varies slowly, on the other hand it tends to $\infty $
as $\left\vert x\right\vert \rightarrow \infty $. In the physics literature,
Lieb, Seiringer and Yngvason remarked in \cite{Lieb1} that the confining
potential is typically $\sim \left\vert x\right\vert ^{2}$ in the available
experiments. Mathematically speaking, the strongest trap we can deal with in
the usual regularity setting of NLS is the quadratic trap since the work 
\cite{Yajima} by Yajima and Zhang points out that the ordinary Strichartz
estimates start to fail as the trap exceeds quadratic.

\item[Step B.] Switch the trap in order to enable measurement or direct
observation. It is assumed that such a shift of the confining potential is
instant and does not destroy the BEC obtained from Step A. To be more
precise about the word "switch": in \cite{Anderson, Davis} the trap is
removed, in \cite{Stamper} the initial magnetic trap is switched to an
optical trap, in \cite{Cornish} the trap is enhanced, in \cite{Philips} the
trap is turned off in 2 spatial directions to generate a 2D Bose gas. Hence
we have a different trap after the switch, in other words, the trapping
potential becomes $\omega ^{2}\left\vert x\right\vert ^{2}$. The system is
then time dependent unless $\omega =\omega _{0}$. Therefore, the factorized
structure obtained in Step A must be preserved in time for the observation
of BEC. Mathematically, this step stands for the following problem.

\begin{problem}
\label{Problem:Mine}Take the BEC state obtained in Step A. as initial datum,
show that the solution to the many body Schr\"{o}dinger equation%
\begin{equation}
i\partial _{t}\psi _{N}=\sum_{j=1}^{N}\left( -\frac{1}{2}\triangle
_{x_{j}}+\omega ^{2}\frac{\left\vert x_{j}\right\vert ^{2}}{2}\right) \psi
_{N}+\frac{1}{N}\sum_{1\leqslant i<j\leqslant N}N^{3\beta }V(N^{\beta
}\left( x_{i}-x_{j}\right) )\psi _{N}
\label{equation:N-Body Schrodinger with switchable trap}
\end{equation}%
is a BEC state / factorized state.
\end{problem}
\end{itemize}

We first remark that neither of the problems listed above admits a
factorized state solution. It is also unrealistic to solve the equations in
Problems \ref{Problem:Lieb} and \ref{Problem:Mine} for large $N$. Moreover,
both problems are linear so that it is not clear how the cubic NLS arises
from either problem. Therefore, in order to justify the statement that the
cubic NLS depicts BEC, we have to show mathematically that, in an
appropriate sense,%
\begin{equation*}
\psi _{N}(t,\mathbf{x}_{N})\sim \dprod\limits_{j=1}^{N}\phi (t,x_{j})\text{
as }N\rightarrow \infty
\end{equation*}%
for some one particle state $\phi $ which solves a cubic NLS. However, when $%
\phi \neq \phi ^{\prime }$%
\begin{equation*}
\left\Vert \dprod\limits_{j=1}^{N}\phi (t,x_{j})-\dprod\limits_{j=1}^{N}\phi
^{\prime }(t,x_{j})\right\Vert _{L^{2}}^{2}\rightarrow 2\text{ }as\text{ }%
N\rightarrow \infty .
\end{equation*}%
i.e. our desired limit (the BEC state) is not stable against small
perturbations. One way to circumvent this difficulty is to use the concept
of the k-particle marginal density $\gamma _{N}^{(k)}$ associated with $\psi
_{N}\ $defined as%
\begin{equation}
\gamma _{N}^{(k)}(t,\mathbf{x}_{k};\mathbf{x}_{k}^{\prime })=\int \psi
_{N}(t,\mathbf{x}_{k},\mathbf{x}_{N-k})\overline{\psi _{N}}(t,\mathbf{x}%
_{k}^{\prime },\mathbf{x}_{N-k})d\mathbf{x}_{N-k},\text{ }\mathbf{x}_{k},%
\mathbf{x}_{k}^{\prime }\in \mathbb{R}^{3k}  \label{def:marginal density}
\end{equation}%
and show that%
\begin{equation*}
\gamma _{N}^{(k)}(t,\mathbf{x}_{k};\mathbf{x}_{k}^{\prime })\sim
\dprod\limits_{j=1}^{k}\phi (t,x_{j})\bar{\phi}(t,x_{j}^{\prime })\text{ as }%
N\rightarrow \infty .
\end{equation*}%
Penrose and Onsager \cite{Penrose} suggested such a formulation. Another
approach is to add a second order correction to the mean field
approximation. See \cite{Chen2ndOrder, GMM1, GMM2}.

For Problem \ref{Problem:Lieb}, Lieb, Seiringer, Solovej and Yngvason showed
that the ground state of the Hamiltonian exhibits complete BEC in \cite%
{Lieb2}, provided that the trapping potential $V_{trap}(x)$ satisfies $%
\inf_{\left\vert x\right\vert >R}V_{trap}(x)$ $\rightarrow \infty $ for $%
R\rightarrow \infty $ and the interaction potential is spherically
symmetric. To be more precise, let $\psi _{N,0}$ be the ground state, then 
\begin{equation*}
\gamma _{N,0}^{(1)}\rightarrow \left\vert \phi _{GP}\right\rangle
\left\langle \phi _{GP}\right\vert \text{ as }N\rightarrow \infty ,
\end{equation*}%
where $\gamma _{N,0}^{(1)}$ is the corresponding one particle marginal
density defined via formula \ref{def:marginal density} and $\phi _{GP}$ is
the minimizer of the Gross-Pitaevskii energy functional with coupling
constant $4\pi a_{0}\footnote{%
Here $a_{0}$ is the $N\rightarrow \infty $ limit of $N$ times the scattering
length of $N^{3\beta -1}V(N^{\beta }\cdot ).$ i.e. $a_{0}$ is the scattering
length of $V$ when $\beta =1$ and $a_{0}$ should be $\left( \int V\right)
/8\pi $ when $\beta \in \left( 0,1\right) $. See \cite{E-S-Y2}.},$ 
\begin{equation*}
\int \mathbf{(}\left\vert \nabla \phi \right\vert ^{2}+V_{trap}(x)\left\vert
\phi \right\vert ^{2}+4\pi a_{0}\left\vert \phi \right\vert ^{4}\mathbf{)}dx%
\mathbf{.}
\end{equation*}

So far, there has not been any work regarding the Hamiltonian evolution in
Step B in the case when $\omega \neq 0$. Motivated by the above
considerations, we aim to investigate the evolution of a many-body Boson
system with a quadratic trap and study the dynamics after the switch of the
trap. We derive rigorously the 3D cubic NLS with a quadratic trap from the $%
N-$body linear equation \ref{equation:N-Body Schrodinger with switchable
trap}. To be specific, we establish the following theorem in this paper.

\begin{theorem}
\label{Theorem:3D BEC}(Main Theorem) Let $\left\{ \gamma _{N}^{(k)}\right\} $
be the family of marginal densities associated with $\psi _{N}$, the
solution of the N-body Schr\"{o}dinger equation \ref{equation:N-Body
Schrodinger with switchable trap} for some $\beta \in \left( 0,\frac{2}{7}%
\right] $. Assume that the pair interaction $V\ $is a nonnegative $L^{1}(%
\mathbb{R}^{3})\cap H^{2}(\mathbb{R}^{3})\cap W^{2,\infty }(\mathbb{R}^{3})$
spherically symmetric function. Moreover, suppose the initial datum of
equation \ref{equation:N-Body Schrodinger with switchable trap} verifies the
following the conditions:

(a) the initial datum is normalized i.e. 
\begin{equation*}
\left\Vert \psi _{N}(0)\right\Vert _{L^{2}}=1,
\end{equation*}

(b) the initial datum is asymptotically factorized%
\begin{equation}
\lim_{N\rightarrow \infty }\limfunc{Tr}\left\vert \gamma
_{N}^{(1)}(0,x_{1};x_{1}^{\prime })-\phi _{0}(x_{1})\overline{\phi _{0}}%
(x_{1}^{\prime })\right\vert =0,  \label{eqn:asym factorized}
\end{equation}%
for some one particle wave function $\phi _{0}\in H^{1}\left( \mathbb{R}%
^{3}\right) $.

(c) the initial datum has bounded energy per particle i.e. 
\begin{equation*}
\sup_{N}\frac{1}{N}\left\langle \psi _{N}(0),H_{N}\psi _{N}(0)\right\rangle
<\infty ,
\end{equation*}%
where the Hamiltonian $H_{N}$ is%
\begin{equation*}
H_{N}=\sum_{j=1}^{N}\left( -\frac{1}{2}\triangle _{x_{j}}+\omega ^{2}\frac{%
\left\vert x_{j}\right\vert ^{2}}{2}\right) +\frac{1}{N}\sum_{1\leqslant
i<j\leqslant N}N^{3\beta }V(N^{\beta }\left( x_{i}-x_{j}\right) ).
\end{equation*}%
Then $\forall t\geqslant 0$, $\forall k\geqslant 1$, we have the convergence
in the trace norm that%
\begin{equation*}
\lim_{N\rightarrow \infty }\limfunc{Tr}\left\vert \gamma _{N}^{(k)}(t,%
\mathbf{x}_{k};\mathbf{x}_{k}^{\prime })-\dprod_{j=1}^{k}\phi (t,x_{j})%
\overline{\phi }(t,x_{j}^{\prime })\right\vert =0,
\end{equation*}%
where $\phi (t,x)$ is the solution to the 3D cubic NLS with a quadratic trap%
\begin{eqnarray}
i\partial _{t}\phi &=&\left( -\frac{1}{2}\triangle _{x}+\omega ^{2}\frac{%
\left\vert x\right\vert ^{2}}{2}\right) \phi +b_{0}\left\vert \phi
\right\vert ^{2}\phi \text{ in }\mathbb{R}^{3+1}
\label{equation:cubicNLSwithSwitch} \\
\phi (0,x) &=&\phi _{0}(x)  \notag
\end{eqnarray}%
and the coupling constant $b_{0}=\int_{\mathbb{R}^{3}}V(x)dx.$
\end{theorem}

For the $\omega =0$ case, the approach which uses the marginal densities $%
\left\{ \gamma _{N}^{(k)}\right\} $ for the dynamics problem has been proven
to be successful in the fundamental papers \cite{E-E-S-Y1, E-Y1,
E-S-Y1,E-S-Y2,E-S-Y4, E-S-Y5, E-S-Y3} by Elgart, Erd\"{o}s, Schlein, and
Yau. As pointed out in \cite{E-S-Y3}, their work corresponds to the
evolution after the removal of the traps. Motivated by a kinetic formulation
of Spohn \cite{Spohn}, their program consists of two principal parts: in one
part, they prove that an appropriate limit of the sequence $\left\{ \gamma
_{N}^{(k)}\right\} $ as $N\rightarrow \infty $ solves the Gross-Pitaevskii
hierarchy 
\begin{equation}
\left( i\partial _{t}+\frac{1}{2}\triangle _{\mathbf{x}_{k}}-\frac{1}{2}%
\triangle _{\mathbf{x}_{k}^{\prime }}\right) \gamma
^{(k)}=b_{0}\sum_{j=1}^{k}B_{j,k+1}\left( \gamma ^{(k+1)}\right) ,\text{ }%
k=1,...,n,...  \label{equation:Gross-Pitaevskii hiearchy without a trap}
\end{equation}%
where $B_{j,k+1}=B_{j,k+1}^{1}-B_{j,k+1}^{2}$ and%
\begin{eqnarray*}
&&B_{j,k+1}^{1}\left( \gamma ^{(k+1)}\right) (t,\mathbf{x}_{k};\mathbf{x}%
_{k}^{\prime })=\int \int \delta (x_{j}-x_{k+1})\delta
(x_{k+1}-x_{k+1}^{\prime })\gamma ^{(k+1)}(t,\mathbf{x}_{k+1};\mathbf{x}%
_{k+1}^{\prime })dx_{k+1}dx_{k+1}^{\prime }, \\
&&B_{j,k+1}^{2}\left( \gamma ^{(k+1)}\right) (t,\mathbf{x}_{k};\mathbf{x}%
_{k}^{\prime })=\int \int \delta (x_{j}^{\prime }-x_{k+1})\delta
(x_{k+1}-x_{k+1}^{\prime })\gamma ^{(k+1)}(t,\mathbf{x}_{k+1};\mathbf{x}%
_{k+1}^{\prime })dx_{k+1}dx_{k+1}^{\prime };
\end{eqnarray*}%
in another part, they show that hierarchy $\ref{equation:Gross-Pitaevskii
hiearchy without a trap}$ has a unique solution which is therefore a
completely factorized state. However, as remarked by Terence Tao, the
uniqueness theory for hierarchy $\ref{equation:Gross-Pitaevskii hiearchy
without a trap}$ is surprisingly delicate due to the fact that it is a
system of infinitely many coupled equations over an unbounded number of
variables. In \cite{KlainermanAndMachedon}, by assuming a space-time bound
on the limit of $\left\{ \gamma _{N}^{(k)}\right\} $, Klainerman and
Machedon gave another proof of the uniqueness in \cite{E-S-Y2} through a
collapsing estimate originated from the ordinary multilinear Strichartz
estimates in their null form paper \cite{KlainermanMachedonNullForm} and a
board game argument inspired by the Feynman graph argument in \cite{E-S-Y2}.

Later, the method in Klainerman and Machedon \cite{KlainermanAndMachedon}
was taken up by Kirkpatrick, Schlein, and Staffilani \cite{Kirpatrick}, who
studied the corresponding problem in 2D; by Chen and Pavlovi\'{c} \cite%
{TChenAndNpGP1, TChenAndNP}, who considered the 1D and 2D 3-body interaction
problem and the general existence theory of hierarchy $\ref%
{equation:Gross-Pitaevskii hiearchy without a trap}$; and by the author \cite%
{ChenAnisotropic}, who investigated the trapping problem in 2D. In \cite%
{TCNPNT, TCNPNT1}, Chen, Pavlovi\'{c} and Tzirakis worked out the virial and
Morawetz identities for hierarchy \ref{equation:Gross-Pitaevskii hiearchy
without a trap}.

Recently, for the 3D case without traps, Chen and Pavlovi\'{c} \cite%
{TChenAndNPSpace-Time} proved that, for $\beta \in (0,1/4)$, the limit of $%
\left\{ \gamma _{N}^{(k)}\right\} $ actually satisfies the space-time bound
assumed by Klainerman and Machedon \cite{KlainermanAndMachedon} as $%
N\rightarrow \infty $. This has been a well-known open problem in the field.
Moreover, they showed that the solution to the BBGKY hierarchy converges
strongly to the solution to hierarchy \ref{equation:Gross-Pitaevskii
hiearchy without a trap} in $H^{1}$ without assuming asymptotically
factorized initial datum. In this paper, we adapt and simplify their
argument in establishing the Klainerman-Machedon space-time bound. We also
extend the range of $\beta $ from $\left( 0,1/4\right) $ in Chen and Pavlovi%
\'{c} \cite{TChenAndNPSpace-Time} to $\left( 0,2/7\right] .$ Through simple
functional analysis, we obtain a convergence result without assuming
asymptotically factorized initial datum as well. (See Corollary \ref%
{Corollary:general data convergence}) But we are not claiming it as a main
result in this paper. We compare our result and the one in Chen and Pavlovi%
\'{c} \cite{TChenAndNPSpace-Time} briefly in Section \ref%
{subsection:comparison} below.

\subsection{Comparison with Chen and Pavlovi\'{c} \protect\cite%
{TChenAndNPSpace-Time}\label{subsection:comparison}}

For comparison purpose, we transform Theorem \ref%
{Theorem:SmoothVersionofTheMainTheorem} in Section \ref%
{Section:ProofOfTheMainTheorem}, which implies Theorem \ref{Theorem:3D BEC},
into the general convergence result below without the assumption of
asymptotically factorized initial data since the result in Chen and Pavlovi%
\'{c} \cite{TChenAndNPSpace-Time} is under the same regularity setting
(condition \ref{condition:high energy bound}) as Theorem \ref%
{Theorem:SmoothVersionofTheMainTheorem}.

\begin{corollary}
\label{Corollary:general data convergence}(For comparison purpose only, not
a main result.) Let $\gamma _{N}^{(k)}(t,\mathbf{x}_{k};\mathbf{x}%
_{k}^{\prime })$ be the k-marginal density associated with the solution of
the N-body Schr\"{o}dinger equation \ref{equation:N-Body Schrodinger with
switchable trap} for some $\beta \in \left( 0,2/7\right] $\ and $\gamma
^{(k)}(t,\mathbf{x}_{k};\mathbf{x}_{k}^{\prime })$ be the solution to the
Gross-Pitaevskii hierarchy with a quadratic trap $\left( \text{hierarchy }%
\ref{hierarchy:TheGPHierarchyWithTraps}\right) $. Assume the initial data of 
$\gamma _{N}^{(k)}$ and $\gamma ^{(k)}$ satisfy the following conditions

($a^{\prime }$)%
\begin{equation*}
\lim_{N\rightarrow \infty }\limfunc{Tr}\left\vert \gamma _{N}^{(k)}(0,%
\mathbf{x}_{k};\mathbf{x}_{k}^{\prime })-\gamma ^{(k)}(0,\mathbf{x}_{k};%
\mathbf{x}_{k}^{\prime })\right\vert =0,
\end{equation*}

($b^{\prime }$)%
\begin{equation}
\left\langle \psi _{N}(0),H_{N}^{k}\psi _{N}(0)\right\rangle \leqslant
C^{k}N^{k}.  \label{condition:high energy bound}
\end{equation}%
Then we have the convergence of the evolution in trace norm%
\begin{equation*}
\lim_{N\rightarrow \infty }\limfunc{Tr}\left\vert \gamma _{N}^{(k)}(t,%
\mathbf{x}_{k};\mathbf{x}_{k}^{\prime })-\gamma ^{(k)}(t,\mathbf{x}_{k};%
\mathbf{x}_{k}^{\prime })\right\vert =0.
\end{equation*}
\end{corollary}

\begin{proof}
This result for the non-trap case should be credited to Erd\"{o}s, Schlein
and Yau though they did not state it as one of their main theorems. They
mentioned it on page 297 of their paper \cite{E-S-Y3}. Once we have
established Theorem \ref{Theorem:SmoothVersionofTheMainTheorem}, its proof
together with some simple functional analysis proves this corollary. We
include the proof in Appendix I (Section \ref{Section:Appendix-0}) for
completeness.
\end{proof}

Briefly, the main result (Theorem 3.1) in Chen and Pavlovi\'{c} \cite%
{TChenAndNPSpace-Time} is the following.

\begin{theorem}
\cite{TChenAndNPSpace-Time} Let $\omega =0$. Suppose $\gamma _{N}^{(k)}(t,%
\mathbf{x}_{k};\mathbf{x}_{k}^{\prime })$ is the k-marginal density
associated with the solution of the N-body Schr\"{o}dinger equation \ref%
{equation:N-Body Schrodinger with switchable trap} for some $\beta \in
(0,1/4)$ and $\gamma ^{(k)}(t,\mathbf{x}_{k};\mathbf{x}_{k}^{\prime })$ is
the solution to the ordinary Gross-Pitaevskii hierarchy without traps$.$ If
the initial data of $\gamma _{N}^{(k)}$ and $\gamma ^{(k)}$ satisfy the
following conditions

($a^{\prime \prime }$)%
\begin{equation*}
\lim_{N\rightarrow \infty }\left\Vert \left( \dprod_{j=1}^{k}\left(
1-\triangle _{x_{j}}\right) ^{\frac{1}{2}}\left( 1-\triangle _{x_{j}^{\prime
}}\right) ^{\frac{1}{2}}\right) \left( \gamma _{N}^{(k)}(0,\mathbf{x}_{k};%
\mathbf{x}_{k}^{\prime })-\gamma ^{(k)}(0,\mathbf{x}_{k};\mathbf{x}%
_{k}^{\prime })\right) \right\Vert _{L^{2}(d\mathbf{x}_{k}d\mathbf{x}%
_{k}^{\prime })}=0,
\end{equation*}

($b^{\prime \prime }$)%
\begin{equation*}
\left\langle \psi _{N}(0),H_{N}^{k}\psi _{N}(0)\right\rangle \leqslant
C^{k}N^{k},
\end{equation*}%
then we have the convergence of the evolution in $H^{1}$ norm%
\begin{equation*}
\lim_{N\rightarrow \infty }\left\Vert \left( \dprod_{j=1}^{k}\left(
1-\triangle _{x_{j}}\right) ^{\frac{1}{2}}\left( 1-\triangle _{x_{j}^{\prime
}}\right) ^{\frac{1}{2}}\right) \left( \gamma _{N}^{(k)}(t,\mathbf{x}_{k};%
\mathbf{x}_{k}^{\prime })-\gamma ^{(k)}(t,\mathbf{x}_{k};\mathbf{x}%
_{k}^{\prime })\right) \right\Vert _{L^{2}(d\mathbf{x}_{k}d\mathbf{x}%
_{k}^{\prime })}=0.
\end{equation*}
\end{theorem}

One can then easily tell the following: on the one hand, the result in this
work allows a quadratic trap and a larger range of $\beta $ in the analysis
(these are also the main novelty and the main technical improvement of this
paper); on the other hand, the result in \cite{TChenAndNPSpace-Time} yields
a stronger convergence ($H^{1}$ convergence) when the initial data admits a
stronger convergence. The main purpose of this paper is to justify the
mean-field approximation and offer a rigorous derivation of the 3D cubic NLS
with a quadratic trap (Theorem \ref{Theorem:3D BEC}). Thus we establish the
convergence of probability densities which is the trace norm convergence in
this context. The Chen-Pavlovi\'{c} result is crucial for their work on the
Cauchy problem of Gross-Pitaevskii hierarchies. Before we delve into the
proofs, we remark that taking the coupling level $k$ to be $\ln N$ in
Section \ref{SubSection:InteractionPartOfBBGKY} is exactly the place where
we follow the original idea of Chen and Pavlovi\'{c} \cite%
{TChenAndNPSpace-Time}.

\subsection{The Anisotropic Version of the Main Theorem}

It is of interest to use anisotropic traps in laboratory experiments. (See,
for example, \cite{Philips}.) Our proof for Theorem \ref{Theorem:3D BEC}
also applies to the case with anisotropic traps. In fact, we have the
following theorem.

\begin{theorem}
\label{Theorem:Anisotropic}Define a 3*3 diagonal matrix $Q$ whose diagonal
entries are $(\omega _{1}^{2},\omega _{2}^{2},\omega _{3}^{2})$. Let $%
\left\{ \gamma _{N}^{(k)}\right\} $ be the family of marginal densities
associated with $\psi _{N}$, the solution of the N-body Schr\"{o}dinger
equation 
\begin{equation*}
i\partial _{t}\psi _{N}=\sum_{j=1}^{N}\left( -\frac{1}{2}\triangle _{x_{j}}+%
\frac{1}{2}x_{j}^{T}Qx_{j}\right) \psi _{N}+\frac{1}{N}\sum_{1\leqslant
i<j\leqslant N}N^{3\beta }V(N^{\beta }\left( x_{i}-x_{j}\right) )\psi _{N}
\end{equation*}%
for some $\beta \in \left( 0,\frac{2}{7}\right] $. Assume that the pair
interaction $V\ $is a nonnegative $L^{1}(\mathbb{R}^{3})\cap H^{2}(\mathbb{R}%
^{3})\cap W^{2,\infty }(\mathbb{R}^{3})$ spherically symmetric function.
Moreover, suppose the initial datum of equation \ref{equation:N-Body
Schrodinger with switchable trap} verifies Conditions (a), (b), and (c) in
the assumption of Theorem \ref{Theorem:3D BEC}. Then $\forall t\geqslant 0$, 
$\forall k\geqslant 1$, we have the convergence in the trace norm that%
\begin{equation*}
\lim_{N\rightarrow \infty }\limfunc{Tr}\left\vert \gamma _{N}^{(k)}(t,%
\mathbf{x}_{k};\mathbf{x}_{k}^{\prime })-\dprod_{j=1}^{k}\phi (t,x_{j})%
\overline{\phi }(t,x_{j}^{\prime })\right\vert =0,
\end{equation*}%
where $\phi (t,x)$ is the solution to the 3D cubic NLS with anisotropic
quadratic traps%
\begin{eqnarray*}
i\partial _{t}\phi &=&\left( -\frac{1}{2}\triangle _{x}+\frac{1}{2}%
x^{T}Qx\right) \phi +b_{0}\left\vert \phi \right\vert ^{2}\phi \text{ in }%
\mathbb{R}^{3+1} \\
\phi (0,x) &=&\phi _{0}(x)
\end{eqnarray*}%
and $b_{0}=\int_{\mathbb{R}^{3}}V(x)dx.$
\end{theorem}

The anisotropic version of the main theorem stated above yields to the same
techniques as Theorem \ref{Theorem:3D BEC} but is technically more
complicated. Therefore we prove only the latter in detail, merely suggesting
during the course of the proof the appropriate modifications needed to
obtain the more general theorem.

\section{Proof of the Main Theorem (Theorem \protect\ref{Theorem:3D BEC}) 
\label{Section:ProofOfTheMainTheorem}}

We establish Theorem \ref{Theorem:3D BEC} with a smooth approximation
argument and the following theorem.

\begin{theorem}
\label{Theorem:SmoothVersionofTheMainTheorem}Let $\left\{ \gamma
_{N}^{(k)}\right\} $ be the family of marginal densities associated with $%
\psi _{N}$, the solution of the N-body Schr\"{o}dinger equation \ref%
{equation:N-Body Schrodinger with switchable trap} for some $\beta \in
\left( 0,\frac{2}{7}\right] $. Assume that the pair interaction $V\ $is a
nonnegative $L^{1}(\mathbb{R}^{3})\cap H^{2}(\mathbb{R}^{3})\cap W^{2,\infty
}(\mathbb{R}^{3})$ spherically symmetric function. Moreover, suppose the
initial datum of equation \ref{equation:N-Body Schrodinger with switchable
trap} is normalized and asymptotically factorized in the sense of (a) and
(b) of Theorem \ref{Theorem:3D BEC}, and verifies the condition that there
is a $C$ independent of $k$ or $N$ such that%
\begin{equation*}
\left\langle \psi _{N}(0),H_{N}^{k}\psi _{N}(0)\right\rangle \leqslant
C^{k}N^{k}
\end{equation*}%
Then $\forall t\geqslant 0$ and $\forall k\geqslant 1$, we have the
convergence in the trace norm that%
\begin{equation*}
\lim_{N\rightarrow \infty }\limfunc{Tr}\left\vert \gamma _{N}^{(k)}(t,%
\mathbf{x}_{k};\mathbf{x}_{k}^{\prime })-\dprod_{j=1}^{k}\phi (t,x_{j})%
\overline{\phi }(t,x_{j}^{\prime })\right\vert =0,
\end{equation*}%
where $\phi (t,x)$ is the solution to the 3D cubic NLS with a quadratic trap%
\begin{eqnarray*}
i\partial _{t}\phi &=&\left( -\frac{1}{2}\triangle _{x}+\omega ^{2}\frac{%
\left\vert x\right\vert ^{2}}{2}\right) \phi +b_{0}\left\vert \phi
\right\vert ^{2}\phi \text{ in }\mathbb{R}^{3+1} \\
\phi (0,x) &=&\phi _{0}(x)
\end{eqnarray*}%
and $b_{0}=\int_{\mathbb{R}^{3}}V(x)dx.$
\end{theorem}

Before presenting the proof of Theorem \ref%
{Theorem:SmoothVersionofTheMainTheorem}, we discuss how to deduce Theorem %
\ref{Theorem:3D BEC} from Theorem \ref{Theorem:SmoothVersionofTheMainTheorem}%
. It is a well-known smooth approximation argument. We include it for
completeness. For technical details, we refer the readers to Erd\"{o}%
s-Schlein-Yau \cite{E-S-Y5, E-S-Y3} and Kirkpatrick-Schlein-Staffilani \cite%
{Kirpatrick}.

Write the spaces of compact operators and trace class operators of $L^{2}(%
\mathbb{R}^{3k})$ as $\mathcal{K}_{k}$ and $\mathcal{L}_{k}^{1}$. Then $%
\left( \mathcal{K}_{k}\right) ^{\prime }=\mathcal{L}_{k}^{1}$. Via the fact
that $\mathcal{K}_{k}$ is separable, we select a dense countable subset of
the unit ball of $\mathcal{K}_{k}$ and call it $\left\{ J_{i}^{(k)}\right\}
_{i\geqslant 1}\subset \mathcal{K}_{k}$. We have $\left\Vert
J_{i}^{(k)}\right\Vert \leqslant 1$ where $\left\Vert \cdot \right\Vert $ is
the operator norm. We set up the following metric on $\mathcal{L}_{k}^{1}$,
for $\gamma ^{(k)}$, $\tilde{\gamma}^{(k)}\in \mathcal{L}_{k}^{1}:$ define 
\begin{equation*}
d_{k}(\gamma ^{(k)},\tilde{\gamma}^{(k)})=\sum_{i=1}^{\infty
}2^{-i}\left\vert \limfunc{Tr}J_{i}^{(k)}\left( \gamma ^{(k)}-\tilde{\gamma}%
^{(k)}\right) \right\vert .
\end{equation*}%
Then a uniformly bounded sequence $\gamma _{N}^{(k)}\in \mathcal{L}_{k}^{1}$
converges to $\gamma ^{(k)}\in \mathcal{L}_{k}^{1}$ with respect to the
weak* topology if and only if 
\begin{equation*}
\lim_{N\rightarrow \infty }d_{k}(\gamma _{N}^{(k)},\gamma ^{(k)})=0.
\end{equation*}

Fix $\kappa >0$ and $\chi \in C_{c}^{\infty }\left( \mathbb{R}\right) $,
with $0\leqslant \chi \leqslant 1,$ $\chi (s)=1$, for $0\leqslant s\leqslant
1,$ and $\chi (s)=0$, for $s\geqslant 2.$ We regularize the initial data of
the $N$-body Schr\"{o}dinger equation \ref{equation:N-Body Schrodinger with
switchable trap} with%
\begin{equation*}
\tilde{\psi}_{N}\left( 0\right) =\frac{\chi \left( \kappa H_{N}/N\right)
\psi _{N}\left( 0\right) }{\left\Vert \chi \left( \kappa H_{N}/N\right) \psi
_{N}\left( 0\right) \right\Vert },
\end{equation*}%
and we denote $\tilde{\psi}_{N}(t)$ the solution of the $N$-body Schr\"{o}%
dinger equation \ref{equation:N-Body Schrodinger with switchable trap}
subject to this regularized initial data and $\left\{ \tilde{\gamma}%
_{N}^{(k)}(t)\right\} _{k=1}^{\infty }$ the family of marginal densities
associated with $\tilde{\psi}_{N}(t)$.

With these notations, if $\kappa >0$ small enough, on the one hand, we have%
\begin{equation*}
\left\langle \tilde{\psi}_{N}\left( 0\right) ,H_{N}^{k}\tilde{\psi}%
_{N}\left( 0\right) \right\rangle \leqslant \tilde{C}^{k}N^{k},
\end{equation*}%
and%
\begin{equation*}
\lim_{N\rightarrow \infty }\limfunc{Tr}J^{(k)}\left( \tilde{\gamma}%
_{N}^{(k)}(0)-\dprod_{j=1}^{k}\phi _{0}\overline{\phi }_{0}\right) =0,
\end{equation*}%
for every $J^{(k)}\in \mathcal{K}_{k}.$ Thus $\tilde{\gamma}_{N}^{(k)}(0,%
\mathbf{x}_{k};\mathbf{x}_{k}^{\prime })\rightarrow \dprod_{j=1}^{k}\phi _{0}%
\overline{\phi }_{0}$ as $N\rightarrow \infty $ in the weak* topology. Since 
$\dprod_{j=1}^{k}\phi _{0}\overline{\phi }_{0}$ is an orthogonal projection,
the convergence in the weak* topology is equivalent to the convergence in
the trace norm. Consequently, the conditions of Theorem \ref%
{Theorem:SmoothVersionofTheMainTheorem} are verified and it implies that $%
\forall t\in \lbrack 0,T_{0}]$ and $\forall k\geqslant 1$,%
\begin{equation*}
\lim_{N\rightarrow \infty }\limfunc{Tr}\left\vert \tilde{\gamma}_{N}^{(k)}(t,%
\mathbf{x}_{k};\mathbf{x}_{k}^{\prime })-\dprod_{j=1}^{k}\phi (t,x_{j})%
\overline{\phi }(t,x_{j}^{\prime })\right\vert =0.
\end{equation*}%
On the other hand, there is a constant $C$ independent of $N$ and $\kappa $
such that 
\begin{equation*}
\left\vert \limfunc{Tr}J^{(k)}\left( \tilde{\gamma}_{N}^{(k)}(t)-\gamma
_{N}^{(k)}(t)\right) \right\vert \leqslant C\left\Vert J^{(k)}\right\Vert
\kappa ^{\frac{1}{2}}
\end{equation*}%
for every $J^{(k)}\in \mathcal{K}_{k}$. Therefore,%
\begin{eqnarray*}
&&\left\vert \limfunc{Tr}J^{(k)}\left( \gamma _{N}^{(k)}(t,\mathbf{x}_{k};%
\mathbf{x}_{k}^{\prime })-\dprod_{j=1}^{k}\phi (t,x_{j})\overline{\phi }%
(t,x_{j}^{\prime })\right) \right\vert \\
&\leqslant &\left\vert \limfunc{Tr}J^{(k)}\left( \gamma _{N}^{(k)}(t)-\tilde{%
\gamma}_{N}^{(k)}(t)\right) \right\vert +\left\vert \limfunc{Tr}%
J^{(k)}\left( \tilde{\gamma}_{N}^{(k)}(t)-\dprod_{j=1}^{k}\phi (t,x_{j})%
\overline{\phi }(t,x_{j}^{\prime })\right) \right\vert \\
&\leqslant &C\left\Vert J^{(k)}\right\Vert \kappa ^{\frac{1}{2}}+\left\vert 
\limfunc{Tr}J^{(k)}\left( \tilde{\gamma}_{N}^{(k)}(t)-\dprod_{j=1}^{k}\phi
(t,x_{j})\overline{\phi }(t,x_{j}^{\prime })\right) \right\vert .
\end{eqnarray*}%
As the above inequality holds for all $\kappa >0$ small enough, we know $%
\gamma _{N}^{(k)}(t,\mathbf{x}_{k};\mathbf{x}_{k}^{\prime })\rightarrow
\dprod_{j=1}^{k}\phi (t,x_{j})\overline{\phi }(t,x_{j}^{\prime })$ as $%
N\rightarrow \infty $ in the weak* topology. This convergence is again
equivalent to the convergence in the trace norm because $\dprod_{j=1}^{k}%
\phi (t,x_{j})\overline{\phi }(t,x_{j}^{\prime })$ is an orthogonal
projection as well. Whence we have established our main theorem (Theorem \ref%
{Theorem:3D BEC}) through Theorem \ref{Theorem:SmoothVersionofTheMainTheorem}%
. It remains to show Theorem \ref{Theorem:SmoothVersionofTheMainTheorem}. We
prove Theorem \ref{Theorem:SmoothVersionofTheMainTheorem} with the help of
the lens transform. We define the lens transform and show its related
properties in the next section then we establish Theorem \ref%
{Theorem:SmoothVersionofTheMainTheorem} in Section \ref%
{Section:ProofOfSmooth3DBEC}.

\section{Lens Transform\label{Section:LensTrasnform}}

In this section, we first define the lens transform and review its relevant
properties, then we prove an energy estimate (Proposition \ref%
{Proposition:GettingTheE-S-YEnergyBound}) which relates the energy on the
two sides of the lens transform. It aids in the proof of Theorem \ref%
{Theorem:SmoothVersionofTheMainTheorem} in the sense that it links the
analysis of $-\triangle _{x}+\omega ^{2}\left\vert x\right\vert ^{2}$ to the
analysis of $-\triangle _{y}$ which is a better understood operator. We
denote $(t,x)$ the space-time on the Hermite side and $(\tau ,y)$ the
space-time on the Laplacian side. We now define the lens transform we need.

\begin{definition}
\label{def:functionG-Lens}Let $\mathbf{x}_{N},\mathbf{y}_{N}\in \mathbb{R}%
^{3N}$. We define the lens transform for $L^{2}$ functions $M_{N}:L^{2}(d%
\mathbf{y}_{N})\rightarrow $ $L^{2}(d\mathbf{x}_{N})$ and its inverse by%
\begin{eqnarray*}
\left( M_{N}u_{N}\right) (t,\mathbf{x}_{N}) &=&\frac{e^{-i\omega \tan \omega
t\frac{\left\vert \mathbf{x}_{N}\right\vert ^{2}}{2}}}{(\cos \omega t)^{%
\frac{3N}{2}}}u_{N}(\frac{\tan \omega t}{\omega },\frac{\mathbf{x}_{N}}{\cos
\omega t}) \\
\left( M_{N}^{-1}\psi _{N}\right) (\tau ,\mathbf{y}_{N}) &=&\frac{e^{i\frac{%
\omega ^{2}\tau }{1+\omega ^{2}\tau ^{2}}\frac{\left\vert \mathbf{y}%
_{N}\right\vert ^{2}}{2}}}{\left( 1+\omega ^{2}\tau ^{2}\right) ^{\frac{3N}{4%
}}}\psi _{N}(\frac{\arctan \left( \omega \tau \right) }{\omega },\frac{%
\mathbf{y}_{N}}{\sqrt{1+\omega ^{2}\tau ^{2}}}).
\end{eqnarray*}%
$M_{N}$ is unitary by definition and the variables are related by%
\begin{equation*}
\tau =\frac{\tan \omega t}{\omega },\text{ }\mathbf{y}_{N}=\frac{\mathbf{x}%
_{k}}{\cos \omega t}.
\end{equation*}
\end{definition}

\begin{definition}
\label{def:KernelG-Lens}Let $\mathbf{x}_{k},\mathbf{x}_{k}^{\prime },\mathbf{%
y}_{k},\mathbf{y}_{k}^{\prime }\in \mathbb{R}^{3k}$. We define the lens
transform for Hilbert-Schmidt kernels $T_{k}:$ $L^{2}(d\mathbf{y}_{k}d%
\mathbf{y}_{k}^{\prime })\rightarrow $ $L^{2}(d\mathbf{x}_{k}d\mathbf{x}%
_{k}^{\prime })$ and its inverse by%
\begin{eqnarray}
\left( T_{k}u^{(k)}\right) (t,\mathbf{x}_{k};\mathbf{x}_{k}^{\prime }) &=&%
\frac{e^{-i\omega \tan \omega t\frac{\left( \left\vert \mathbf{x}%
_{k}\right\vert ^{2}-\left\vert \mathbf{x}_{k}^{\prime }\right\vert
^{2}\right) }{2}}}{(\cos \omega t)^{3k}}u^{(k)}(\frac{\tan \omega t}{\omega }%
,\frac{\mathbf{x}_{k}}{\cos \omega t};\frac{\mathbf{x}_{k}^{\prime }}{\cos
\omega t})  \notag \\
\left( T_{k}^{-1}\gamma ^{(k)}\right) (\tau ,\mathbf{y}_{k};\mathbf{y}%
_{k}^{\prime }) &=&\frac{e^{i\frac{\omega ^{2}\tau }{1+\omega ^{2}\tau ^{2}}%
\frac{\left( \left\vert \mathbf{y}_{k}\right\vert ^{2}-\left\vert \mathbf{y}%
_{k}^{\prime }\right\vert ^{2}\right) }{2}}}{\left( 1+\omega ^{2}\tau
^{2}\right) ^{\frac{3k}{2}}}\gamma ^{(k)}(\frac{\arctan \left( \omega \tau
\right) }{\omega },\frac{\mathbf{y}_{k}}{\sqrt{1+\omega ^{2}\tau ^{2}}};%
\frac{\mathbf{y}_{k}^{\prime }}{\sqrt{1+\omega ^{2}\tau ^{2}}}).
\label{def:inverse lens}
\end{eqnarray}%
$T_{k}$ is unitary by definition as well and the variables are again related
by%
\begin{equation*}
\tau =\frac{\tan \omega t}{\omega },\text{ }\mathbf{y}_{k}=\frac{\mathbf{x}%
_{k}}{\cos \omega t}\text{ and }\mathbf{y}_{k}^{\prime }=\frac{\mathbf{x}%
_{k}^{\prime }}{\cos \omega t}.
\end{equation*}
\end{definition}

Before we characterize the exact effect of the lens transform, we clarify
the motivation of such definitions by a lemma.

\begin{lemma}
\cite{Carles, ChenAnisotropic} Define $\alpha $ and $\beta $ via the system%
\begin{eqnarray*}
\ddot{\alpha}(t)+\eta (t)\alpha (t) &=&0,\alpha (0)=0,\dot{\alpha}(0)=1, \\
\ddot{\beta}(t)+\eta (t)\beta (t) &=&0,\beta (0)=1,\dot{\beta}(0)=0.
\end{eqnarray*}%
If $\beta $ is nonzero in the time interval $[0,T]$, then the solution of
the 1D Schr\"{o}dinger equation with a time dependent quadratic trap%
\begin{eqnarray*}
i\partial _{t}\psi &=&\left( -\frac{1}{2}\partial _{x}^{2}+\frac{1}{2}\eta
(t)x^{2}\right) \psi \text{ in }\mathbb{R}^{1+1} \\
\psi (0,x) &=&f(x)\in L^{2}(\mathbb{R})
\end{eqnarray*}%
in $[0,T]$ is given by%
\begin{equation*}
\psi (t,x)=\frac{e^{i\frac{\dot{\beta}(t)}{\beta (t)}\frac{x^{2}}{2}}}{%
\left( \beta (t)\right) ^{\frac{1}{2}}}u(\frac{\alpha (t)}{\beta (t)},\frac{x%
}{\beta (t)}),
\end{equation*}%
if $u(\tau ,y)$ solves the 1D free Schr\"{o}dinger equation%
\begin{equation*}
i\partial _{\tau }u=-\frac{1}{2}\partial _{y}^{2}u\text{ in }\mathbb{R}^{1+1}
\end{equation*}%
subject to the same initial data.
\end{lemma}

\begin{proof}
See \cite{Carles} for a proof by direct computation and \cite%
{ChenAnisotropic} for an algebraic proof using the metaplectic
representation. When $\eta (t)=\omega ^{2}$, such a transformation has a
long history, we refer the readers to \cite{Carles} and the references
within.
\end{proof}

To make formulas shorter, let us write formula \ref{def:inverse lens} as%
\begin{equation*}
\left( T_{k}^{-1}\gamma ^{(k)}\right) (\tau ,\mathbf{y}_{k};\mathbf{y}%
_{k}^{\prime })=\gamma ^{(k)}(t,\mathbf{x}_{k};\mathbf{x}_{k}^{\prime
})h^{(k)}\left( \tau ,\mathbf{y}_{k};\mathbf{y}_{k}^{\prime }\right) 
\end{equation*}%
where $h^{(k)}\left( \tau ,\mathbf{y}_{k};\mathbf{y}_{k}^{\prime }\right) $
represents the factor away from $\gamma ^{(k)}$ in formula \ref{def:inverse
lens}, then precisely, the lens transform has the following effect.

\begin{proposition}
\label{Proposition:ChangeVariableEffectOfGLens}%
\begin{eqnarray*}
&&\left( i\partial _{\tau }+\frac{1}{2}\triangle _{\mathbf{y}_{k}}-\frac{1}{2%
}\triangle _{\mathbf{y}_{k}^{\prime }}\right) \left( T_{k}^{-1}\gamma
^{(k)}\right) (\tau ,\mathbf{y}_{k};\mathbf{y}_{k}^{\prime }) \\
&=&\frac{h^{(k)}\left( \tau ,\mathbf{y}_{k};\mathbf{y}_{k}^{\prime }\right) 
}{\left( 1+\omega ^{2}\tau ^{2}\right) }\left( i\partial _{t}-\left( -\frac{1%
}{2}\triangle _{\mathbf{x}_{k}}+\omega ^{2}\frac{\left\vert \mathbf{x}%
_{k}\right\vert ^{2}}{2}\right) +\left( -\frac{1}{2}\triangle _{\mathbf{x}%
_{k}^{\prime }}+\omega ^{2}\frac{\left\vert \mathbf{x}_{k}^{\prime
}\right\vert ^{2}}{2}\right) \right) \left( \gamma ^{(k)}\right) (t,\mathbf{x%
}_{k};\mathbf{x}_{k}^{\prime })
\end{eqnarray*}
\end{proposition}

\begin{proof}
This is a direct computation.
\end{proof}

Via Proposition \ref{Proposition:ChangeVariableEffectOfGLens}, we know how
the lens transform acts on the Bogoliubov--Born--Green--Kirkwood--Yvon
(BBGKY) hierarchy and the Gross-Pitaevskii hierarchy.

\begin{lemma}
\label{Lemma:BBGKY hierarchy under GLens}$\left( \text{BBGKY hierarchy under
the lens transform}\right) $ Write $V_{N}\left( x\right) =N^{3\beta }V\left(
N^{\beta }x\right) .$ $\left\{ \gamma _{N}^{(k)}\right\} $ solves the 3D
BBGKY hierarchy with a quadratic trap 
\begin{eqnarray}
&&i\partial _{t}\gamma _{N}^{(k)}-\left( -\frac{1}{2}\triangle _{\mathbf{x}%
_{k}}+\omega ^{2}\frac{\left\vert \mathbf{x}_{k}\right\vert ^{2}}{2}\right)
\gamma _{N}^{(k)}+\left( -\frac{1}{2}\triangle _{\mathbf{x}_{k}^{\prime
}}+\omega ^{2}\frac{\left\vert \mathbf{x}_{k}^{\prime }\right\vert ^{2}}{2}%
\right) \gamma _{N}^{(k)}  \label{hierarchy:TheBBGKYHierarchyWithTraps} \\
&=&\frac{1}{N}\sum_{1\leqslant i<j\leqslant k}\left(
V_{N}(x_{i}-x_{j})-V_{N}(x_{i}^{\prime }-x_{j}^{\prime })\right) \gamma
_{N}^{(k)}  \notag \\
&&+\frac{N-k}{N}\sum_{j=1}^{k}\int \left(
V_{N}(x_{j}-x_{k+1})-V_{N}(x_{j}^{\prime }-x_{k+1})\right) \gamma
_{N}^{(k+1)}\left( t,\mathbf{x}_{k},x_{k+1};\mathbf{x}_{k}^{\prime
},x_{k+1}\right) dx_{k+1},  \notag
\end{eqnarray}%
in $\left[ -T_{0},T_{0}\right] $ if and only if $\left\{
u_{N}^{(k)}=T_{k}^{-1}\gamma _{N}^{(k)}\right\} $ solves the hierarchy%
\begin{eqnarray}
&&\left( i\partial _{\tau }+\frac{1}{2}\triangle _{\mathbf{y}_{k}}-\frac{1}{2%
}\triangle _{\mathbf{y}_{k}^{\prime }}\right) u_{N}^{(k)}
\label{hierarchy:GlensBBGKY} \\
&=&\frac{1}{\left( 1+\omega ^{2}\tau ^{2}\right) }\frac{1}{N}%
\sum_{1\leqslant i<j\leqslant k}\left( V_{N}(\frac{y_{i}-y_{j}}{\left(
1+\omega ^{2}\tau ^{2}\right) ^{\frac{1}{2}}})-V_{N}(\frac{y_{i}^{\prime
}-y_{j}^{\prime }}{\left( 1+\omega ^{2}\tau ^{2}\right) ^{\frac{1}{2}}}%
)\right) u_{N}^{(k)}  \notag \\
&&+\frac{N-k}{N}\frac{1}{\left( 1+\omega ^{2}\tau ^{2}\right) }%
\sum_{j=1}^{k}\int \left( V_{N}(\frac{y_{j}-y_{k+1}}{\left( 1+\omega
^{2}\tau ^{2}\right) ^{\frac{1}{2}}})-V_{N}(\frac{y_{j}^{\prime }-y_{k+1}}{%
\left( 1+\omega ^{2}\tau ^{2}\right) ^{\frac{1}{2}}})\right)   \notag \\
&&\times u_{N}^{(k+1)}\left( \tau ,\mathbf{y}_{k},y_{k+1};\mathbf{y}%
_{k}^{\prime },y_{k+1}\right) dy_{k+1},  \notag
\end{eqnarray}%
in $\left[ -\frac{\tan \omega T_{0}}{\omega },\frac{\tan \omega T_{0}}{%
\omega }\right] .$
\end{lemma}

\begin{lemma}
\label{Lemma:GP hierarchy under GLens}$\left( \text{Gross-Pitaevskii
hierarchy under the lens transform}\right) $ $\left\{ \gamma ^{(k)}\right\} $
solves the 3D Gross-Pitaevskii hierarchy with a quadratic trap 
\begin{equation}
i\partial _{t}\gamma ^{(k)}-\left( -\frac{1}{2}\triangle _{\mathbf{x}%
_{k}}+\omega ^{2}\frac{\left\vert \mathbf{x}_{k}\right\vert ^{2}}{2}\right)
\gamma ^{(k)}+\left( -\frac{1}{2}\triangle _{\mathbf{x}_{k}^{\prime
}}+\omega ^{2}\frac{\left\vert \mathbf{x}_{k}^{\prime }\right\vert ^{2}}{2}%
\right) \gamma ^{(k)}=b_{0}\sum_{j=1}^{k}B_{j,k+1}\gamma ^{(k+1)},
\label{hierarchy:TheGPHierarchyWithTraps}
\end{equation}%
in $\left[ -T_{0},T_{0}\right] $ if and only if $\left\{
u^{(k)}=T_{k}^{-1}\gamma ^{(k)}\right\} $ solves the hierarchy%
\begin{equation}
\left( i\partial _{\tau }+\frac{1}{2}\triangle _{\mathbf{y}_{k}}-\frac{1}{2}%
\triangle _{\mathbf{y}_{k}^{\prime }}\right) u^{(k)}=\left( 1+\omega
^{2}\tau ^{2}\right) ^{\frac{1}{2}}b_{0}\sum_{j=1}^{k}B_{j,k+1}u^{(k+1)},
\label{hierarchy:GLensGP}
\end{equation}%
in $\left[ -\frac{\tan \omega T_{0}}{\omega },\frac{\tan \omega T_{0}}{%
\omega }\right] .$
\end{lemma}

The lens transform for the Hilbert-Schmidt kernels is not only by definition
a unitary transform on $L^{2}\left( \mathbb{R}^{6k}\right) $, it is also an
isometry on the space of the self-adjoint trace class operator kernels.

\begin{lemma}
\label{Lemma:TraceNormPreservation}$\forall t\in \lbrack -T_{0},T_{0}]$, $%
\forall K(\mathbf{y}_{k},\mathbf{y}_{k}^{\prime })$ the kernel of a
self-adjoint trace class operator on $L^{2}\left( \mathbb{R}^{3k}\right) $.
If 
\begin{equation*}
\int K(\mathbf{y}_{k},\mathbf{y}_{k}^{\prime })f(\mathbf{y}_{k}^{\prime })d%
\mathbf{y}_{k}^{\prime }=\lambda f(\mathbf{y}_{k}),
\end{equation*}%
then%
\begin{equation*}
\int \left( T_{k}K\right) (\mathbf{x}_{k},\mathbf{x}_{k}^{\prime })\left(
M_{N}f\right) \left( \mathbf{x}_{k}^{\prime }\right) d\mathbf{x}_{k}^{\prime
}=\lambda \left( M_{N}f\right) \left( \mathbf{x}_{k}\right) .
\end{equation*}%
In other words, the eigenvectors of the kernel $\left( T_{k}K\right) (%
\mathbf{x}_{k},\mathbf{x}_{k}^{\prime })$ are exactly the lens transform
(Lemma \ref{def:functionG-Lens}) of the eigenvectors of the kernel $K(%
\mathbf{y}_{k},\mathbf{y}_{k}^{\prime })$ with the same eigenvalues. In
particular, we have%
\begin{equation*}
\limfunc{Tr}\left\vert T_{k}K\right\vert =\limfunc{Tr}\left\vert
K\right\vert .
\end{equation*}
\end{lemma}

\begin{proof}
This is a straight forward computation. We remark that we have defined the
generalized lens transform for a function and a kernel separately via
Definitions \ref{def:functionG-Lens} and \ref{def:KernelG-Lens}.
\end{proof}

Once we have proved the following proposition which relates the energy of
the two sides of the lens transform, we can start the proof of Theorem \ref%
{Theorem:SmoothVersionofTheMainTheorem}.

\begin{proposition}
\label{Proposition:GettingTheE-S-YEnergyBound}Let $\psi _{N}(t,\mathbf{x}%
_{N})$ be the solution to equation \ref{equation:N-Body Schrodinger with
switchable trap} for some $\beta \in \left( 0,3/5\right) $ subject to
initial $\psi _{N}(0)$ which satisfies the energy condition%
\begin{equation*}
\left\langle \psi _{N}(0),H_{N}^{k}\psi _{N}(0)\right\rangle \leqslant
C^{k}N^{k}.
\end{equation*}%
If $u_{N}(\tau ,\mathbf{y}_{N})=M_{N}^{-1}\psi _{N}$, then there is a $%
C\geqslant 0$, for all $k\geqslant 0$, $\exists N_{0}\left( k\right) $ such
that 
\begin{equation*}
\left\langle u_{N}(\tau ),\dprod\limits_{j=1}^{k}\left( 1-\triangle
_{y_{j}}\right) u_{N}(\tau )\right\rangle \leqslant C^{k},
\end{equation*}%
for all $N\geqslant N_{0}$ and all $\tau \in \left[ -\frac{\tan \omega T_{0}%
}{\omega },\frac{\tan \omega T_{0}}{\omega }\right] $ provided that $T_{0}<%
\frac{\pi }{2\omega }.$
\end{proposition}

The rest of this section is the proof of Proposition \ref%
{Proposition:GettingTheE-S-YEnergyBound}. We prove it for $\omega >0$
through Lemmas \ref{Lemma:ESYLemma} and \ref{Lemma:momentum between two
sides of the lens transform} since the case $\omega =0$ has already been
studied in \cite{E-E-S-Y1, E-S-Y2}.

\begin{lemma}
\label{Lemma:ESYLemma}For $\beta \in \left( 0,3/5\right) ,$ there is a $%
C\geqslant 0$, for all $k\geqslant 0$, $\exists N_{0}\left( k\right) $ such
that 
\begin{equation*}
\left\langle \varphi ,H_{N}^{k}\varphi \right\rangle \geqslant
C^{k}N^{k}\left\langle \varphi ,\dprod\limits_{j=1}^{k}\left( -\triangle
_{x_{j}}+\omega ^{2}\left\vert x_{j}\right\vert ^{2}\right) \varphi
\right\rangle ,
\end{equation*}%
for all $N\geqslant N_{0}$ and all $\varphi \in L_{s}^{2}\left( \mathbb{R}%
^{3N}\right) .$
\end{lemma}

\begin{proof}
The proof basically follows Proposition 1 in \cite{E-E-S-Y1} step by step if
one replaces $\left( 1-\triangle _{x_{j}}\right) $ by $\left( -\triangle
_{x_{j}}+\omega ^{2}\left\vert x_{j}\right\vert ^{2}\right) $ and notices
that for $\omega >0$,%
\begin{equation*}
\left\Vert \left( 1-\triangle \right) ^{\frac{\alpha }{2}}f\right\Vert
_{L^{2}}\leqslant C_{\alpha }\left\Vert \left( -\triangle +\omega
^{2}\left\vert x\right\vert ^{2}\right) ^{\frac{\alpha }{2}}f\right\Vert
_{L^{2}}\text{ (\cite{Thangavelu})}
\end{equation*}%
when one uses Sobolev. There are some extra error terms which can be easily
handled. We illustrate the control of the extra error terms through the
following example. Write $S_{i}^{2}=\left( -\triangle _{x_{i}}+\omega
^{2}\left\vert x_{i}\right\vert ^{2}\right) $, consider%
\begin{eqnarray*}
&&N^{3\beta }\left( \left\langle \varphi ,S_{1}^{2}...S_{n+1}^{2}V\left(
N^{\beta }\left( x_{1}-x_{2}\right) \right) \varphi \right\rangle
+c.c.\right) \\
&=&N^{3\beta }\left( \left\langle S_{3}...S_{n+1}\varphi
,S_{1}^{2}S_{2}^{2}V\left( N^{\beta }\left( x_{1}-x_{2}\right) \right)
S_{3}...S_{n+1}\varphi \right\rangle +c.c.\right) ,
\end{eqnarray*}%
where c.c. denotes complex conjugates. Neglecting $S_{3}...S_{n+1}$, we have%
\begin{eqnarray*}
&&N^{3\beta }\left( \left\langle \varphi ,S_{1}^{2}S_{2}^{2}V\left( N^{\beta
}\left( x_{1}-x_{2}\right) \right) \varphi \right\rangle +c.c.\right) \\
&=&N^{3\beta }\left( \left\langle \varphi ,\left( -\triangle _{x_{1}}+\omega
^{2}\left\vert x_{1}\right\vert ^{2}\right) \left( -\triangle
_{x_{2}}+\omega ^{2}\left\vert x_{2}\right\vert ^{2}\right) V\left( N^{\beta
}\left( x_{1}-x_{2}\right) \right) \varphi \right\rangle +c.c.\right) \\
&=&N^{3\beta }\left( \left\langle \varphi ,\left( -\triangle _{x_{1}}\right)
\left( -\triangle _{x_{2}}\right) V\left( N^{\beta }\left(
x_{1}-x_{2}\right) \right) \varphi \right\rangle +c.c.\right) \\
&&+N^{3\beta }\left( \left\langle \varphi ,\left( -\triangle _{x_{1}}\right)
\left( \omega ^{2}\left\vert x_{2}\right\vert ^{2}\right) V\left( N^{\beta
}\left( x_{1}-x_{2}\right) \right) \varphi \right\rangle +c.c.\right) \\
&&+N^{3\beta }\left( \left\langle \varphi ,\left( \omega ^{2}\left\vert
x_{1}\right\vert ^{2}\right) \left( -\triangle _{x_{2}}\right) V\left(
N^{\beta }\left( x_{1}-x_{2}\right) \right) \varphi \right\rangle
+c.c.\right) \\
&&+N^{3\beta }\left( \left\langle \varphi ,\left( \omega ^{2}\left\vert
x_{1}\right\vert ^{2}\right) \left( \omega ^{2}\left\vert x_{2}\right\vert
^{2}\right) V\left( N^{\beta }\left( x_{1}-x_{2}\right) \right) \varphi
\right\rangle +c.c.\right) \\
&=&I+II+III+IV
\end{eqnarray*}%
Compared to \cite{E-E-S-Y1}, the extra error terms are $II,III,$ and $IV.$
Since $IV$ is positive, we only look at%
\begin{eqnarray*}
II &\geqslant &-2N^{4\beta }\left\vert \left\langle \nabla _{x_{1}}\omega
\left\vert x_{2}\right\vert \varphi ,\left\vert \left( \nabla V\right)
\left( N^{\beta }\left( x_{1}-x_{2}\right) \right) \right\vert \omega
\left\vert x_{2}\right\vert \varphi \right\rangle \right\vert \\
&\geqslant &-CN^{4\beta }\alpha \left\langle \nabla _{x_{1}}\omega
\left\vert x_{2}\right\vert \varphi ,\left\vert \left( \nabla V\right)
\left( N^{\beta }\left( x_{1}-x_{2}\right) \right) \right\vert \nabla
_{x_{1}}\omega \left\vert x_{2}\right\vert \varphi \right\rangle \\
&&-CN^{4\beta }\alpha ^{-1}\left\langle \omega \left\vert x_{2}\right\vert
\varphi ,\left\vert \left( \nabla V\right) \left( N^{\beta }\left(
x_{1}-x_{2}\right) \right) \right\vert \omega \left\vert x_{2}\right\vert
\varphi \right\rangle \\
&\geqslant &-CN^{4\beta }\left( \alpha \left\langle \varphi
,S_{1}^{2}S_{2}^{2}\varphi \right\rangle +\alpha ^{-1}N^{-2\beta
}\left\langle \omega \left\vert x_{2}\right\vert \varphi ,S_{1}^{2}\omega
\left\vert x_{2}\right\vert \varphi \right\rangle \right) \text{ (Sobolev at
the second term)} \\
&\geqslant &-CN^{3\beta }\left\langle \varphi ,S_{1}^{2}S_{2}^{2}\varphi
\right\rangle \text{ (}\alpha =N^{-\beta }\text{)} \\
&\geqslant &-CN^{3\beta -2}N^{2}\left\langle \varphi
,S_{1}^{2}S_{2}^{2}\varphi \right\rangle
\end{eqnarray*}%
As long as $\beta \in \left( 0,\frac{2}{3}\right) ,$ we can absorb the extra
error terms into the main term $N^{2}\left\langle \varphi
,S_{1}^{2}S_{2}^{2}\varphi \right\rangle $. The Sobolev we used is%
\begin{eqnarray*}
\int V\left( N^{\beta }\left( x_{1}-x_{2}\right) \right) \left\vert \varphi
\right\vert ^{2}dx_{1}dx_{2} &\leqslant &\left\Vert V\left( N^{\beta }\cdot
\right) \right\Vert _{L^{\frac{3}{2}}}\int \left( \int \left\vert \varphi
\right\vert ^{6}dx_{1}\right) ^{\frac{1}{3}}dx_{2} \\
&\leqslant &CN^{-2\beta }\left\Vert V\right\Vert _{L^{\frac{3}{2}}}\int
\left\vert \left( I-\triangle _{x_{1}}\right) ^{\frac{1}{2}}\varphi
\right\vert ^{2}dx_{1}dx_{2} \\
&\leqslant &CN^{-2\beta }\left\Vert V\right\Vert _{L^{\frac{3}{2}}}\int
\left\vert \left( -\triangle _{x_{1}}+\omega ^{2}\left\vert x_{1}\right\vert
^{2}\right) ^{\frac{1}{2}}\varphi \right\vert ^{2}dx_{1}dx_{2}.
\end{eqnarray*}
\end{proof}

\begin{lemma}
\label{Lemma:momentum between two sides of the lens transform}Let 
\begin{equation*}
P_{x}\left( t\right) =i\nabla _{x}\cos \omega t-\omega x\sin \omega t.
\end{equation*}%
If $u_{N}(\tau ,\mathbf{y}_{N})=M_{N}^{-1}\left( \psi _{N}\right) $, then%
\begin{eqnarray*}
\left\langle u_{N}(\tau ),\left( -\triangle _{y_{j}}\right) u_{N}(\tau
)\right\rangle &=&\left\langle P_{x_{j}}\left( t\right) \psi
_{N}(t),P_{x_{j}}\left( t\right) \psi _{N}(t)\right\rangle \\
&=&\left\langle \psi _{N}(t),P_{x_{j}}^{2}\left( t\right) \psi
_{N}(t)\right\rangle .
\end{eqnarray*}
\end{lemma}

\begin{proof}
We provide a proof through direct computation here. We remark that $%
P_{x}\left( t\right) $ is in fact the evolution of momentum. See \cite%
{ChenAnisotropic}. Without lose of generality, we may assume $N=1$, then%
\begin{eqnarray*}
P_{x}\left( t\right) \psi _{1}(t) &=&P_{x}\left( t\right) \left( \frac{e^{-i%
\frac{\omega \tan \omega t}{2}\left\vert x\right\vert ^{2}}}{\left( \cos
\omega t\right) ^{\frac{3}{2}}}u_{1}\left( \frac{\tan \omega t}{\omega },%
\frac{x}{\left( \cos \omega t\right) }\right) \right) \\
&=&\frac{e^{-i\frac{\omega \tan \omega t}{2}\left\vert x\right\vert ^{2}}}{%
\left( \cos \omega t\right) ^{\frac{3}{2}}}\left( i\nabla _{x}\cos \omega
t\right) u_{1}\left( \frac{\tan \omega t}{\omega },\frac{x}{\left( \cos
\omega t\right) }\right) \\
&=&\frac{e^{-i\frac{\omega \tan \omega t}{2}\left\vert x\right\vert ^{2}}}{%
\left( \cos \omega t\right) ^{\frac{3}{2}}}\left( i\nabla _{y}\right)
u_{1}\left( \tau ,y\right) .
\end{eqnarray*}%
Thus%
\begin{equation*}
\left\langle u_{1}(\tau ),\left( -\triangle _{y}\right) u_{1}(\tau
)\right\rangle =\int \left\vert \nabla _{y}u_{1}\left( \tau ,y\right)
\right\vert ^{2}dy=\int \frac{\left\vert \nabla _{y}u_{1}\left( \tau
,y\right) \right\vert ^{2}}{\left\vert \cos \omega t\right\vert ^{3}}%
dx=\left\langle P_{x}\left( t\right) \psi _{1}(t),P_{x}\left( t\right) \psi
_{1}(t)\right\rangle .
\end{equation*}
\end{proof}

\begin{proof}[Proof of Proposition \protect\ref%
{Proposition:GettingTheE-S-YEnergyBound}]
We first notice that%
\begin{equation*}
\left\langle u_{N}(\tau ),\dprod\limits_{j=1}^{k}\left( 1-\triangle
_{y_{j}}\right) u_{N}(\tau )\right\rangle =\left\langle \psi
_{N}(t),\dprod\limits_{j=1}^{k}\left( 1+P_{x_{j}}^{2}\left( t\right) \right)
\psi _{N}(t)\right\rangle \text{ }\left( \text{Lemma \ref{Lemma:momentum
between two sides of the lens transform}}\right) .
\end{equation*}%
Since%
\begin{eqnarray*}
&&\left\langle f,\left( 1+P_{x_{j}}^{2}\left( t\right) \right) f\right\rangle
\\
&=&\left\langle f,f\right\rangle +\cos ^{2}\omega t\left\langle \nabla
f,\nabla f\right\rangle +\omega ^{2}\sin ^{2}\omega t\left\langle
xf,xf\right\rangle -2\omega \sin \omega t\cos \omega t\func{Im}\left\langle
\nabla f,xf\right\rangle \\
&\leqslant &C\left\langle f,\left( -\triangle _{x}+\omega ^{2}\left\vert
x\right\vert ^{2}\right) f\right\rangle ,
\end{eqnarray*}%
we know%
\begin{eqnarray*}
\left\langle \psi _{N}(t),\dprod\limits_{j=1}^{k}\left(
1+P_{x_{j}}^{2}\left( t\right) \right) \psi _{N}(t)\right\rangle &\leqslant
&C^{k}\left\langle \psi _{N}(t),\dprod\limits_{j=1}^{k}\left( -\triangle
_{x_{j}}+\omega ^{2}\left\vert x_{j}\right\vert ^{2}\right) \psi
_{N}(t)\right\rangle \\
&\leqslant &\frac{C^{k}}{N^{k}}\left\langle \psi _{N}(t),H_{N}^{k}\psi
_{N}(t)\right\rangle .\text{ (Lemma \ref{Lemma:ESYLemma})}
\end{eqnarray*}%
From the energy condition on the initial datum $\psi _{N}(0)$, we then
deduce Proposition \ref{Proposition:GettingTheE-S-YEnergyBound} that is%
\begin{equation*}
\left\langle u_{N}(\tau ),\dprod\limits_{j=1}^{k}\left( 1-\triangle
_{y_{j}}\right) u_{N}(\tau )\right\rangle \leqslant C^{k}.
\end{equation*}
\end{proof}

\section{Proof of Theorem \protect\ref{Theorem:SmoothVersionofTheMainTheorem}
\label{Section:ProofOfSmooth3DBEC}}

We devote this section to establishing Theorem \ref%
{Theorem:SmoothVersionofTheMainTheorem}. The main idea is to first prove
that, in the time period $\tau \in \left[ 0,\frac{\tan \omega T_{0}}{\omega }%
\right] $ with $T_{0}<\frac{\pi }{2\omega }$, as $N\rightarrow \infty ,$ $%
\left\{ u_{N}^{(k)}=T_{k}^{-1}\gamma _{N}^{(k)}\right\} ,$ the lens
transform of the solution to the BBGKY hierarchy \ref%
{hierarchy:TheBBGKYHierarchyWithTraps}, converges to $\left\{
u^{(k)}=T_{k}^{-1}\gamma ^{(k)}\right\} $, the lens transform of the
solution to the Gross-Pitaevskii hierarchy \ref%
{hierarchy:TheGPHierarchyWithTraps} in the trace norm, then use Lemma \ref%
{Lemma:TraceNormPreservation} to conclude the convergence $\gamma
_{N}^{(k)}\rightarrow \gamma ^{(k)}$ in the trace norm as $N\rightarrow
\infty $ in the time period $t\in \left[ 0,T_{0}\right] $. Then the time
translation invariance of equation \ref{equation:N-Body Schrodinger with
switchable trap} proves Theorem \ref{Theorem:SmoothVersionofTheMainTheorem}
for all time.

\begin{itemize}
\item[Step I.] In Proposition \ref{Proposition:GettingTheE-S-YEnergyBound},
we have already established the energy estimate for $\left\{
u_{N}=M_{N}^{-1}\psi _{N}\right\} ,$%
\begin{equation*}
\sup_{\tau \in \left[ 0,\frac{\tan \omega T_{0}}{\omega }\right]
}\left\langle u_{N}(\tau ),\dprod\limits_{j=1}^{k}\left( 1-\triangle
_{y_{j}}\right) u_{N}(\tau )\right\rangle \leqslant C^{k}.
\end{equation*}%
which becomes 
\begin{equation}
\sup_{\tau \in \left[ 0,\frac{\tan \omega T_{0}}{\omega }\right] }\limfunc{Tr%
}\left( \dprod_{j=1}^{k}\left( 1-\triangle _{y_{j}}\right) \right)
u_{N}^{(k)}\leqslant C^{k},  \label{estimate:a-priori bound of BBGKY}
\end{equation}%
for $\left\{ u_{N}^{(k)}=T_{k}^{-1}\gamma _{N}^{(k)}\right\} $. Therefore we
can utilize the proof in Erd\"{o}s-Schlein-Yau \cite{E-S-Y2} or
Kirkpatrick-Schlein-Staffilani \cite{Kirpatrick} to show that the sequence $%
\left\{ u_{N}^{(k)}=T_{k}^{-1}\gamma _{N}^{(k)}\right\} $ is compact with
respect to the weak* topology on the trace class operators and every limit
point $\left\{ u^{(k)}\right\} $ solves hierarchy \ref{hierarchy:GLensGP}
which, we recall, is 
\begin{equation}
\left( i\partial _{\tau }+\frac{1}{2}\triangle _{\mathbf{y}_{k}}-\frac{1}{2}%
\triangle _{\mathbf{y}_{k}^{\prime }}\right) u^{(k)}=\sqrt{1+\omega ^{2}\tau
^{2}}b_{0}\sum_{j=1}^{k}B_{j,k+1}u^{(k+1)}.
\label{hierarchy:GPwithTimeFactor}
\end{equation}%
This is a fixed time argument. We omit the details here.
\end{itemize}

\begin{notation}
To make formulas shorter, let us write%
\begin{equation*}
g(\tau )=\left( 1+\omega ^{2}\tau ^{2}\right) ^{-\frac{1}{2}}
\end{equation*}%
from here on. The only property of $g(\tau )$ we are going to need is that $%
0<c\leqslant g(\tau )\leqslant C<\infty $ in any finite time period.
\end{notation}

\begin{itemize}
\item[Step II.] In this step, we use the a-priori estimate \ref%
{estimate:a-priori bound of BBGKY} to provide a space-time bound of $u^{(k)}$
so that we can employ Theorem \ref{Theorem:Uniqueness of GP} in Step III$.$
We transform estimate \ref{estimate:a-priori bound of BBGKY} into the
following theorem.
\end{itemize}

\begin{theorem}
\label{Theorem:Space-Time Bound for GlensBBGKY}(Main Auxiliary Theorem)
Assume $V\ $is a nonnegative $L^{1}(\mathbb{R}^{3})\cap H^{2}(\mathbb{R}%
^{3}) $ function. Let 
\begin{equation*}
\tilde{V}_{N,\tau }(y)=N^{3\beta }\tilde{V}_{\tau }(N^{\beta }y)=N^{3\beta
}g^{3}(\tau )V(g(\tau )N^{\beta }y)
\end{equation*}%
be the interaction potential with the interaction parameter $\beta \in
\left( 0,\frac{2}{7}\right] .$ Suppose that $u_{N}^{(k)}$ solves hierarchy %
\ref{hierarchy:GlensBBGKY} in $\left[ 0,T\right] \subset \left[ 0,\frac{\tan
\omega T_{0}}{\omega }\right] $, which, written in the integral form, is 
\begin{eqnarray}
u_{N}^{(k)} &=&U^{(k)}(\tau )u_{N,0}^{(k)}  \label{equation:Duhamel of BBGKY}
\\
&&-\frac{i}{N}\sum_{1\leqslant i<j\leqslant k}\int_{0}^{\tau }U^{(k)}(\tau
-s)\frac{\left( \tilde{V}_{N,s}(y_{i}-y_{j})-\tilde{V}_{N,s}(y_{i}^{\prime
}-y_{j}^{\prime })\right) }{g(s)}u_{N}^{(k)}(s,\mathbf{y}_{k};\mathbf{y}%
_{k}^{\prime })ds  \notag \\
&&-i\frac{N-k}{N}\sum_{j=1}^{k}\int_{0}^{\tau }U^{(k)}(\tau -s)\left( \frac{%
\tilde{B}_{N,j,k+1,s}}{g(s)}u_{N}^{(k+1)}\right) ds,  \notag
\end{eqnarray}%
subject to the condition that%
\begin{equation}
\sup_{\tau \in \left[ 0,T\right] }\limfunc{Tr}\left( \dprod_{j=1}^{k}\left(
1-\triangle _{y_{j}}\right) \right) u_{N}^{(k)}\leqslant C^{k},
\label{condition:EnergyBoundForBBGKY}
\end{equation}%
where $\tilde{B}_{N,j,k+1,\tau }=\tilde{B}_{N,j,k+1,\tau }^{1}-\tilde{B}%
_{N,j,k+1,\tau }^{2}$ with 
\begin{eqnarray*}
&&\left( \tilde{B}_{N,j,k+1,\tau }^{1}u_{N}^{(k+1)}\right) (\tau ,\mathbf{y}%
_{k};\mathbf{y}_{k}^{\prime })=\int \tilde{V}_{N,\tau
}(y_{j}-y_{k+1})u_{N}^{(k+1)}(\tau ,\mathbf{y}_{k},y_{k+1};\mathbf{y}%
_{k}^{\prime },y_{k+1})dy_{k+1}, \\
&&\left( \tilde{B}_{N,j,k+1,\tau }^{2}u_{N}^{(k+1)}\right) (\tau ,\mathbf{y}%
_{k};\mathbf{y}_{k}^{\prime })=\int \tilde{V}_{N,\tau }(y_{j}^{\prime
}-y_{k+1})u_{N}^{(k+1)}(\tau ,\mathbf{y}_{k},y_{k+1};\mathbf{y}_{k}^{\prime
},y_{k+1})dy_{k+1}.
\end{eqnarray*}%
and $U^{(k)}(\tau )$ is the solution operator to the free equation, that is%
\begin{equation*}
U^{(k)}(\tau )=e^{\frac{i\tau }{2}\triangle _{\mathbf{y}_{k}}}e^{-\frac{%
i\tau }{2}\triangle _{\mathbf{y}_{k}^{\prime }}}.
\end{equation*}%
Then there is a $C$ independent of $j,$ $k$ and $N$ such that%
\begin{equation*}
\int_{0}^{T}\left\Vert R^{(k)}\tilde{B}_{N,j,k+1,\tau
}u_{N}^{(k+1)}\right\Vert _{L^{2}}d\tau \leqslant C^{k}.
\end{equation*}%
where%
\begin{equation*}
R^{(k)}=\dprod_{j=1}^{k}\left( \left\vert \nabla _{y_{j}}\right\vert
\left\vert \nabla _{y_{j}^{\prime }}\right\vert \right) .
\end{equation*}
\end{theorem}

\begin{proof}
We prove our main auxiliary theorem in Section \ref{Section:Space-TimeBound}%
. This theorem establishes the Klainerman-Machedon space-time bound for $%
\beta \in \left( 0,\frac{2}{7}\right] $.
\end{proof}

Via the above theorem, we infer that every limit point $\left\{
u^{(k)}\right\} $ of $\left\{ u_{N}^{(k)}\right\} $ satisfies the space time
bound%
\begin{equation*}
\int_{0}^{\frac{\tan \omega T_{0}}{\omega }}\left\Vert
R^{(k)}B_{j,k+1}u^{(k+1)}(\tau ,\mathbf{\cdot };\mathbf{\cdot })\right\Vert
_{L^{2}\left( \mathbb{R}^{3k}\times \mathbb{R}^{3k}\right) }d\tau \leqslant
C^{k},
\end{equation*}%
for some $C>0$ and all $1\leqslant j\leqslant k.$

\begin{itemize}
\item[Step III.] Regarding the solution to the infinite hierarchy \ref%
{hierarchy:GPwithTimeFactor}, we have the following uniqueness theorem.
\end{itemize}

\begin{theorem}
\label{Theorem:Uniqueness of GP}Let $\left\{ u^{(k)}\right\} $ be a solution
of the infinite hierarchy \ref{hierarchy:GPwithTimeFactor} in $\left[ s,T%
\right] \subset \left[ 0,\frac{\tan \omega T_{0}}{\omega }\right] $ subject
to zero initial data that is 
\begin{equation*}
u^{(k)}(s,\mathbf{y}_{k};\mathbf{y}_{k}^{\prime })=0,\forall k,
\end{equation*}%
and the space time bound%
\begin{equation}
\int_{s}^{T}\left\Vert R^{(k)}B_{j,k+1}u^{(k+1)}(\tau ,\mathbf{\cdot };%
\mathbf{\cdot })\right\Vert _{L^{2}\left( \mathbb{R}^{3k}\times \mathbb{R}%
^{3k}\right) }d\tau \leqslant C^{k},  \label{formula:theSpace-TimeBound}
\end{equation}%
for some $C>0$ and all $1\leqslant j\leqslant k.$ Then $\forall k,\tau \in %
\left[ s,T\right] $, we have%
\begin{equation*}
\left\Vert R^{(k)}u^{(k)}(\tau ,\mathbf{\cdot };\mathbf{\cdot })\right\Vert
_{L^{2}\left( \mathbb{R}^{3k}\times \mathbb{R}^{3k}\right) }=0.
\end{equation*}
\end{theorem}

\begin{proof}
See Section \ref{Section:Uniqueness}.
\end{proof}

Since we have shown the space-time bound \ref{formula:theSpace-TimeBound} in
Step III, we apply the above uniqueness theorem and deduce that%
\begin{equation}
u^{(k)}(\tau ,\mathbf{y}_{k};\mathbf{y}_{k}^{\prime })=\dprod_{j=1}^{k}%
\tilde{\phi}(\tau ,y_{j})\overline{\tilde{\phi}}(\tau ,y_{j}^{\prime }),
\label{formula:solnToGP}
\end{equation}%
where $\tilde{\phi}(\tau ,y)$ solves the 3D NLS%
\begin{eqnarray}
i\partial _{\tau }\tilde{\phi} &=&-\frac{1}{2}\triangle _{y}\tilde{\phi}+%
\frac{b_{0}\left\vert \tilde{\phi}\right\vert ^{2}\tilde{\phi}}{g(\tau )}%
\text{ in }\mathbb{R}^{3+1}
\label{equation:cubicNLS with time dependent coupling} \\
\tilde{\phi}(0,y) &=&\phi _{0}.  \notag
\end{eqnarray}

Hence the compact sequence $\left\{ u_{N}^{(k)}\right\} $ has only one limit
point. So%
\begin{equation*}
u_{N}^{(k)}\rightarrow \dprod_{j=1}^{k}\tilde{\phi}(\tau ,y_{j})\overline{%
\tilde{\phi}}(\tau ,y_{j}^{\prime })\text{ as }N\rightarrow \infty
\end{equation*}%
in the weak* topology. Since $u^{(k)}$ is an orthogonal projection, the
convergence in the weak* topology is then equivalent to the convergence in
the trace norm.

\begin{example}
At the suggestion of Professor Walter Strauss, we give a brief explanation
on why a factorized state like formula \ref{formula:solnToGP} is a solution
to the Gross-Pitaevskii hierarchy. Consider $k=1$, then plugging $\tilde{\phi%
}(\tau ,y_{1})\overline{\tilde{\phi}}(\tau ,y_{1}^{\prime })$ into the
infinite hierarchy yields 
\begin{eqnarray*}
&&\left( i\partial _{\tau }+\frac{1}{2}\triangle _{y_{1}}-\frac{1}{2}%
\triangle _{y_{1}^{\prime }}\right) \left( \tilde{\phi}(\tau ,y_{1})%
\overline{\tilde{\phi}}(\tau ,y_{1}^{\prime })\right) \\
&=&b_{0}\frac{\left( \left\vert \tilde{\phi}\right\vert ^{2}\tilde{\phi}%
\right) (\tau ,y_{1})\overline{\tilde{\phi}}(\tau ,y_{1}^{\prime })}{g(\tau )%
}-b_{0}\frac{\tilde{\phi}(\tau ,y_{1})\left( \left\vert \tilde{\phi}%
\right\vert ^{2}\overline{\tilde{\phi}}\right) (\tau ,y_{1}^{\prime })}{%
g(\tau )} \\
&=&\frac{b_{0}B_{1,2}^{1}\left( \tilde{\phi}(\tau ,y_{1})\tilde{\phi}(\tau
,y_{2})\overline{\tilde{\phi}}(\tau ,y_{1}^{\prime })\overline{\tilde{\phi}}%
(\tau ,y_{2}^{\prime })\right) -b_{0}B_{1,2}^{2}\left( \tilde{\phi}(\tau
,y_{1})\tilde{\phi}(\tau ,y_{2})\overline{\tilde{\phi}}(\tau ,y_{1}^{\prime
})\overline{\tilde{\phi}}(\tau ,y_{2}^{\prime })\right) }{g(\tau )},
\end{eqnarray*}%
which is%
\begin{equation*}
\left( i\partial _{\tau }+\frac{1}{2}\triangle _{y_{1}}-\frac{1}{2}\triangle
_{y_{1}^{\prime }}\right) u^{(1)}=\frac{1}{g(\tau )}b_{0}B_{1,2}u^{(2)}.
\end{equation*}
\end{example}

\begin{itemize}
\item[Step IV.] In Step III, we have concluded the convergence%
\begin{equation*}
\lim_{N\rightarrow \infty }\limfunc{Tr}\left\vert u_{N}^{(k)}(\tau ,\mathbf{y%
}_{k};\mathbf{y}_{k}^{\prime })-\dprod_{j=1}^{k}\tilde{\phi}(\tau ,y_{j})%
\overline{\tilde{\phi}}(\tau ,y_{j}^{\prime })\right\vert =0,\text{ }\forall
\tau \in \left[ 0,\frac{\tan \omega T_{0}}{\omega }\right] .
\end{equation*}%
Notice that $u_{N}^{(k)}=T_{k}^{-1}\gamma _{N}^{(k)}$ and the lens transform
of $u^{(k)}$ is%
\begin{equation*}
\gamma ^{(k)}(t,\mathbf{x}_{k};\mathbf{x}_{k}^{\prime
})=\dprod_{j=1}^{k}\phi (t,x_{j})\overline{\phi }(t,x_{j}^{\prime }),
\end{equation*}%
where $\phi (t,x)$ solves 
\begin{eqnarray*}
i\partial _{t}\phi &=&\left( -\frac{1}{2}\triangle _{x}+\omega ^{2}\frac{%
\left\vert x\right\vert ^{2}}{2}\right) \phi +b_{0}\left\vert \phi
\right\vert ^{2}\phi \text{ in }\mathbb{R}^{3+1} \\
\phi (0,x) &=&\phi _{0}(x)
\end{eqnarray*}%
i.e. equation \ref{equation:cubicNLSwithSwitch}. Thence we conclude that%
\begin{equation*}
\lim_{N\rightarrow \infty }\limfunc{Tr}\left\vert \gamma _{N}^{(k)}(t,%
\mathbf{x}_{k};\mathbf{x}_{k}^{\prime })-\dprod_{j=1}^{k}\phi (t,x_{j})%
\overline{\phi }(t,x_{j}^{\prime })\right\vert =0,\text{ }\forall t\in \left[
0,T_{0}\right] ,
\end{equation*}%
as a result of the fact that the lens transform preserves the trace norm
(Lemma \ref{Lemma:TraceNormPreservation}). Since equation \ref%
{equation:N-Body Schrodinger with switchable trap} is time translation
invariant, we have established Theorem \ref%
{Theorem:SmoothVersionofTheMainTheorem} and consequently Theorem \ref%
{Theorem:3D BEC}. The purpose of the rest of this paper is to prove Theorems %
\ref{Theorem:Uniqueness of GP} and \ref{Theorem:Space-Time Bound for
GlensBBGKY}.
\end{itemize}

\begin{remark}
Since $\tilde{\phi}=M_{1}^{-1}\phi $, the global $H^{1}$ well-posedness of
equation \ref{equation:cubicNLS with time dependent coupling} is implied by
the global well-posedness in the scattering space $\sum $ of equation \ref%
{equation:cubicNLSwithSwitch} which comes from the Strichartz estimates.
Through the lens transform, we always have a $L^{2}$ solution to equation %
\ref{equation:cubicNLS with time dependent coupling}. Lemma \ref%
{Lemma:momentum between two sides of the lens transform} then shows $\nabla 
\tilde{\phi}\in L^{2}.$
\end{remark}

\section{The Uniqueness of Hierarchy \protect\ref{hierarchy:GPwithTimeFactor}
(Proof of Theorem \protect\ref{Theorem:Uniqueness of GP}) \label%
{Section:Uniqueness}}

In this section, we produce Theorem \ref{Theorem:Uniqueness of GP} with
Lemmas \ref{Theorem:CollapsingEstimateForGP} and \ref%
{lemma:Klainerman-MachedonBoardGameForGP}. For convenience, we set the
coupling constant $b_{0}$ in the infinite hierarchy \ref%
{hierarchy:GPwithTimeFactor}$\ $to be $1$.

\begin{lemma}
\label{Theorem:CollapsingEstimateForGP}\cite{KlainermanAndMachedon} Assume $%
u^{(k+1)}$ verifies%
\begin{equation*}
\left( i\partial _{\tau }+\frac{1}{2}\triangle _{\mathbf{y}_{k+1}}-\frac{1}{2%
}\triangle _{\mathbf{y}_{k+1}^{\prime }}\right) u^{(k+1)}=0,
\end{equation*}%
then there is a $C>0,$ independent of $j,$ $k,$ and $u^{(k+1)}$ s.t. 
\begin{eqnarray*}
&&\left\Vert R^{(k)}\left( B_{j,k+1}u^{(k+1)}\right) (\tau ,\mathbf{y}_{k};%
\mathbf{y}_{k}^{\prime })\right\Vert _{L^{2}(\mathbb{R}\times \mathbb{R}%
^{3k}\times \mathbb{R}^{3k})} \\
&\leqslant &C\left\Vert R^{(k+1)}u^{(k+1)}(0,\mathbf{y}_{k+1};\mathbf{y}%
_{k+1}^{\prime })\right\Vert _{L^{2}(\mathbb{R}^{3(k+1)}\times \mathbb{R}%
^{3(k+1)})}.
\end{eqnarray*}
\end{lemma}

\begin{proof}
This is Theorem 1.3 of \cite{KlainermanAndMachedon}. For some other
estimates of this type, see \cite{ChenDie,ChenAnisotropic,GM,Kirpatrick}.
\end{proof}

\begin{lemma}
\label{lemma:Klainerman-MachedonBoardGameForGP}Assuming zero initial data
i.e. 
\begin{equation*}
u^{(k)}(s,\mathbf{y}_{k};\mathbf{y}_{k}^{\prime })=0,\forall k,
\end{equation*}%
then one can express $u^{(1)}(\tau _{1},\mathbf{\cdot };\mathbf{\cdot })$ in
the Gross-Pitaevskii hierarchy \ref{hierarchy:GPwithTimeFactor} as a sum of
at most $4^{n}$ terms of the form 
\begin{equation*}
\int_{D}J(\underline{\tau }_{n+1},\mu _{m})\left( u^{(n+1)}(\tau
_{n+1})\right) d\underline{\tau }_{n+1},
\end{equation*}%
or in other words, 
\begin{equation}
u^{(1)}(\tau _{1},\mathbf{\cdot };\mathbf{\cdot })=\sum_{m}\int_{D}J(%
\underline{\tau }_{n+1},\mu _{m})\left( u^{(n+1)}(\tau _{n+1})\right) d%
\underline{\tau }_{n+1}.  \label{formula:MateiKlainermanFormula}
\end{equation}%
Here $\underline{\tau }_{n+1}=(\tau _{2},\tau _{3},...,\tau _{n+1})$, $%
D\subset \lbrack s,\tau _{1}]^{n}$, $\mu _{m}$ are a set of maps from $%
\{2,...,n+1\}$ to $\{1,...,n\}$ satisfying $\mu _{m}(2)=1$ and $\mu
_{m}(j)<j $ for all $j,$ and%
\begin{eqnarray*}
J(\underline{\tau }_{n+1},\mu _{m})\left( u^{(n+1)}(\tau _{n+1})\right)
&=&\left( \dprod_{j=1}^{n}g\left( \tau _{j+1}\right) \right)
^{-1}U^{(1)}(\tau _{1}-\tau _{2})B_{1,2}U^{(2)}(\tau _{2}-\tau _{3})B_{\mu
_{m}(3),2}... \\
&&U^{(n)}(\tau _{n}-\tau _{n+1})B_{\mu _{m}(n+1),n+1}(u^{(n+1)}(\tau _{n+1},%
\mathbf{\cdot };\mathbf{\cdot })).
\end{eqnarray*}
\end{lemma}

\begin{proof}
The RHS of formula \ref{formula:MateiKlainermanFormula} is in fact an
application of Duhamel's principle involving only the inhomogeneous terms
since we have zero initial data. The parameter $n$ is the coupling level we
take. This lemma follows from the proof of Theorem 3.4 in \cite%
{KlainermanAndMachedon}. One needs only notice that factors depending solely
on $\tau ,$ e.g. 
\begin{equation*}
\frac{1}{g\left( \tau _{j+1}\right) }
\end{equation*}%
commutes with $U^{(k)}$ and $B_{j,k+1}$ $\forall j,k$.
\end{proof}

With Lemmas \ref{Theorem:CollapsingEstimateForGP} and \ref%
{lemma:Klainerman-MachedonBoardGameForGP}, we prove Theorem \ref%
{Theorem:Uniqueness of GP}. Let $D_{\tau _{2}}=\left\{ \left( \tau
_{3},...,\tau _{n+1}\right) |\left( \tau _{2},\tau _{3},...,\tau
_{n+1}\right) \in D\right\} $ where $D$ is as in Lemma \ref%
{lemma:Klainerman-MachedonBoardGameForGP}. Given that we have already
checked that 
\begin{equation*}
\left\Vert R^{(1)}u^{(1)}(s_{0},\cdot )\right\Vert _{L^{2}(\mathbb{R}%
^{3}\times \mathbb{R}^{3})}=0,
\end{equation*}%
applying Lemma \ref{lemma:Klainerman-MachedonBoardGameForGP} to $[s_{0},\tau
_{1}]\subset \lbrack s,T]\subset \lbrack 0,\frac{\tan \omega T_{0}}{\omega }%
] $, we have

\begin{eqnarray*}
&&\left\Vert R^{(1)}u^{(1)}(\tau _{1},\cdot )\right\Vert _{L^{2}(\mathbb{R}%
^{3}\times \mathbb{R}^{3})} \\
&\leqslant &\sum_{m}\left\Vert R^{(1)}\int_{D}J(\underline{\tau }_{n+1},\mu
_{m})\left( u^{(n+1)}(\tau _{n+1})\right) d\underline{\tau }%
_{n+1}\right\Vert _{L^{2}(\mathbb{R}^{3}\times \mathbb{R}^{3})} \\
&\leqslant &\sum_{m}\int_{[s_{0},\tau _{1}]^{n}}\left\Vert R^{(1)}J(%
\underline{\tau }_{n+1},\mu _{m})\left( u^{(n+1)}(\tau _{n+1})\right)
\right\Vert _{L^{2}(\mathbb{R}^{3}\times \mathbb{R}^{3})}d\underline{\tau }%
_{n+1} \\
&\leqslant &\sum_{m}\left( \inf\limits_{\tau \in \lbrack 0,\frac{\tan \omega
T_{0}}{\omega }]}g(\tau )\right) ^{-n}\int_{[s_{0},\tau _{1}]^{n}}\left\Vert
R^{(1)}U^{(1)}(\tau _{1}-\tau _{2})B_{1,2}U^{(2)}(\tau _{2}-\tau _{3})B_{\mu
_{m}(3),2}...\right\Vert _{L^{2}(\mathbb{R}^{3}\times \mathbb{R}^{3})}d%
\underline{\tau }_{n+1} \\
&=&\sum_{m}C^{n}\int_{[s_{0},\tau _{1}]^{n}}\left\Vert
R^{(1)}B_{1,2}U^{(2)}(\tau _{2}-\tau _{3})B_{\mu _{m}(3),2}...\right\Vert
_{L^{2}(\mathbb{R}^{3}\times \mathbb{R}^{3})}d\underline{\tau }_{n+1} \\
&\leqslant &\sum_{m}C^{n}(\tau _{1}-s_{0})^{\frac{1}{2}}\int_{[s_{0},\tau
_{1}]^{n-1}}\left( \int \left\Vert R^{(1)}B_{1,2}U^{(2)}(\tau _{2}-\tau
_{3})B_{\mu _{m}(3),2}...\right\Vert _{L^{2}(\mathbb{R}^{3}\times \mathbb{R}%
^{3})}^{2}d\tau _{2}\right) ^{\frac{1}{2}}d\tau _{3}...d\tau _{n+1} \\
&\leqslant &\sum_{m}C^{n}C(\tau _{1}-s_{0})^{\frac{1}{2}}\int_{[s_{0},\tau
_{1}]^{n-1}}\left\Vert R^{(2)}B_{\mu _{m}(3),2}U^{(3)}(\tau _{3}-\tau
_{4})...\right\Vert _{L^{2}(\mathbb{R}^{3}\times \mathbb{R}^{3})}d\tau
_{3}...d\tau _{n+1}\text{ (Lemma \ref{Theorem:CollapsingEstimateForGP})} \\
&&(\text{Iterate }n-2\text{ times}) \\
&&... \\
&\leqslant &\sum_{m}C^{n}\left( C(\tau _{1}-s_{0})\right) ^{\frac{n-1}{2}%
}\int_{s_{0}}^{\tau _{1}}\left\Vert R^{(n)}B_{\mu
_{m}(n+1),n+1}u^{(n+1)}(\tau _{n+1},\mathbf{\cdot };\mathbf{\cdot }%
)\right\Vert _{L^{2}(\mathbb{R}^{3}\times \mathbb{R}^{3})}d\tau _{n+1} \\
&\leqslant &C\left( C(\tau _{1}-s_{0})\right) ^{\frac{n-1}{2}}.
\end{eqnarray*}%
Let $(\tau _{1}-s_{0})$ be sufficiently small, and $n\rightarrow \infty $,
we infer that 
\begin{equation*}
\left\Vert R^{(1)}u^{(1)}(\tau _{1},\cdot )\right\Vert _{L^{2}(\mathbb{R}%
^{3}\times \mathbb{R}^{3})}=0\text{ in }\left[ s_{0},\tau _{1}\right] .
\end{equation*}%
Such a choice of $(\tau _{1}-s_{0})$ works for all of $\left[ s_{0},T\right]
.\ $Accordingly, we have $\left\Vert R^{(k)}u^{(k)}(\tau ,\cdot )\right\Vert
_{L^{2}(\mathbb{R}^{3}\times \mathbb{R}^{3})}=0$, $\forall k,\tau \in \left[
s_{0},T\right] $ by iterating the above argument$.$ Thence we have attained
Theorem \ref{Theorem:Uniqueness of GP}.

\section{The Space-Time Bound of the BBGKY Hierarchy for $\protect\beta \in
\left( 0,2/7\right] $ (Proof of Theorem \protect\ref{Theorem:Space-Time
Bound for GlensBBGKY})\label{Section:Space-TimeBound}}

We establish Theorem \ref{Theorem:Space-Time Bound for GlensBBGKY} in this
section. This section also serves as a simplification and an extension of
Chen-Pavlovic \cite{TChenAndNPSpace-Time}. Without loss of generality, we
may assume $k=1$ that is%
\begin{equation}
\int_{0}^{T}\left\Vert R^{(1)}\tilde{B}_{N,1,2,\tau }u_{N}^{(2)}\right\Vert
_{L^{2}}d\tau \leqslant C.
\label{estimate:MainEstimateOfSpace-TimeBoundSection}
\end{equation}%
We are going to prove estimate \ref%
{estimate:MainEstimateOfSpace-TimeBoundSection}\ for a sufficiently small
time $T$ determined by the controlling constant in condition \ref%
{condition:EnergyBoundForBBGKY} and independent of $N$, then the
bootstrapping argument in Section \ref{Section:Uniqueness} (Proof of Theorem %
\ref{Theorem:Uniqueness of GP}) and condition \ref%
{condition:EnergyBoundForBBGKY} provide the bound for every finite time $%
T\in \lbrack 0,\frac{\tan \omega T_{0}}{\omega }]$. Since we work with $%
L^{2} $ norms here, we transform condition \ref%
{condition:EnergyBoundForBBGKY} into the $H^{1}$ energy bound:%
\begin{equation*}
\int \left\vert \left( \dprod\limits_{j=1}^{k}\left( 1-\triangle
_{y_{j}}\right) ^{\frac{1}{2}}\left( 1-\triangle _{y_{j}^{\prime }}\right) ^{%
\frac{1}{2}}\right) u_{N}^{(k)}\right\vert ^{2}d\mathbf{y}_{k}d\mathbf{y}%
_{k}^{\prime }\leqslant \left( \limfunc{Tr}\left( \dprod_{j=1}^{k}\left(
1-\triangle _{y_{j}}\right) \right) u_{N}^{(k)}\right) ^{2}.
\end{equation*}%
To obtain the above estimate, one notices%
\begin{eqnarray*}
&&\int \left\vert \left( 1-\triangle _{y}\right) ^{\frac{1}{2}}\left(
1-\triangle _{y^{\prime }}\right) ^{\frac{1}{2}}\int \phi \left( y,r\right) 
\overline{\phi \left( y^{\prime },r\right) }dr\right\vert ^{2}dydy^{\prime }
\\
&=&\int \left\vert \int \left( 1-\triangle _{y}\right) ^{\frac{1}{2}}\phi
\left( y,r\right) \overline{\left( 1-\triangle _{y^{\prime }}\right) ^{\frac{%
1}{2}}\phi \left( y^{\prime },r\right) }dr\right\vert ^{2}dydy^{\prime } \\
&\leqslant &\int \left( \int \left( 1-\triangle _{y}\right) ^{\frac{1}{2}%
}\phi \left( y,r\right) \overline{\left( 1-\triangle _{y}\right) ^{\frac{1}{2%
}}\phi \left( y,r\right) }dr\right) \left( \int \left( 1-\triangle
_{y^{\prime }}\right) ^{\frac{1}{2}}\phi \left( y^{\prime },r\right) 
\overline{\left( 1-\triangle _{y^{\prime }}\right) ^{\frac{1}{2}}\phi \left(
y^{\prime },r\right) }dr\right) dydy^{\prime } \\
&=&\left( \int \phi \left( y,r\right) \overline{\left( 1-\triangle
_{y}\right) \phi \left( y,r\right) }dydr\right) ^{2},
\end{eqnarray*}%
the energy bound then follows from the definition of $u_{N}^{(k)}$. The
analysis of Theorem \ref{Theorem:Space-Time Bound for GlensBBGKY} also
involves $\tilde{B}_{N,j,k+1,\tau }$ which approximates $B_{j,k+1}$ for
every $\tau ,$ so we generalize Lemma \ref{Theorem:CollapsingEstimateForGP}
to the following collapsing estimate.

\begin{theorem}
\label{Theorem:3*3d}Suppose $u(\tau ,y_{1},y_{2},y_{2}^{\prime })$ solves
the Schr\"{o}dinger equation 
\begin{eqnarray}
iu_{\tau }+\frac{1}{2}\triangle _{y_{1}}u+\frac{1}{2}\triangle _{y_{2}}u-%
\frac{1}{2}\triangle _{y_{2}^{\prime }}u &=&0\text{ in }\mathbb{R}^{9+1}
\label{eqn:3*3d} \\
u(0,y_{1},y_{2},y_{2}^{\prime }) &=&f(y_{1},y_{2},y_{2}^{\prime }),  \notag
\end{eqnarray}%
and $g(\tau )\geqslant c_{0}>0$, then there is a $C$ independent of $N$ and $%
u$ such that%
\begin{eqnarray*}
&&\int_{\mathbb{R}^{3+1}}\left\vert \left\vert \nabla _{y}\right\vert \left(
\int \left( g^{3}(\tau )V_{N}(g(\tau )\left( y-y_{2}\right) )\right) \delta
(y_{2}-y_{2}^{\prime })u(\tau ,y,y_{2};y_{2}^{\prime })dy_{2}dy_{2}^{\prime
}\right) \right\vert ^{2}dyd\tau \\
&\leqslant &Cb_{0}^{2}\left\Vert \left\vert \nabla _{y_{1}}\right\vert
\left\vert \nabla _{y_{2}}\right\vert \left\vert \nabla _{y_{2}^{\prime
}}\right\vert f\right\Vert _{2}^{2}.
\end{eqnarray*}%
where $b_{0}=\int V_{N}dy.$
\end{theorem}

\begin{proof}
Theorem \ref{Theorem:3*3d} follows from a slightly modified version of the
proof of Theorem 2 of \cite{ChenAnisotropic}. We include it in Appendix II
for completeness.
\end{proof}

We now present the proof of estimate \ref%
{estimate:MainEstimateOfSpace-TimeBoundSection}. In order to more
conveniently apply the Klainerman-Machedon board game, let us start by
rewriting hierarchy \ref{equation:Duhamel of BBGKY} as%
\begin{eqnarray}
u_{N}^{(k)}(\tau _{k}) &=&U^{(k)}(\tau _{k})u_{N,0}^{(k)}+\int_{0}^{\tau
_{k}}U^{(k)}(\tau _{k}-\tau _{k+1})\frac{\tilde{V}_{N,\tau
_{k+1}}^{(k)}u_{N}^{(k)}(\tau _{k+1})}{g\left( \tau _{k+1}\right) }d\tau
_{k+1}  \label{equation:short duhamel of BBGKY} \\
&&+\frac{N-k}{N}\int_{0}^{\tau _{k}}U^{(k)}(\tau _{k}-\tau _{k+1})\frac{%
\tilde{B}_{N,\tau _{k+1}}^{(k+1)}u_{N}^{(k+1)}(\tau _{k+1})}{g\left( \tau
_{k+1}\right) }d\tau _{k+1}  \notag
\end{eqnarray}%
where 
\begin{eqnarray*}
\tilde{V}_{N,\tau }^{(k)}u_{N}^{(k)}(\tau ,\mathbf{y}_{k};\mathbf{y}%
_{k}^{\prime }) &=&\frac{1}{N}\sum_{1\leqslant i<j\leqslant k}\left( \tilde{V%
}_{N,\tau }(y_{i}-y_{j})-\tilde{V}_{N,\tau }(y_{i}^{\prime }-y_{j}^{\prime
})\right) u_{N}^{(k)}(\tau ,\mathbf{y}_{k};\mathbf{y}_{k}^{\prime }) \\
\tilde{B}_{N,\tau }^{(k+1)}u_{N}^{(k+1)} &=&\sum_{j=1}^{k}\tilde{B}%
_{N,j,k+1,\tau }u_{N}^{(k+1)}=\sum_{j=1}^{k}\left( \tilde{B}_{N,j,k+1,\tau
}^{1}-\tilde{B}_{N,j,k+1,\tau }^{2}\right) u_{N}^{(k+1)}.
\end{eqnarray*}%
We omit the imaginary unit in front of the potential term and the
interaction term so that we do not need to keep track of its exact power.

Iterate Duhamel's principle (equation \ref{equation:short duhamel of BBGKY}) 
$k$ times, we have%
\begin{eqnarray*}
&&u_{N}^{(2)}(\tau _{2}) \\
&=&U^{(2)}(\tau _{2})u_{N,0}^{(2)}+\int_{0}^{\tau _{2}}U^{(2)}(\tau
_{2}-\tau _{3})\frac{\tilde{V}_{N,\tau _{3}}^{(2)}u_{N}^{(2)}(\tau _{3})}{%
g\left( \tau _{3}\right) }d\tau _{3}+\frac{N-2}{N}\int_{0}^{\tau
_{2}}U^{(2)}(\tau _{2}-\tau _{3})\frac{\tilde{B}_{N,\tau
_{3}}^{(3)}u_{N}^{(3)}(\tau _{3})}{g\left( \tau _{3}\right) }d\tau _{3} \\
&=&U^{(2)}(\tau _{2})u_{N,0}^{(2)}+\frac{N-2}{N}\int_{0}^{\tau
_{2}}U^{(2)}(\tau _{2}-\tau _{3})\frac{\tilde{B}_{N,\tau
_{3}}^{(3)}U^{(3)}(\tau _{3})u_{N,0}^{(3)}}{g\left( \tau _{3}\right) }d\tau
_{3}+\int_{0}^{\tau _{2}}U^{(2)}(\tau _{2}-\tau _{3})\frac{\tilde{V}_{N,\tau
_{3}}^{(2)}u_{N}^{(2)}(\tau _{3})}{g\left( \tau _{3}\right) }d\tau _{3} \\
&&+\frac{N-2}{N}\int_{0}^{\tau _{2}}U^{(2)}(\tau _{2}-\tau _{3})\frac{\tilde{%
B}_{N,\tau _{3}}^{(3)}}{g\left( \tau _{3}\right) }\int_{0}^{\tau
_{3}}U^{(3)}(\tau _{3}-\tau _{4})\frac{\tilde{V}_{N,\tau
_{4}}^{(3)}u_{N}^{(3)}(\tau _{4})}{g\left( \tau _{4}\right) }d\tau _{4}d\tau
_{3} \\
&&+\frac{N-2}{N}\frac{N-3}{N}\int_{0}^{\tau _{2}}U^{(2)}(\tau _{2}-\tau _{3})%
\frac{\tilde{B}_{N,\tau _{3}}^{(3)}}{g\left( \tau _{3}\right) }%
\int_{0}^{\tau _{3}}U^{(3)}(\tau _{3}-\tau _{4})\frac{\tilde{B}_{N,\tau
_{4}}^{(4)}u_{N}^{(4)}(\tau _{4})}{g\left( \tau _{4}\right) }d\tau _{4}d\tau
_{3} \\
&&... \\
&=&FreePart^{(k)}+PotentialPart^{(k)}+InteractionPart^{(k)}
\end{eqnarray*}%
where%
\begin{eqnarray*}
&&FreePart^{(k)} \\
&=&U^{(2)}(\tau _{2})u_{N,0}^{(2)}+\sum_{j=3}^{k}\left( \dprod_{l=3}^{j}%
\frac{N+1-l}{N}\right) \\
&&\times \int_{0}^{\tau _{2}}...\int_{0}^{\tau _{j-1}}\left(
\dprod_{l=2}^{j-1}g\left( \tau _{l+1}\right) \right) ^{-1}U^{(2)}(\tau
_{2}-\tau _{3})\tilde{B}_{N,\tau _{3}}^{(3)}...U^{(j-1)}(\tau _{j-1}-\tau
_{j})\tilde{B}_{N,\tau _{j}}^{(j)} \\
&&\times \left( U^{(j)}(\tau _{j})u_{N,0}^{(j)}\right) d\tau _{3}...d\tau
_{j},
\end{eqnarray*}%
\begin{eqnarray*}
&&PotentialPart^{(k)} \\
&=&\int_{0}^{\tau _{2}}U^{(2)}(\tau _{2}-\tau _{3})\frac{\tilde{V}_{N,\tau
_{3}}^{(2)}u_{N}^{(2)}(\tau _{3})}{g\left( \tau _{3}\right) }d\tau
_{3}+\sum_{j=3}^{k}\left( \dprod_{l=3}^{j}\frac{N+1-l}{N}\right) \\
&&\times \int_{0}^{\tau _{2}}...\int_{0}^{\tau _{j-1}}\left(
\dprod_{l=2}^{j-1}g\left( \tau _{l+1}\right) \right) ^{-1}U^{(2)}(\tau
_{2}-\tau _{3})\tilde{B}_{N,\tau _{3}}^{(3)}...U^{(j-1)}(\tau _{j-1}-\tau
_{j})\tilde{B}_{N,\tau _{j}}^{(j)} \\
&&\times \left( \int_{0}^{\tau _{j}}U^{(j)}(\tau _{j}-\tau _{j+1})\frac{%
\tilde{V}_{N,\tau _{j+1}}^{(j)}u_{N}^{(j)}(\tau _{j+1})}{g\left( \tau
_{j+1}\right) }d\tau _{j+1}\right) d\tau _{3}...d\tau _{j},
\end{eqnarray*}%
\begin{eqnarray*}
&&InteractionPart^{(k)} \\
&=&\left( \dprod_{l=3}^{k+1}\frac{N+1-l}{N}\right) \\
&&\times \int_{0}^{\tau _{2}}...\int_{0}^{\tau _{k}}\left(
\dprod_{l=2}^{k}g\left( \tau _{l+1}\right) \right) ^{-1}U^{(2)}(\tau
_{2}-\tau _{3})\tilde{B}_{N,\tau _{3}}^{(3)}...U^{(k)}(\tau _{k}-\tau _{k+1})%
\tilde{B}_{N,\tau _{k+1}}^{(k+1)} \\
&&\times \left( u_{N}^{(k+1)}(\tau _{k+1})\right) d\tau _{3}...d\tau _{k+1}
\end{eqnarray*}%
From here on out, the $k$'s in the formulas are the number of Duhamel
iterations we take to prove estimate \ref%
{estimate:MainEstimateOfSpace-TimeBoundSection}. We will call it the
coupling level for short. It is distinct from the $k$ in the statement of
Theorem \ref{Theorem:Space-Time Bound for GlensBBGKY}.

We are going to argue%
\begin{eqnarray}
\int_{0}^{T}\left\Vert R^{(1)}\tilde{B}_{N,1,2,\tau _{2}}FreePart^{(k)}(\tau
_{2})\right\Vert _{L^{2}}d\tau _{2} &\leqslant &C
\label{estimate:FreePartofBBGKY} \\
\int_{0}^{T}\left\Vert R^{(1)}\tilde{B}_{N,1,2,\tau
_{2}}PotentialPart^{(k)}(\tau _{2})\right\Vert _{L^{2}}d\tau _{2} &\leqslant
&C  \label{estimate:PotentialPartofBBGKY} \\
\int_{0}^{T}\left\Vert R^{(1)}\tilde{B}_{N,1,2,\tau
_{2}}InteractionPart^{(k)}(\tau _{2})\right\Vert _{L^{2}}d\tau _{2}
&\leqslant &C  \label{estimate:InteractionPartofBBGKY}
\end{eqnarray}%
for some $C$ and a sufficiently small $T$ determined by the controlling
constant in condition \ref{condition:EnergyBoundForBBGKY} and independent of 
$N.$ We observe that $\tilde{B}_{N,\tau _{j}}^{(j)}$ has $2j$ terms inside
so that each summand of $u_{N}^{(2)}(\tau _{2})$ contains factorially many
terms $\left( \sim k!\right) $. So we use the Klainerman-Machedon board game
to reduce the number of terms. Define%
\begin{equation*}
J_{N}(\underline{\tau }_{j+1})(f)=\left( \dprod_{l=2}^{j}g\left( \tau
_{l+1}\right) \right) ^{-1}U^{(2)}(\tau _{2}-\tau _{3})\tilde{B}_{N,\tau
_{3}}^{(3)}...U^{(j)}(\tau _{j}-\tau _{j+1})\tilde{B}_{N,\tau
_{j+1}}^{(j+1)}f,
\end{equation*}%
where $\underline{\tau }_{j+1}$ means $\left( \tau _{3},...,\tau
_{j+1}\right) ,$ then the Klainerman-Machedon board game implies the lemma.

\begin{lemma}
\label{lemma:Klainerman-MachedonBoardGameForBBGKY}\cite%
{KlainermanAndMachedon} One can express 
\begin{equation*}
\int_{0}^{\tau _{2}}...\int_{0}^{\tau _{j}}J_{N}(\underline{\tau }_{j+1})(f)d%
\underline{\tau }_{j+1}
\end{equation*}%
as a sum of at most $4^{j-1}$ terms of the form 
\begin{equation*}
\int_{D}J_{N}(\underline{\tau }_{j+1},\mu _{m})(f)d\underline{\tau }_{j+1},
\end{equation*}%
or in other words, 
\begin{equation*}
\int_{0}^{\tau _{2}}...\int_{0}^{\tau _{j}}J_{N}(\underline{\tau }_{j+1})(f)d%
\underline{\tau }_{j+1}=\sum_{m}\int_{D}J_{N}(\underline{\tau }_{j+1},\mu
_{m})(f)d\underline{\tau }_{j+1}.
\end{equation*}%
Here $D\subset \lbrack 0,\tau _{2}]^{j-1}$, $\mu _{m}$ are a set of maps
from $\{3,...,j+1\}$ to $\{2,...,j\}$ satisfying $\mu _{m}(3)=2$ and $\mu
_{m}(l)<l$ for all $l,$ and%
\begin{eqnarray*}
J_{N}(\underline{\tau }_{j+1},\mu _{m})(f) &=&\left( \dprod_{l=2}^{j}g\left(
\tau _{l+1}\right) \right) ^{-1}U^{(2)}(\tau _{2}-\tau _{3})\tilde{B}%
_{N,2,3,\tau _{3}}U^{(3)}(\tau _{3}-\tau _{4})\tilde{B}_{N,\mu
_{m}(4),4,\tau _{4}}... \\
&&U^{(j)}(\tau _{j}-\tau _{j+1})\tilde{B}_{N,\mu _{m}(j+1),j+1,\tau
_{j+1}}(f).
\end{eqnarray*}
\end{lemma}

\begin{remark}
There is no difference between Lemma \ref%
{lemma:Klainerman-MachedonBoardGameForBBGKY} and the one we used for the
uniqueness of hierarchy \ref{hierarchy:GPwithTimeFactor} (Lemma \ref%
{lemma:Klainerman-MachedonBoardGameForGP}). We have restated it to remind
the reader of its exact form since we start from $u_{N}^{(2)}$ here.
\end{remark}

With the above lemma and the collapsing estimate (Theorem \ref{Theorem:3*3d}%
), we have the following relation, which is essentially part of the proof of
Theorem \ref{Theorem:Uniqueness of GP}, to help establishing estimates \ref%
{estimate:FreePartofBBGKY}, \ref{estimate:PotentialPartofBBGKY}, and \ref%
{estimate:InteractionPartofBBGKY}.%
\begin{eqnarray}
&&\int_{0}^{T}\left\Vert R^{(1)}\tilde{B}_{N,1,2,\tau _{2}}\int_{D}J_{N}(%
\underline{\tau }_{j+1},\mu _{m})(f)d\underline{\tau }_{j+1}\right\Vert
_{L^{2}}d\tau _{2}  \label{Relation:IteratngCollapsingEstimate} \\
&=&\int_{0}^{T}\left\Vert \int_{D}\left( \dprod_{l=2}^{j}g\left( \tau
_{l+1}\right) \right) ^{-1}R^{(1)}\tilde{B}_{N,1,2,\tau _{2}}U^{(2)}(\tau
_{2}-\tau _{3})\tilde{B}_{N,2,3,\tau _{3}}...d\tau _{3}...d\tau
_{j+1}\right\Vert _{L^{2}}d\tau _{2}  \notag \\
&\leqslant &\left( \inf\limits_{\tau \in \lbrack 0,\frac{\tan \omega T_{0}}{%
\omega }]}g(\tau )\right) ^{-j+1}\int_{\left[ 0,T\right] ^{j}}\left\Vert
R^{(1)}\tilde{B}_{N,1,2,\tau _{2}}U^{(2)}(\tau _{2}-\tau _{3})\tilde{B}%
_{N,2,3,\tau _{3}}...\right\Vert _{L^{2}}d\tau _{2}d\tau _{3}...d\tau _{j+1}
\notag \\
&\leqslant &C^{j-1}T^{\frac{1}{2}}\int_{\left[ 0,T\right] ^{j-1}}\left( \int
\left\Vert R^{(1)}\tilde{B}_{N,1,2,\tau _{2}}U^{(2)}(\tau _{2}-\tau _{3})%
\tilde{B}_{N,2,3,\tau _{3}}...\right\Vert _{L^{2}}^{2}d\tau _{2}\right) ^{%
\frac{1}{2}}d\tau _{3}...d\tau _{j+1}  \notag \\
&&\left( \text{Cauchy-Schwarz}\right)  \notag \\
&\leqslant &C^{j-1}\left( CT^{\frac{1}{2}}\right) \int_{\left[ 0,T\right]
^{j-1}}\left\Vert R^{(2)}\tilde{B}_{N,2,3,\tau _{3}}U^{(3)}(\tau _{3}-\tau
_{4})...\right\Vert d\tau _{3}...d\tau _{j+1}\text{ }\left( Theorem\text{ }%
\ref{Theorem:3*3d}\right)  \notag \\
&&(\text{Iterate }j-2\text{ times})  \notag \\
&&...  \notag \\
&\leqslant &(CT^{\frac{1}{2}})^{j-1}\int_{0}^{T}\left\Vert R^{(j)}\tilde{B}%
_{N,\mu _{m}(j+1),j+1,\tau _{j+1}}f\right\Vert _{L^{2}}d\tau _{j+1}.  \notag
\end{eqnarray}

We show estimate \ref{estimate:FreePartofBBGKY} in Section \ref%
{SubSection:FreePartOfBBGKY}. Assuming Proposition \ref%
{proposition:BoundingThePotentialTerms}, whose proof is postponed to Section %
\ref{SubSection:Proof of the potential terms bound}, we derive estimate \ref%
{estimate:PotentialPartofBBGKY} in Section \ref%
{SubSection:PotentialPartOfBBGKY}. Finally, by taking the coupling level $k$
to be $\ln N$, we check estimate \ref{estimate:InteractionPartofBBGKY} in
Section \ref{SubSection:InteractionPartOfBBGKY}. We remark that the proof of
estimates \ref{estimate:FreePartofBBGKY} and \ref%
{estimate:PotentialPartofBBGKY} is independent of the choice of the coupling
level $k.$ Taking the coupling level $k$ to be $\ln N$ in the estimate of
the interaction part is exactly the place where we follow the original idea
of Chen and Pavlovi\'{c} \cite{TChenAndNPSpace-Time}. Moreover, the proof of
estimate \ref{estimate:PotentialPartofBBGKY} (Section \ref%
{SubSection:PotentialPartOfBBGKY}) is the only place which relies on $\beta
\in \left( 0,\frac{2}{7}\right] $ in this paper.

\subsubsection{Estimate of the Free Part of $u_{N}^{(2)}\label%
{SubSection:FreePartOfBBGKY}$}

Applying Lemma \ref{lemma:Klainerman-MachedonBoardGameForBBGKY} and relation %
\ref{Relation:IteratngCollapsingEstimate} to the free part of $u_{N}^{(2)}$,
we obtain%
\begin{eqnarray*}
&&\int_{0}^{T}\left\Vert R^{(1)}\tilde{B}_{N,1,2,\tau
_{2}}FreePart^{(k)}(\tau _{2})\right\Vert _{L^{2}}d\tau _{2} \\
&\leqslant &CT^{\frac{1}{2}}\left\Vert R^{(2)}u_{N,0}^{(2)}\right\Vert
_{L^{2}}+\sum_{j=3}^{k}\sum_{m}\int_{0}^{T}\left\Vert R^{(1)}\tilde{B}%
_{N,1,2,\tau _{2}}\int_{D}J_{N}(\underline{\tau }_{j},\mu _{m})(U^{(j)}(\tau
_{j})u_{N,0}^{(j)})d\underline{\tau }_{j}\right\Vert _{L^{2}}d\tau _{2} \\
&\leqslant &CT^{\frac{1}{2}}\left\Vert R^{(2)}u_{N,0}^{(2)}\right\Vert
_{L^{2}}+\sum_{j=3}^{k}\sum_{m}C(CT^{\frac{1}{2}})^{j-2}\int_{0}^{T}\left%
\Vert R^{(j-1)}\tilde{B}_{N,\mu _{m}(j),j,\tau _{j}}U^{(j)}(\tau
_{j})u_{N,0}^{(j)}\right\Vert _{L^{2}}d\tau _{j}. \\
&\leqslant &CT^{\frac{1}{2}}\left\Vert R^{(2)}u_{N,0}^{(2)}\right\Vert
_{L^{2}}+C\sum_{j=3}^{k}4^{j-2}(CT^{\frac{1}{2}})^{j-1}\left\Vert
R^{(j)}u_{N,0}^{(j)}\right\Vert _{L^{2}} \\
&\leqslant &CT^{\frac{1}{2}}\left\Vert R^{(2)}u_{N,0}^{(2)}\right\Vert
_{L^{2}}+C\sum_{j=3}^{\infty }(CT^{\frac{1}{2}})^{j-1}C^{j}\text{ }(\text{%
Condition \ref{condition:EnergyBoundForBBGKY}}) \\
&\leqslant &C<\infty \text{ for }T\text{ small enough.}
\end{eqnarray*}%
Whence, we have shown estimate \ref{estimate:FreePartofBBGKY}.

\subsubsection{Estimate of the Potential Part of $u_{N}^{(2)}$\label%
{SubSection:PotentialPartOfBBGKY}}

We have%
\begin{eqnarray*}
&&\int_{0}^{T}\left\Vert R^{(1)}\tilde{B}_{N,1,2,\tau
_{2}}PotentialPart^{(k)}(\tau _{2})\right\Vert _{L^{2}}d\tau _{2} \\
&\leqslant &\int_{0}^{T}\left\Vert \int_{0}^{\tau _{2}}R^{(1)}\tilde{B}%
_{N,1,2,\tau _{2}}U^{(2)}(\tau _{2}-\tau _{3})\frac{\tilde{V}_{N,\tau
_{3}}^{(2)}u_{N}^{(2)}(\tau _{3})}{g\left( \tau _{3}\right) }d\tau
_{3}\right\Vert _{L^{2}}d\tau _{2} \\
&&+\sum_{j=3}^{k}\sum_{m}\int_{0}^{T}\left\Vert R^{(1)}\tilde{B}_{N,1,2,\tau
_{2}}\int_{D}J_{N}(\underline{\tau }_{j},\mu _{m})(\int_{0}^{\tau
_{j}}U^{(j)}(t_{j}-t_{j+1})\frac{\tilde{V}_{N,\tau
_{j+1}}^{(j)}u_{N}^{(j)}(\tau _{j+1})}{g\left( \tau _{j+1}\right) }d\tau
_{j+1})d\underline{\tau }_{j}\right\Vert _{L^{2}}d\tau _{2},
\end{eqnarray*}%
thus same procedure in Section \ref{SubSection:FreePartOfBBGKY} deduces,%
\begin{eqnarray*}
&\leqslant &CT^{\frac{1}{2}}\int_{0}^{T}\left\Vert R^{(2)}\tilde{V}_{N,\tau
_{3}}^{(2)}u_{N}^{(2)}(\tau _{3})\right\Vert _{L^{2}}d\tau _{3} \\
&&+\sum_{j=3}^{k}\sum_{m}C(CT^{\frac{1}{2}})^{j-2}\int_{0}^{T}\left\Vert
R^{(j-1)}\tilde{B}_{N,\mu _{m}(j),j,\tau _{j}}\int_{0}^{\tau
_{j}}U^{(j)}(\tau _{j}-\tau _{j+1})\frac{\tilde{V}_{N,\tau
_{j+1}}^{(j)}u_{N}^{(j)}(\tau _{j+1})}{g\left( \tau _{j+1}\right) }d\tau
_{j+1}\right\Vert _{L^{2}}d\tau _{j} \\
&\leqslant &CT^{\frac{1}{2}}\int_{0}^{T}\left\Vert R^{(2)}\tilde{V}_{N,\tau
_{3}}^{(2)}u_{N}^{(2)}(\tau _{3})\right\Vert _{L^{2}}d\tau
_{3}+C\sum_{j=3}^{k}4^{j-2}(CT^{\frac{1}{2}})^{j-1}\left(
\int_{0}^{T}\left\Vert R^{(j)}\tilde{V}_{N,\tau
_{j+1}}^{(j)}u_{N}^{(j)}(\tau _{j+1})\right\Vert _{L^{2}}d\tau _{j+1}\right)
.
\end{eqnarray*}%
Assume for the moment that we have the estimate%
\begin{equation*}
\int_{0}^{T}\left\Vert R^{(k)}\tilde{V}_{N,\tau }^{(k)}u_{N}^{(k)}(\tau
)\right\Vert _{L^{2}}d\tau \leqslant C_{0}C^{k}T
\end{equation*}%
where $C$ and $C_{0}$ are independent of $T$, $k$ and $N$, then%
\begin{equation*}
\int_{0}^{T}\left\Vert R^{(1)}\tilde{B}_{N,1,2,\tau
_{2}}PotentialPart^{(k)}(\tau _{2})\right\Vert _{L^{2}}d\tau _{2}\leqslant
C<\infty .
\end{equation*}%
for a sufficiently small $T$ and a $C$ independent of $k$ and $N$. As a
result, we complete the proof of estimate \ref{estimate:PotentialPartofBBGKY}
with the following proposition.

\begin{proposition}
\label{proposition:BoundingThePotentialTerms}Assume $\beta \in \left( 0,%
\frac{2}{7}\right] $ and $V\ $is a nonnegative $L^{1}(\mathbb{R}^{3})\cap
H^{2}(\mathbb{R}^{3})$ function. Then, given $T\in \left[ 0,\frac{\tan
\omega T_{0}}{\omega }\right] $, there are $C$ and $C_{0}$ independent of $T$%
, $k$ and $N$ such that 
\begin{equation*}
\int_{0}^{T}\left\Vert R^{(k)}\tilde{V}_{N,\tau }^{(k)}u_{N}^{(k)}(\tau
)\right\Vert _{L^{2}}d\tau \leqslant C_{0}C^{k}T.
\end{equation*}
\end{proposition}

\begin{proof}
The proof is elementary and we relegate it to Section \ref{SubSection:Proof
of the potential terms bound}.
\end{proof}

\begin{remark}
This proposition is exactly the reason we restrict $\beta \in \left( 0,\frac{%
2}{7}\right] $ in this paper.
\end{remark}

\subsubsection{Estimate of the Interaction Part of $u_{N}^{(2)}$\label%
{SubSection:InteractionPartOfBBGKY}}

We proceed like Sections \ref{SubSection:FreePartOfBBGKY} and \ref%
{SubSection:PotentialPartOfBBGKY}.%
\begin{eqnarray*}
&&\int_{0}^{T}\left\Vert R^{(1)}\tilde{B}_{N,1,2,\tau
_{2}}InteractionPart^{(k)}(\tau _{2})\right\Vert _{L^{2}}d\tau _{2} \\
&\leqslant &\sum_{m}\int_{0}^{T}\left\Vert R^{(1)}\tilde{B}_{N,1,2,\tau
_{2}}\int_{D}J_{N}(\underline{\tau }_{k+1},\mu _{m})(u_{N}^{(k+1)}(\tau
_{k+1}))d\underline{\tau }_{k+1}\right\Vert _{L^{2}}d\tau _{2} \\
&\leqslant &\sum_{m}C(CT^{\frac{1}{2}})^{k-1}\int_{0}^{T}\left\Vert R^{(k)}%
\tilde{B}_{N,\mu _{m}(k+1),k+1,\tau _{k+1}}u_{N}^{(k+1)}(\tau
_{k+1})\right\Vert _{L^{2}}d\tau _{k+1}.
\end{eqnarray*}%
Then the next step is to investigate%
\begin{equation*}
\int_{0}^{T}\left\Vert R^{(k)}\tilde{B}_{N,\mu _{m}(k+1),k+1,\tau
_{k+1}}u_{N}^{(k+1)}(\tau _{k+1})\right\Vert _{L^{2}}d\tau _{k+1}.
\end{equation*}%
Without loss of generality, set $\mu _{m}(k+1)=1$ and look at $\tilde{B}%
_{N,1,k+1,\tau _{k+1}}^{1}$, we have%
\begin{eqnarray*}
&&\int \left\vert R^{(k)}\tilde{B}_{N,1,k+1,\tau
_{k+1}}^{1}u_{N}^{(k+1)}(\tau _{k+1})\right\vert ^{2}d\mathbf{y}_{k}d\mathbf{%
y}_{k}^{\prime } \\
&=&\int \left\vert R^{(k)}\int \tilde{V}_{N,\tau
_{k+1}}(y_{1}-y_{k+1})u_{N}^{(k+1)}(\tau _{k+1},\mathbf{y}_{k},y_{k+1};%
\mathbf{y}_{k}^{\prime },y_{k+1})dy_{k+1}\right\vert ^{2}d\mathbf{y}_{k}d%
\mathbf{y}_{k}^{\prime } \\
&\leqslant &C\int \left\vert \int \left( \tilde{V}_{N,\tau _{k+1}}\right)
^{\prime }(y_{1}-y_{k+1})\left( \dprod_{j=2}^{k}\left\vert \nabla
_{y_{j}}\right\vert \right) \left( \dprod_{j=1}^{k}\left\vert \nabla
_{y_{j}^{\prime }}\right\vert \right) u_{N}^{(k+1)}(\tau _{k+1},\mathbf{y}%
_{k},y_{k+1};\mathbf{y}_{k}^{\prime },y_{k+1})dy_{k+1}\right\vert ^{2}d%
\mathbf{y}_{k}d\mathbf{y}_{k}^{\prime } \\
&&+C\int \left\vert \int \tilde{V}_{N,\tau
_{k+1}}(y_{1}-y_{k+1})R^{(k)}u_{N}^{(k+1)}(\tau _{k+1},\mathbf{y}%
_{k},y_{k+1};\mathbf{y}_{k}^{\prime },y_{k+1})dy_{k+1}\right\vert ^{2}d%
\mathbf{y}_{k}d\mathbf{y}_{k}^{\prime } \\
&=&C(I+II).
\end{eqnarray*}%
Noticing that $\left\Vert \tilde{V}_{\tau }\right\Vert _{H^{2}}\leqslant
C\left\Vert V\right\Vert _{H^{2}}$ uniformly for $\tau \in \left[ 0,\frac{%
\tan \omega T_{0}}{\omega }\right] $, we can then estimate%
\begin{eqnarray*}
&&I \\
&=&\int \left\vert \int \left( \tilde{V}_{N,\tau _{k+1}}\right) ^{\prime
}(y_{1}-y_{k+1})\left( \dprod_{j=2}^{k}\left\vert \nabla _{y_{j}}\right\vert
\right) \left( \dprod_{j=1}^{k}\left\vert \nabla _{y_{j}^{\prime
}}\right\vert \right) u_{N}^{(k+1)}(\tau _{k+1},\mathbf{y}_{k},y_{k+1};%
\mathbf{y}_{k}^{\prime },y_{k+1})dy_{k+1}\right\vert ^{2}d\mathbf{y}_{k}d%
\mathbf{y}_{k}^{\prime } \\
&\leqslant &\int d\mathbf{y}_{k}d\mathbf{y}_{k}^{\prime }\left( \int
\left\vert \left( \tilde{V}_{N,\tau _{k+1}}\right) ^{\prime
}(y_{1}-y_{k+1})\right\vert ^{2}dy_{k+1}\right) \\
&&\times \left( \int \left\vert \left( \dprod_{j=2}^{k}\left\vert \nabla
_{y_{j}}\right\vert \right) \left( \dprod_{j=1}^{k}\left\vert \nabla
_{y_{j}^{\prime }}\right\vert \right) u_{N}^{(k+1)}(\tau _{k+1},\mathbf{y}%
_{k},y_{k+1};\mathbf{y}_{k}^{\prime },y_{k+1})\right\vert
^{2}dy_{k+1}\right) \text{ }(\text{Cauchy-Schwarz}) \\
&\leqslant &CN^{5\beta }\left\Vert V^{\prime }\right\Vert _{L^{2}}^{2}\int
\left( \int \left\vert \left( \dprod_{j=2}^{k}\left\vert \nabla
_{y_{j}}\right\vert \right) \left( \dprod_{j=1}^{k}\left\vert \nabla
_{y_{j}^{\prime }}\right\vert \right) u_{N}^{(k+1)}(\tau _{k+1},\mathbf{y}%
_{k},y_{k+1};\mathbf{y}_{k}^{\prime },y_{k+1})\right\vert
^{2}dy_{k+1}\right) d\mathbf{y}_{k}d\mathbf{y}_{k}^{\prime } \\
&\leqslant &CN^{5\beta }\left\Vert V^{\prime }\right\Vert _{L^{2}}^{2}\int d%
\mathbf{y}_{k}d\mathbf{y}_{k}^{\prime } \\
&&\times \left( \int \left\vert \left( 1-\triangle _{y_{k+1}}\right) ^{\frac{%
1}{2}}\left( 1-\triangle _{y_{k+1}^{\prime }}\right) ^{\frac{1}{2}}\left(
\dprod_{j=2}^{k}\left\vert \nabla _{y_{j}}\right\vert \right) \left(
\dprod_{j=1}^{k}\left\vert \nabla _{y_{j}^{\prime }}\right\vert \right)
u_{N}^{(k+1)}(\tau _{k+1},\mathbf{y}_{k},y_{k+1};\mathbf{y}_{k}^{\prime
},y_{k+1}^{\prime })\right\vert ^{2}dy_{k+1}dy_{k+1}^{\prime }\right) \\
&&(\text{Trace Theorem}) \\
&\leqslant &CN^{5\beta }\left\Vert V^{\prime }\right\Vert _{L^{2}}^{2}C^{k+1}%
\text{ }(\text{Condition \ref{condition:EnergyBoundForBBGKY}})
\end{eqnarray*}%
and%
\begin{eqnarray*}
II &=&\int \left\vert \int \tilde{V}_{N,\tau
_{k+1}}(y_{1}-y_{k+1})R^{(k)}u_{N}^{(k+1)}(\tau _{k+1},\mathbf{y}%
_{k},y_{k+1};\mathbf{y}_{k}^{\prime },y_{k+1})dy_{k+1}\right\vert ^{2}d%
\mathbf{y}_{k}d\mathbf{y}_{k}^{\prime } \\
&\leqslant &CN^{3\beta }\left\Vert V\right\Vert _{L^{2}}^{2}C^{k+1}\text{ }(%
\text{Same method as }I).
\end{eqnarray*}%
Accordingly,%
\begin{equation*}
\int \left\vert R^{(k)}\tilde{B}_{N,1,k+1,\tau _{k+1}}^{1}u_{N}^{(k+1)}(\tau
_{k+1})\right\vert ^{2}d\mathbf{y}_{k}d\mathbf{y}_{k}^{\prime }\leqslant
CN^{5\beta }\left\Vert V\right\Vert _{H^{2}}^{2}C^{k+1}.
\end{equation*}

\begin{remark}
The estimates of $I$ and $II$ may not be optimal. But they are good enough
for proving estimate \ref{estimate:InteractionPartofBBGKY} for arbitrary $%
\beta >0$.
\end{remark}

Thence%
\begin{eqnarray*}
&&\int_{0}^{T}\left\Vert R^{(1)}\tilde{B}_{N,1,2.\tau
_{2}}InteractionPart^{(k)}(\tau _{2})\right\Vert _{L^{2}}d\tau _{2} \\
&\leqslant &\sum_{m}C(CT^{\frac{1}{2}})^{k-1}\int_{0}^{T}\left\Vert R^{(k)}%
\tilde{B}_{N,\mu _{m}(k+1),k+1,\tau _{k+1}}u_{N}^{(k+1)}(\tau
_{k+1})\right\Vert _{L^{2}}d\tau _{k+1}. \\
&\leqslant &4^{k-1}C(CT^{\frac{1}{2}})^{k-1}T\left( CN^{\frac{5\beta }{2}%
}\left\Vert V\right\Vert _{H^{2}}C^{k+1}\right) \\
&\leqslant &C\left\Vert V\right\Vert _{H^{2}}(T^{\frac{1}{2}})^{k+2}N^{\frac{%
5\beta }{2}}C^{k}.
\end{eqnarray*}%
Take the coupling level $k=\ln N$, we have%
\begin{equation*}
\int_{0}^{T}\left\Vert R^{(1)}\tilde{B}_{N,1,2.\tau
_{2}}InteractionPart^{(k)}(\tau _{2})\right\Vert _{L^{2}}d\tau _{2}\leqslant
C\left\Vert V\right\Vert _{H^{2}}(T^{\frac{1}{2}})^{2+\ln N}N^{\frac{5\beta 
}{2}}N^{c}.
\end{equation*}%
Selecting $T$ such that 
\begin{equation*}
T\leqslant e^{-(5\beta +2C)}
\end{equation*}%
ensures that%
\begin{equation*}
(T^{\frac{1}{2}})^{\ln N}N^{\frac{5\beta }{2}}N^{c}\leqslant 1
\end{equation*}%
and thence%
\begin{equation*}
\int_{0}^{T}\left\Vert R^{(1)}\tilde{B}_{N,1,2.\tau
_{2}}InteractionPart^{(k)}(\tau _{2})\right\Vert _{L^{2}}d\tau _{2}\leqslant
C
\end{equation*}%
where $C$ is independent of $N.$

We remind the reader that, at this point, we have obtained estimates \ref%
{estimate:FreePartofBBGKY}, \ref{estimate:PotentialPartofBBGKY}, and \ref%
{estimate:InteractionPartofBBGKY} for a sufficiently small $T$ determined by
the controlling constant in condition \ref{condition:EnergyBoundForBBGKY}
and independent of $N.$ Thus one can repeat the argument to acquire the
estimates for any finite time $T\in \left[ 0,\frac{\tan \omega T_{0}}{\omega 
}\right] $ through bootstrapping and condition \ref%
{condition:EnergyBoundForBBGKY}. Whence, we have earned estimate \ref%
{estimate:MainEstimateOfSpace-TimeBoundSection} and established Theorem \ref%
{Theorem:Space-Time Bound for GlensBBGKY}. The rest of this section is the
proof of Proposition \ref{proposition:BoundingThePotentialTerms}.

\subsubsection{Proof of Proposition \protect\ref%
{proposition:BoundingThePotentialTerms} \label{SubSection:Proof of the
potential terms bound}}

We will utilize the lemma.

\begin{lemma}
\label{Lemma:ESYIntergral}\cite{E-S-Y2}%
\begin{equation*}
\int V(x_{1}-x_{2})\left\vert f(x_{1},x_{2})\right\vert
^{2}dx_{1}dx_{2}\leqslant C\left\Vert V\right\Vert _{L^{1}}\int \left\vert
\left( 1-\triangle _{x_{1}}\right) ^{\frac{1}{2}}\left( 1-\triangle
_{x_{2}}\right) ^{\frac{1}{2}}f(x_{1},x_{2})\right\vert ^{2}dx_{1}dx_{2}.
\end{equation*}%
In particular,%
\begin{equation*}
\int \left\vert V(x_{1}-x_{2})\right\vert ^{2}\left\vert
f(x_{1},x_{2})\right\vert ^{2}dx_{1}dx_{2}\leqslant C\left\Vert V\right\Vert
_{L^{2}}^{2}\int \left\vert \left( 1-\triangle _{x_{1}}\right) ^{\frac{1}{2}%
}\left( 1-\triangle _{x_{2}}\right) ^{\frac{1}{2}}f(x_{1},x_{2})\right\vert
^{2}dx_{1}dx_{2}.
\end{equation*}
\end{lemma}

\begin{proof}
This is Lemma A.3 in \cite{E-S-Y2}.
\end{proof}

Without loss of generality, we show Proposition \ref%
{proposition:BoundingThePotentialTerms} for $k=2$ which corresponds to%
\begin{equation}
\int_{0}^{T}\left\Vert \left( \left\vert \nabla _{y_{1}}\right\vert
\left\vert \nabla _{y_{2}}\right\vert \left\vert \nabla _{y_{1}^{\prime
}}\right\vert \left\vert \nabla _{y_{2}^{^{\prime }}}\right\vert \right)
\left( \frac{\tilde{V}_{N,\tau }(y_{1}-y_{2})}{N}u_{N}^{(2)}(\tau ,\mathbf{y}%
_{2};\mathbf{y}_{2}^{\prime })\right) \right\Vert _{L^{2}}d\tau \leqslant
C_{0}C^{2}T  \label{estimate:one term in potential part}
\end{equation}%
and%
\begin{equation*}
\int_{0}^{T}\left\Vert \left( \left\vert \nabla _{y_{1}}\right\vert
\left\vert \nabla _{y_{2}}\right\vert \left\vert \nabla _{y_{1}^{\prime
}}\right\vert \left\vert \nabla _{y_{2}^{^{\prime }}}\right\vert \right)
\left( \frac{\tilde{V}_{N,\tau }(y_{1}^{\prime }-y_{2}^{\prime })}{N}%
u_{N}^{(2)}(\tau ,\mathbf{y}_{2};\mathbf{y}_{2}^{\prime })\right)
\right\Vert _{L^{2}}d\tau \leqslant C_{0}C^{2}T.
\end{equation*}%
By similarity we only prove estimate \ref{estimate:one term in potential
part}. For a general $k,$ there are $2k^{2}$ terms in $\tilde{V}_{N,\tau
}^{(k)},$ whence $R_{N}^{(k)}\tilde{V}_{N,\tau }^{(k)}\gamma _{N}^{(k)}$ has
about $8k^{2}$ terms by Leibniz's rule. Since $8k^{2}$ can be absorbed into $%
C^{k},$ the method here applies.

First of all, $\beta \in \left( 0,\frac{2}{7}\right] $ implies the following
properties of $\tilde{V}_{N,\tau }/N:$%
\begin{eqnarray*}
\left\Vert \tilde{V}_{N,\tau }/N\right\Vert _{L^{\infty }} &=&N^{3\beta
-1}\left\Vert \tilde{V}_{\tau }\right\Vert _{L^{\infty }}<N^{-\frac{1}{7}%
}\left\Vert \tilde{V}_{\tau }\right\Vert _{L^{\infty }}, \\
\left\Vert \left( \tilde{V}_{N,\tau }\right) ^{\prime }/N\right\Vert
_{L^{p}} &=&N^{4\beta -1-\frac{3\beta }{p}}\left\Vert \left( \tilde{V}_{\tau
}\right) ^{\prime }\right\Vert _{L^{p}}=N^{-\left( \frac{6}{7p}-\frac{1}{7}%
+\left( 4-\frac{3}{p}\right) \varepsilon _{0}\right) }\left\Vert \left( 
\tilde{V}_{\tau }\right) ^{\prime }\right\Vert _{L^{p}},\text{ decays up to }%
p=6, \\
\left\Vert \left( \tilde{V}_{N,\tau }\right) ^{\prime \prime }/N\right\Vert
_{L^{2}} &\leqslant &N^{-\frac{7}{2}\varepsilon _{0}}\left\Vert \left( 
\tilde{V}_{\tau }\right) ^{\prime \prime }\right\Vert _{L^{2}}.
\end{eqnarray*}%
where $\varepsilon _{0}$ is $\left( \frac{2}{7}-\beta \right) \geqslant 0$.
On the one hand, $\left\Vert \tilde{V}_{\tau }\right\Vert _{H^{2}}\leqslant
C\left\Vert V\right\Vert _{H^{2}}$ uniformly for $\tau \in \left[ 0,\frac{%
\tan \omega T_{0}}{\omega }\right] $. On the other hand we assume $V\ $is a
nonnegative $L^{1}(\mathbb{R}^{3})\cap H^{2}(\mathbb{R}^{3})$ function. Thus
we know $\tilde{V}_{\tau }\in L^{\infty }$ and $\left( \tilde{V}_{\tau
}\right) ^{\prime }\in L^{6}.$

We compute%
\begin{eqnarray*}
&&\left( \left\vert \nabla _{y_{1}}\right\vert \left\vert \nabla
_{y_{2}}\right\vert \left\vert \nabla _{y_{1}^{\prime }}\right\vert
\left\vert \nabla _{y_{2}^{^{\prime }}}\right\vert \right) \left( N^{-1}%
\tilde{V}_{N,\tau }(y_{1}-y_{2})u_{N}^{(2)}(\tau ,\mathbf{y}_{2};\mathbf{y}%
_{2}^{\prime })\right) \\
&=&\left( \left\vert \nabla _{y_{1}}\right\vert \left\vert \nabla
_{y_{2}}\right\vert N^{-1}\tilde{V}_{N,\tau }(y_{1}-y_{2})\right) \left(
\left\vert \nabla _{y_{1}^{\prime }}\right\vert \left\vert \nabla
_{y_{2}^{^{\prime }}}\right\vert u_{N}^{(2)}(\tau ,\mathbf{y}_{2};\mathbf{y}%
_{2}^{\prime })\right) \\
&&+\left( \left\vert \nabla _{y_{1}}\right\vert N^{-1}\tilde{V}_{N,\tau
}(y_{1}-y_{2})\right) \left( \left\vert \nabla _{y_{2}}\right\vert
\left\vert \nabla _{y_{1}^{\prime }}\right\vert \left\vert \nabla
_{y_{2}^{^{\prime }}}\right\vert u_{N}^{(2)}(\tau ,\mathbf{y}_{2};\mathbf{y}%
_{2}^{\prime })\right) \\
&&+\left( \left\vert \nabla _{y_{2}}\right\vert N^{-1}\tilde{V}_{N,\tau
}(y_{1}-y_{2})\right) \left( \left\vert \nabla _{y_{1}}\right\vert
\left\vert \nabla _{y_{1}^{\prime }}\right\vert \left\vert \nabla
_{y_{2}^{^{\prime }}}\right\vert u_{N}^{(2)}(\tau ,\mathbf{y}_{2};\mathbf{y}%
_{2}^{\prime })\right) \\
&&+N^{-1}\tilde{V}_{N,\tau }(y_{1}-y_{2})\left( \left\vert \nabla
_{y_{1}}\right\vert \left\vert \nabla _{y_{2}}\right\vert \left\vert \nabla
_{y_{1}^{\prime }}\right\vert \left\vert \nabla _{y_{2}^{^{\prime
}}}\right\vert \right) u_{N}^{(2)}(\tau ,\mathbf{y}_{2};\mathbf{y}%
_{2}^{\prime }).
\end{eqnarray*}%
But%
\begin{eqnarray*}
&&\sup_{\tau }\left\Vert \left( \left\vert \nabla _{y_{1}}\right\vert
\left\vert \nabla _{y_{2}}\right\vert N^{-1}\tilde{V}_{N,\tau
}(y_{1}-y_{2})\right) \left( \left\vert \nabla _{y_{1}^{\prime }}\right\vert
\left\vert \nabla _{y_{2}^{^{\prime }}}\right\vert u_{N}^{(2)}(\tau ,\mathbf{%
y}_{2};\mathbf{y}_{2}^{\prime })\right) \right\Vert _{L^{2}}^{2} \\
&=&\sup_{\tau }\int \left\vert N^{-1}\left( \tilde{V}_{N,\tau }\right)
^{\prime \prime }(y_{1}-y_{2})\right\vert ^{2}\left\vert \left\vert \nabla
_{y_{1}^{\prime }}\right\vert \left\vert \nabla _{y_{2}^{^{\prime
}}}\right\vert u_{N}^{(2)}(\tau ,\mathbf{y}_{2};\mathbf{y}_{2}^{\prime
})\right\vert ^{2}d\mathbf{y}_{2}d\mathbf{y}_{2}^{\prime } \\
&\leqslant &\sup_{\tau }\int \left( \int \left\vert N^{-1}\left( \tilde{V}%
_{N,\tau }\right) ^{\prime \prime }(y)\right\vert ^{2}dy\right) \left( \int
\left\vert \left( 1-\triangle _{y_{1}}\right) ^{\frac{1}{2}}\left(
1-\triangle _{y_{2}}\right) ^{\frac{1}{2}}\left\vert \nabla _{y_{1}^{\prime
}}\right\vert \left\vert \nabla _{y_{2}^{^{\prime }}}\right\vert
u_{N}^{(2)}(\tau ,\mathbf{y}_{2};\mathbf{y}_{2}^{\prime })\right\vert ^{2}d%
\mathbf{y}_{2}\right) d\mathbf{y}_{2}^{\prime } \\
&&(\text{Lemma }\ref{Lemma:ESYIntergral}) \\
&\leqslant &CN^{-7\varepsilon _{0}}\left\Vert V^{\prime \prime }\right\Vert
_{L^{2}}^{2}\sup_{\tau }\int \left\vert \left( 1-\triangle _{y_{1}}\right) ^{%
\frac{1}{2}}\left( 1-\triangle _{y_{2}}\right) ^{\frac{1}{2}}\left(
1-\triangle _{y_{1}^{\prime }}\right) ^{\frac{1}{2}}\left( 1-\triangle
_{y_{2}^{\prime }}\right) ^{\frac{1}{2}}u_{N}^{(2)}(\tau ,\mathbf{y}_{2};%
\mathbf{y}_{2}^{\prime })\right\vert ^{2}d\mathbf{y}_{2}d\mathbf{y}%
_{2}^{\prime } \\
&\leqslant &N^{-7\varepsilon _{0}}\left\Vert V^{\prime \prime }\right\Vert
_{L^{2}}^{2}C^{2}\text{ }(\text{Condition \ref{condition:EnergyBoundForBBGKY}%
})
\end{eqnarray*}%
and%
\begin{eqnarray*}
&&\sup_{\tau }\left\Vert N^{-1}\tilde{V}_{N,\tau }(y_{1}-y_{2})\left(
\left\vert \nabla _{y_{1}}\right\vert \left\vert \nabla _{y_{2}}\right\vert
\left\vert \nabla _{y_{1}^{\prime }}\right\vert \left\vert \nabla
_{y_{2}^{^{\prime }}}\right\vert \right) u_{N}^{(2)}(\tau ,\mathbf{y}_{2};%
\mathbf{y}_{2}^{\prime })\right\Vert _{L^{2}}^{2} \\
&=&\sup_{\tau }\int \left\vert N^{-1}\tilde{V}_{N,\tau
}(y_{1}-y_{2})\right\vert ^{2}\left\vert \left( \left\vert \nabla
_{y_{1}}\right\vert \left\vert \nabla _{y_{2}}\right\vert \left\vert \nabla
_{y_{1}^{\prime }}\right\vert \left\vert \nabla _{y_{2}^{^{\prime
}}}\right\vert \right) u_{N}^{(2)}(\tau ,\mathbf{y}_{2};\mathbf{y}%
_{2}^{\prime })\right\vert ^{2}d\mathbf{y}_{2}d\mathbf{y}_{2}^{\prime } \\
&\leqslant &\sup_{\tau }\int \left\Vert N^{-1}\tilde{V}_{N,\tau }\right\Vert
_{L^{\infty }}^{2}\left\vert \left( 1-\triangle _{y_{1}}\right) ^{\frac{1}{2}%
}\left( 1-\triangle _{y_{2}}\right) ^{\frac{1}{2}}\left( 1-\triangle
_{y_{1}^{\prime }}\right) ^{\frac{1}{2}}\left( 1-\triangle _{y_{2}^{\prime
}}\right) ^{\frac{1}{2}}u_{N}^{(2)}(\tau ,\mathbf{y}_{2};\mathbf{y}%
_{2}^{\prime })\right\vert ^{2}d\mathbf{y}_{2}d\mathbf{y}_{2}^{\prime } \\
&\leqslant &N^{-\frac{2}{7}}\left\Vert V\right\Vert _{L^{\infty
}}^{2}\sup_{\tau }\int \left\vert \left( 1-\triangle _{y_{1}}\right) ^{\frac{%
1}{2}}\left( 1-\triangle _{y_{2}}\right) ^{\frac{1}{2}}\left( 1-\triangle
_{y_{1}^{\prime }}\right) ^{\frac{1}{2}}\left( 1-\triangle _{y_{2}^{\prime
}}\right) ^{\frac{1}{2}}u_{N}^{(2)}(\tau ,\mathbf{y}_{2};\mathbf{y}%
_{2}^{\prime })\right\vert ^{2}d\mathbf{y}_{2}d\mathbf{y}_{2}^{\prime } \\
&\leqslant &CN^{-\frac{2}{7}}\left\Vert V\right\Vert _{H^{2}}^{2}C^{2}\text{ 
}(\text{Sobolev and Condition \ref{condition:EnergyBoundForBBGKY}}).
\end{eqnarray*}%
Since the same method applies to the middle terms, we have obtained estimate %
\ref{estimate:one term in potential part} and hence Proposition \ref%
{proposition:BoundingThePotentialTerms}.

\section{Conclusion}

In this paper, we have rigorously derived the 3D cubic nonlinear Schr\"{o}%
dinger equation with a quadratic trap from the $N$-body linear Schr\"{o}%
dinger equation. The main novelty is that we allow a quadratic trap in our
analysis and the main technical improvements are the simplified proof of the
Klainerman-Machedon space-time bound as compared to non-trap case in Chen
and Pavlovi\'{c} \cite{TChenAndNPSpace-Time}, and the extension of the range
of $\beta $ from $\left( 0,1/4\right) $ in Chen and Pavlovi\'{c} \cite%
{TChenAndNPSpace-Time} to $\left( 0,2/7\right] .$ Compared to the 2D work 
\cite{ChenAnisotropic} which is also by the author, the 3D problem in this
paper is of critical regularity. To explain what we mean by critical: in 2D
one easily obtains the $\left\vert \nabla \right\vert ^{\frac{1}{2}}$%
-space-time bound needed for the uniqueness theorem by a trace theorem; in
3D the only way to obtain the space-time bound \ref%
{formula:theSpace-TimeBound} is through smoothing estimates since one does
not have enough regularity to apply a trace theorem. Thence the key
arguments in 3D are more involved and totally different from the 2D case
which is a subcritical problem. Moreover, we have established the trace norm
convergence in the main theorem which is a stronger result than the
Hilbert-Schmidt norm convergence.

\section{Acknowledgements}

The author would like to first thank his advisors, Professor Matei Machedon
and Professor Manoussos G. Grillakis, for the discussion related to this
work and their guidance on the author's PhD dissertation which consists of 
\cite{Chen2ndOrder, ChenAnisotropic}. It is due to their guidance, the
author became keenly interested in the subject and wrote this paper.

The author is greatly indebted to Professor Thomas Chen and Professor Nata%
\v{s}a Pavlovi\'{c} for telling the author about their important work \cite%
{TChenAndNPSpace-Time} and giving the author many helpful suggestions on
finding his first job during his visit to Austin, and for their very
detailed and helpful comments on this paper.

The author's thanks go to Professor Sergiu Klainerman for the discussion
related to this work, and to Professor Cl\'{e}ment Mouhot for sharing the
physical meaning of the trace norm in this setting with the author.

The author also would like to thank the anonymous referee for many
insightful comments and helpful suggestions.

\section{Appendix I: Proof of Corollary \protect\ref{Corollary:general data
convergence}\label{Section:Appendix-0}}

For the purpose of this Appendix I, we may assume $\omega =0$. Or in other
words, we skip Steps I and IV of the proof of Theorem \ref%
{Theorem:SmoothVersionofTheMainTheorem} here. When the desired limit is an
orthogonal projection, one does not need this proof. This is a functional
analysis argument and all operators mentioned in this Appendix I acts on $%
L^{2}\left( \mathbb{R}^{3k}\right) .$ We prove Corollary \ref%
{Corollary:general data convergence} by verifying the hypothesis of the
following lemma.

\begin{lemma}
\label{Theorem:Simons'}\cite{Simon} Assume the operator sequence $\left\{
A_{n}\right\} $ satisfies that, as bounded operators, $A_{n}\rightharpoonup
A $, $A_{n}^{\ast }\rightharpoonup A^{\ast }$ and $\left\vert
A_{n}\right\vert \rightharpoonup \left\vert A\right\vert $ in the weak
sense. If 
\begin{equation*}
\lim_{N\rightarrow \infty }\limfunc{Tr}\left\vert A_{n}\right\vert =\limfunc{%
Tr}\left\vert A\right\vert ,
\end{equation*}%
then 
\begin{equation*}
\lim_{N\rightarrow \infty }\limfunc{Tr}\left\vert A_{n}-A\right\vert =0.
\end{equation*}
\end{lemma}

\begin{proof}
This is Theorem 2.20 in \cite{Simon}. It implies the Gr\"{u}mm's convergence
theorem (Theorem 2.19 of \cite{Simon}) used in \cite{E-S-Y3}.
\end{proof}

We first observe that condition \ref{condition:high energy bound} implies
the a-priori estimate 
\begin{equation*}
\sup_{t\in \lbrack 0,T]}\limfunc{Tr}\left( \dprod_{j=1}^{k}\left(
1-\triangle _{x_{j}}\right) \right) \gamma _{N}^{(k)}\leqslant C^{k}.
\end{equation*}%
Thus we have the compactness argument and the uniqueness argument to
conclude that, as trace class operator kernels,%
\begin{equation}
\gamma _{N}^{(k)}(t,\mathbf{x}_{k};\mathbf{x}_{k}^{\prime })\rightharpoonup
\gamma ^{(k)}(t,\mathbf{x}_{k};\mathbf{x}_{k}^{\prime })\text{ (weak*)}.
\label{formula:weak*}
\end{equation}

\begin{remark}
Because we assume $\omega =0$, both of the Erd\"{o}s-Schlein-Yau uniqueness
theorem \cite{E-S-Y2} and the Klainerman-Machedon uniqueness theorem \cite%
{KlainermanAndMachedon} apply here. For the general case, one has to use the
main argument in this paper.
\end{remark}

Let $\mathcal{H}_{k}$ be the Hilbert-Schmidt operators on $L^{2}\left( 
\mathbb{R}^{3k}\right) .$ Recall that the test functions for weak*
convergence in $\mathcal{L}_{k}^{1}$ come from $\mathcal{K}_{k}$ and the
test functions for weak* convergence in $\mathcal{H}_{k}$ come from $%
\mathcal{H}_{k}.$ Thus the weak* convergence \ref{formula:weak*} as trace
class operator kernels infers that as Hilbert-Schmidt kernels,%
\begin{equation*}
\gamma _{N}^{(k)}(t,\mathbf{x}_{k};\mathbf{x}_{k}^{\prime })\rightharpoonup
\gamma ^{(k)}(t,\mathbf{x}_{k};\mathbf{x}_{k}^{\prime })\text{ (weak*),}
\end{equation*}%
because $\mathcal{H}_{k}\subset \mathcal{K}_{k}$ i.e. there are fewer test
functions. Since $\mathcal{H}_{k}$ is reflexive, the weak* convergence is no
different from the weak convergence. Thus we know that as Hilbert-Schmidt
kernels and hence as bounded operator kernels,%
\begin{equation*}
\gamma _{N}^{(k)}(t,\mathbf{x}_{k};\mathbf{x}_{k}^{\prime })\rightharpoonup
\gamma ^{(k)}(t,\mathbf{x}_{k};\mathbf{x}_{k}^{\prime })\text{ (weak).}
\end{equation*}%
At this point, we have verified that, as bounded operators, $%
A_{n}\rightharpoonup A$ and $A_{n}^{\ast }\rightharpoonup A^{\ast }$ in the
weak sense since $\gamma _{N}^{(k)}$ and $\gamma ^{(k)}$ are self adjoint.
We now check $\left\vert A_{n}\right\vert \rightharpoonup \left\vert
A\right\vert .$

To check $\left\vert A_{n}\right\vert \rightharpoonup \left\vert
A\right\vert $, one first notices that $\gamma _{N}^{(k)}$ and $\gamma
^{(k)} $ has only real eigenvalues since they are self adjoint. Moreover, $%
\gamma _{N}^{(k)}(t,\mathbf{x}_{k};\mathbf{x}_{k}^{\prime })$ has no
negative eigenvalues by definition \ref{def:marginal density}, in fact, 
\begin{equation*}
\int \left( \int \left( \int \phi (x,z)\bar{\phi}(y,z)dz\right)
f(y)dy\right) \bar{f}(x)dx=\int dz\left\vert \int \phi (x,z)\bar{f}%
(x)dx\right\vert ^{2}\geqslant 0.
\end{equation*}%
Since $f(\mathbf{x}_{k})\bar{f}(\mathbf{x}_{k}^{\prime })$ in the above
estimate is also a Hilbert-Schmidt kernel, the definition of weak
convergence in $L^{2}(d\mathbf{x}_{k}d\mathbf{x}_{k}^{\prime })$ gives%
\begin{equation*}
\int \gamma ^{(k)}(t,\mathbf{x}_{k};\mathbf{x}_{k}^{\prime })\overline{f(%
\mathbf{x}_{k})\bar{f}(\mathbf{x}_{k}^{\prime })}d\mathbf{x}_{k}d\mathbf{x}%
_{k}^{\prime }=\lim_{N\rightarrow \infty }\int \gamma _{N}^{(k)}(t,\mathbf{x}%
_{k};\mathbf{x}_{k}^{\prime })\overline{f(\mathbf{x}_{k})\bar{f}(\mathbf{x}%
_{k}^{\prime })}d\mathbf{x}_{k}d\mathbf{x}_{k}^{\prime }\geqslant 0,
\end{equation*}%
as $\gamma _{N}^{(k)}(t,\mathbf{x}_{k};\mathbf{x}_{k}^{\prime
})\rightharpoonup \gamma ^{(k)}(t,\mathbf{x}_{k};\mathbf{x}_{k}^{\prime })$
weakly in $L^{2}(d\mathbf{x}_{k}d\mathbf{x}_{k}^{\prime }).$ So we have
checked $\left\vert A_{n}\right\vert \rightharpoonup \left\vert A\right\vert 
$ because $\left\vert A_{n}\right\vert =A_{n}$ and $\left\vert A\right\vert
=A.$

To prove Corollary \ref{Corollary:general data convergence}, by Lemma \ref%
{Theorem:Simons'}, it remains to argue%
\begin{equation*}
\lim_{N\rightarrow \infty }\limfunc{Tr}\left\vert \gamma _{N}^{(k)}(t,%
\mathbf{x}_{k};\mathbf{x}_{k}^{\prime })\right\vert =\limfunc{Tr}\left\vert
\gamma ^{(k)}(t,\mathbf{x}_{k};\mathbf{x}_{k}^{\prime })\right\vert .
\end{equation*}%
Notice that we have the conservation of trace%
\begin{eqnarray*}
\int \gamma _{N}^{(k)}(t,\mathbf{x}_{k};\mathbf{x}_{k})d\mathbf{x}_{k}
&=&\int \gamma _{N}^{(k)}(0,\mathbf{x}_{k};\mathbf{x}_{k})d\mathbf{x}_{k}%
\text{ (By definition \ref{def:marginal density})} \\
\int \gamma ^{(k)}(t,\mathbf{x}_{k};\mathbf{x}_{k})d\mathbf{x}_{k} &=&\int
\gamma ^{(k)}(0,\mathbf{x}_{k};\mathbf{x}_{k})d\mathbf{x}_{k}\text{ \cite%
{TCNPNT}}.
\end{eqnarray*}%
and we have shown that $\gamma _{N}^{(k)}(t,\mathbf{x}_{k};\mathbf{x}%
_{k}^{\prime })$ and $\gamma ^{(k)}(t,\mathbf{x}_{k};\mathbf{x}_{k}^{\prime
})$ have no negative eigenvalues, hence 
\begin{eqnarray*}
\limfunc{Tr}\left\vert \gamma _{N}^{(k)}(t,\mathbf{x}_{k};\mathbf{x}%
_{k}^{\prime })\right\vert &=&\int \gamma _{N}^{(k)}(t,\mathbf{x}_{k};%
\mathbf{x}_{k})d\mathbf{x}_{k}=\int \gamma _{N}^{(k)}(0,\mathbf{x}_{k};%
\mathbf{x}_{k})d\mathbf{x}_{k}, \\
\limfunc{Tr}\left\vert \gamma ^{(k)}(t,\mathbf{x}_{k};\mathbf{x}_{k}^{\prime
})\right\vert &=&\int \gamma ^{(k)}(t,\mathbf{x}_{k};\mathbf{x}_{k})d\mathbf{%
x}_{k}=\int \gamma ^{(k)}(0,\mathbf{x}_{k};\mathbf{x}_{k})d\mathbf{x}_{k}.
\end{eqnarray*}%
On the one hand,%
\begin{equation*}
\int \gamma ^{(k)}(0,\mathbf{x}_{k};\mathbf{x}_{k})d\mathbf{x}_{k}=\limfunc{%
Tr}\left\vert \gamma ^{(k)}(0,\mathbf{x}_{k};\mathbf{x}_{k}^{\prime
})\right\vert =\lim_{N\rightarrow \infty }\limfunc{Tr}\left\vert \gamma
_{N}^{(k)}(0,\mathbf{x}_{k};\mathbf{x}_{k}^{\prime })\right\vert
=\lim_{N\rightarrow \infty }\limfunc{Tr}\left\vert \gamma _{N}^{(k)}(t,%
\mathbf{x}_{k};\mathbf{x}_{k}^{\prime })\right\vert ,
\end{equation*}%
on the other hand,%
\begin{equation*}
\limfunc{Tr}\left\vert \gamma ^{(k)}(t,\mathbf{x}_{k};\mathbf{x}_{k}^{\prime
})\right\vert =\int \gamma ^{(k)}(t,\mathbf{x}_{k};\mathbf{x}_{k})d\mathbf{x}%
_{k}=\int \gamma ^{(k)}(0,\mathbf{x}_{k};\mathbf{x}_{k})d\mathbf{x}_{k}.
\end{equation*}%
That is%
\begin{equation*}
\limfunc{Tr}\left\vert \gamma ^{(k)}(t,\mathbf{x}_{k};\mathbf{x}_{k}^{\prime
})\right\vert =\lim_{N\rightarrow \infty }\limfunc{Tr}\left\vert \gamma
_{N}^{(k)}(t,\mathbf{x}_{k};\mathbf{x}_{k}^{\prime })\right\vert .
\end{equation*}%
Whence we conclude the proof of Corollary \ref{Corollary:general data
convergence} by Lemma \ref{Theorem:Simons'}.

\section{Appendix II: Proof of Theorem \protect\ref{Theorem:3*3d}\label%
{Section:Appendix}}

In this Appendix II, we prove Theorem \ref{Theorem:3*3d}. We will make use
of the lemma.

\begin{lemma}
\label{Lemma:MateiLemmaForIntegrals}\cite{KlainermanAndMachedon} Let $\xi
\in \mathbb{R}^{3}$ and $P$ be a 2d plane or sphere in $\mathbb{R}^{3}$ with
the usual induced surface measure $dS$.

(1) Suppose $0<a,$ $b<2,$ $a+b>2,$ then%
\begin{equation*}
\int_{P}\frac{dS(\eta )}{\left\vert \xi -\eta \right\vert ^{a}\left\vert
\eta \right\vert ^{b}}\leqslant \frac{C}{\left\vert \xi \right\vert ^{a+b-2}}%
.
\end{equation*}

(2) Suppose $\varepsilon =\frac{1}{10}$, then%
\begin{equation*}
\int_{P}\frac{dS(\eta )}{\left\vert \frac{\xi }{2}-\eta \right\vert
\left\vert \xi -\eta \right\vert ^{2-\varepsilon }\left\vert \eta
\right\vert ^{2-\varepsilon }}\leqslant \frac{C}{\left\vert \xi \right\vert
^{3-2\varepsilon }}.
\end{equation*}

Both constants in the above estimates are independent of $P.$
\end{lemma}

\begin{proof}
See pages 174 - 175 of \cite{KlainermanAndMachedon}.
\end{proof}

In this appendix II, we write%
\begin{equation*}
V_{N}^{\tau }(x)=g^{3}(\tau )V_{N}(g(\tau )x)
\end{equation*}
which has the property that%
\begin{equation*}
\sup_{\tau ,\xi \mathbf{,}N}\left\vert \widehat{V_{N}^{\tau }}(\xi
)\right\vert \leqslant \int \left\vert V_{N}^{\tau }(x)\right\vert dx=b_{0}.
\end{equation*}

By duality, to gain Theorem \ref{Theorem:3*3d}, it suffices to prove that%
\begin{eqnarray*}
&&\left\vert \int_{\mathbb{R}^{3+1}}h(\tau ,y)\left\vert \nabla
_{y}\right\vert \left( \int V_{N}^{\tau }(y-y_{2})\delta (y_{2}\mathbf{-}%
y_{2}^{\prime })u(\tau ,y,y_{2};y_{2}^{\prime })dy_{2}dy_{2}^{\prime
}\right) dyd\tau \right\vert \\
&\leqslant &Cb_{0}\left\Vert h\right\Vert _{2}\left\Vert \left\vert \nabla
_{y_{1}}\right\vert \left\vert \nabla _{y_{2}}\right\vert \left\vert \nabla
_{y_{2}^{\prime }}\right\vert f\right\Vert _{2}.
\end{eqnarray*}%
From equation \ref{eqn:3*3d}, we compute the spatial Fourier transform of 
\begin{equation*}
\left\vert \nabla _{y}\right\vert \left( \int V_{N}^{\tau }(y-y_{2})\delta
(y_{2}\mathbf{-}y_{2}^{\prime })u(\tau ,y,y_{2};y_{2}^{\prime
})dy_{2}dy_{2}^{\prime }\right)
\end{equation*}%
to be 
\begin{equation*}
\left\vert \xi _{1}\right\vert \int \widehat{V_{N}^{\tau }}(\xi _{2}+\xi
_{2}^{\prime })e^{-\frac{i\tau }{2}(\left\vert \xi _{1}-\xi _{2}-\xi
_{2}^{\prime }\right\vert ^{2}+\left\vert \xi _{2}\right\vert
^{2}-\left\vert \xi _{2}^{\prime }\right\vert ^{2})}\hat{f}(\xi _{1}-\xi
_{2}-\xi _{2}^{\prime },\xi _{2},\xi _{2}^{\prime })d\xi _{2}d\xi
_{2}^{\prime },
\end{equation*}%
thus we have%
\begin{eqnarray*}
&&\left\vert \int_{\mathbb{R}^{3+1}}h(\tau ,y)\left\vert \nabla
_{y}\right\vert \left( \int V_{N}^{\tau }(y-y_{2})\delta (y_{2}\mathbf{-}%
y_{2}^{\prime })u(\tau ,y,y_{2};y_{2}^{\prime })dy_{2}dy_{2}^{\prime
}\right) dyd\tau \right\vert ^{2} \\
&=&\left\vert \int \left\vert \xi _{1}\right\vert \widehat{V_{N}^{\tau }}%
(\xi _{2}+\xi _{2}^{\prime })e^{-\frac{i\tau }{2}(\left\vert \xi _{1}-\xi
_{2}-\xi _{2}^{\prime }\right\vert ^{2}+\left\vert \xi _{2}\right\vert
^{2}-\left\vert \xi _{2}^{\prime }\right\vert ^{2})}\hat{f}(\xi _{1}-\xi
_{2}-\xi _{2}^{\prime },\xi _{2},\xi _{2}^{\prime })\hat{h}(\tau ,\xi
_{1})d\tau d\xi _{1}d\xi _{2}d\xi _{2}^{\prime }\right\vert ^{2} \\
&&(\text{spatial Fourier transform on }h) \\
&=&\bigg|\int \left( \int \widehat{V_{N}^{\tau }}(\xi _{2}+\xi _{2}^{\prime
})\left\vert \xi _{1}\right\vert e^{-\frac{i\tau }{2}(\left\vert \xi
_{1}-\xi _{2}-\xi _{2}^{\prime }\right\vert ^{2}+\left\vert \xi
_{2}\right\vert ^{2}-\left\vert \xi _{2}^{\prime }\right\vert ^{2})}\hat{h}%
(\tau ,\xi _{1})d\tau \right) \\
&&\hat{f}(\xi _{1}-\xi _{2}-\xi _{2}^{\prime },\xi _{2},\xi _{2}^{\prime
})d\xi _{1}d\xi _{2}d\xi _{2}^{\prime }\bigg|^{2} \\
&\leqslant &I(h)\left\Vert \left\vert \nabla _{y_{1}}\right\vert \left\vert
\nabla _{y_{2}}\right\vert \left\vert \nabla _{y_{2}^{\prime }}\right\vert
f\right\Vert _{L^{2}}^{2}\text{ }(\text{Cauchy-Schwarz})
\end{eqnarray*}%
where%
\begin{equation*}
I(h)=\int \frac{\left\vert \xi _{1}\right\vert ^{2}\left\vert \int \widehat{%
V_{N}^{\tau }}(\xi _{2}+\xi _{2}^{\prime })e^{-\frac{i\tau }{2}(\left\vert
\xi _{1}-\xi _{2}-\xi _{2}^{\prime }\right\vert ^{2}+\left\vert \xi
_{2}\right\vert ^{2}-\left\vert \xi _{2}^{\prime }\right\vert ^{2})}\hat{h}%
(\tau ,\xi _{1})d\tau \right\vert ^{2}}{\left\vert \xi _{1}-\xi _{2}-\xi
_{2}^{\prime }\right\vert ^{2}\left\vert \xi _{2}\right\vert ^{2}\left\vert
\xi _{2}^{\prime }\right\vert ^{2}}d\xi _{1}d\xi _{2}d\xi _{2}^{\prime }.
\end{equation*}%
So our purpose in the remainder of this Appendix II is to show that 
\begin{equation*}
I(h)\leqslant Cb_{0}^{2}\left\Vert h\right\Vert _{L^{2}}^{2}.
\end{equation*}%
Noticing that, away from the factor $\widehat{V_{N}^{\tau }}$, the integral $%
I(h)$ is symmetric in $\left\vert \xi _{1}-\xi _{2}-\xi _{2}^{\prime
}\right\vert $ and $\left\vert \xi _{2}\right\vert ,$ it suffices that we
deal with the region: $\left\vert \xi _{1}-\xi _{2}-\xi _{2}^{\prime
}\right\vert >\left\vert \xi _{2}\right\vert $ only since our proof treats $%
\widehat{V_{N}^{\tau }}$ as a harmless factor. We separate this region into
two parts, Cases I and II.

Away from the region $\left\vert \xi _{1}-\xi _{2}-\xi _{2}^{\prime
}\right\vert >\left\vert \xi _{2}\right\vert $, there are other restrictions
on the integration regions in Cases I and II. We state the restrictions in
the beginning of both Cases I and II. Due to the limited space near "$\int $%
", we omit the actual region. The alert reader should bear this mind.

\subsubsection{Case I: $I(h)$ restricted to the region $\left\vert \protect%
\xi _{2}^{\prime }\right\vert <\left\vert \protect\xi _{2}\right\vert $ with
integration order $d\protect\xi _{2}$ prior to $d\protect\xi _{2}^{\prime }$}

Write the phase function of the $d\tau $ integral inside $I(h)$ as%
\begin{equation*}
\left\vert \xi _{1}-\xi _{2}-\xi _{2}^{\prime }\right\vert ^{2}+\left\vert
\xi _{2}\right\vert ^{2}-\left\vert \xi _{2}^{\prime }\right\vert ^{2}=\frac{%
\left\vert \xi _{1}-\xi _{2}^{\prime }\right\vert ^{2}}{2}+2\left\vert \xi
_{2}-\frac{\xi _{1}-\xi _{2}^{\prime }}{2}\right\vert ^{2}-\left\vert \xi
_{2}^{\prime }\right\vert ^{2}.
\end{equation*}%
The change of variable%
\begin{equation}
\xi _{2,new}=\xi _{2,old}-\frac{\xi _{1}-\xi _{2}^{\prime }}{2}
\label{formula:change of variable 3*3d sphere}
\end{equation}%
leads to the expression%
\begin{equation*}
I(h)=\int \frac{\left\vert \xi _{1}\right\vert ^{2}\left\vert \int \widehat{%
V_{N}^{\tau }}(\xi _{2}+\frac{\xi _{1}+\xi _{2}^{\prime }}{2})e^{-\frac{%
i\tau }{2}\left( \frac{\left\vert \xi _{1}-\xi _{2}^{\prime }\right\vert ^{2}%
}{2}+2\left\vert \xi _{2}\right\vert ^{2}-\left\vert \xi _{2}^{\prime
}\right\vert ^{2}\right) }\hat{h}(\tau ,\xi _{1})d\tau \right\vert ^{2}}{%
\left\vert \xi _{2}-\frac{\xi _{1}-\xi _{2}^{\prime }}{2}\right\vert
^{2}\left\vert \xi _{2}+\frac{\xi _{1}\mathbf{-}\xi _{2}^{\prime }}{2}%
\right\vert ^{2}\left\vert \xi _{2}^{\prime }\right\vert ^{2}}d\xi _{1}d\xi
_{2}d\xi _{2}^{\prime }.
\end{equation*}%
Write out the square,%
\begin{eqnarray*}
I(h) &=&\int \frac{\left\vert \xi _{1}\right\vert ^{2}}{\left\vert \xi _{2}-%
\frac{\xi _{1}-\xi _{2}^{\prime }}{2}\right\vert ^{2}\left\vert \xi _{2}+%
\frac{\xi _{1}-\xi _{2}^{\prime }}{2}\right\vert ^{2}\left\vert \xi
_{2}^{\prime }\right\vert ^{2}}\widehat{V_{N}^{\tau }}(\xi _{2}+\frac{\xi
_{1}+\xi _{2}^{\prime }}{2})\overline{\widehat{V_{N}^{\tau ^{\prime }}}(\xi
_{2}+\frac{\xi _{1}+\xi _{2}^{\prime }}{2})} \\
&&e^{-\frac{i\left( \tau -\tau ^{\prime }\right) }{2}\left( \frac{\left\vert
\xi _{1}-\xi _{2}^{\prime }\right\vert ^{2}}{2}+2\left\vert \xi
_{2}\right\vert ^{2}-\left\vert \xi _{2}^{\prime }\right\vert ^{2}\right) }%
\hat{h}(\tau ,\xi _{1})\overline{\hat{h}(\tau ^{\prime },\xi _{1})}d\tau
d\tau ^{\prime }d\xi _{1}d\xi _{2}d\xi _{2}^{\prime } \\
&=&\int d\xi _{1}\int J(\overline{\hat{h}})(\tau ,\xi _{1})\hat{h}(\tau ,\xi
_{1})d\tau
\end{eqnarray*}%
where 
\begin{eqnarray*}
J(\overline{\hat{h}})(\tau ,\xi _{1}) &=&\int \frac{\left\vert \xi
_{1}\right\vert ^{2}e^{-i\tau \left\vert \xi _{2}\right\vert ^{2}}e^{i\tau
^{\prime }\left\vert \xi _{2}\right\vert ^{2}}\widehat{V_{N}^{\tau }}(\xi
_{2}+\frac{\xi _{1}+\xi _{2}^{\prime }}{2})\overline{\widehat{V_{N}^{\tau
^{\prime }}}(\xi _{2}+\frac{\xi _{1}+\xi _{2}^{\prime }}{2})}}{\left\vert
\xi _{2}-\frac{\xi _{1}-\xi _{2}^{\prime }}{2}\right\vert ^{2}\left\vert \xi
_{2}+\frac{\xi _{1}-\xi _{2}^{\prime }}{2}\right\vert ^{2}\left\vert \xi
_{2}^{\prime }\right\vert ^{2}} \\
&&e^{-\frac{i\left( \tau -\tau ^{\prime }\right) }{2}\left( \frac{\left\vert
\xi _{1}-\xi _{2}^{\prime }\right\vert ^{2}}{2}-\left\vert \xi _{2}^{\prime
}\right\vert ^{2}\right) }\overline{\hat{h}(\tau ^{\prime },\xi _{1})}d\tau
^{\prime }d\xi _{2}d\xi _{2}^{\prime }.
\end{eqnarray*}

Assume for the moment that 
\begin{equation*}
\int \left\vert J(\overline{\hat{h}})(\tau ,\xi _{1})\right\vert ^{2}d\tau
\leqslant Cb_{0}^{2}\left\Vert \hat{h}(\cdot ,\xi _{1})\right\Vert _{L_{\tau
}^{2}}^{2}
\end{equation*}%
with $C$ independent of $h$ or $\xi _{1}$, then we deduce that%
\begin{equation*}
I(h)\leqslant Cb_{0}^{2}\int d\xi _{1}\left\Vert \hat{h}(\cdot ,\xi
_{1})\right\Vert _{L_{\tau }^{2}}^{2}.
\end{equation*}

Hence we end Case I by this proposition.

\begin{proposition}
\begin{equation*}
\int \left\vert J(f)(\tau ,\xi _{1})\right\vert ^{2}d\tau \leqslant
Cb_{0}^{2}\left\Vert f(\cdot ,\xi _{1})\right\Vert _{L_{\tau }^{2}}^{2}
\end{equation*}%
where $C$ is independent of $f$ or $\xi _{1}.$
\end{proposition}

\begin{remark}
To avoid confusing notation in the proof of the proposition, we use $f(\tau
^{\prime },\xi _{1})$ to replace $\overline{\hat{h}(\tau ^{\prime },\xi _{1})%
}.$
\end{remark}

\begin{proof}
Again, by duality, we just need to prove%
\begin{equation*}
\left\vert \int J(f)(\tau ,\xi _{1})\overline{g(\tau )}d\tau \right\vert
\leqslant C\left\Vert f(\cdot ,\xi _{1})\right\Vert _{L_{\tau
}^{2}}\left\Vert g\right\Vert _{L_{\tau }^{2}}.
\end{equation*}%
For convenience, let%
\begin{equation*}
\phi (\tau ,\xi _{1},\xi _{2}^{\prime })=\frac{\tau }{2}\left( \frac{%
\left\vert \xi _{1}-\xi _{2}^{\prime }\right\vert ^{2}}{2}-\left\vert \xi
_{2}^{\prime }\right\vert ^{2}\right) .
\end{equation*}%
Then%
\begin{eqnarray*}
&&\left\vert \int J(f)(\tau ,\xi _{1})\overline{g(\tau )}d\tau \right\vert \\
&=&\bigg|\int \frac{\left\vert \xi _{1}\right\vert ^{2}d\xi _{2}d\xi
_{2}^{\prime }}{\left\vert \xi _{2}-\frac{\xi _{1}-\xi _{2}^{\prime }}{2}%
\right\vert ^{2}\left\vert \xi _{2}+\frac{\xi _{1}-\xi _{2}^{\prime }}{2}%
\right\vert ^{2}\left\vert \xi _{2}^{\prime }\right\vert ^{2}}\left[ \int
e^{-i\tau \left\vert \xi _{2}\right\vert ^{2}}\left( \widehat{V_{N}^{\tau }}%
(\xi _{2}+\frac{\xi _{1}+\xi _{2}^{\prime }}{2})\overline{e^{i\phi (\tau
,\xi _{1},\xi _{2}^{\prime })}g(\tau )}\right) d\tau \right] \\
&&\left[ \int e^{i\tau ^{\prime }\left\vert \xi _{2}\right\vert ^{2}}\left( 
\overline{\widehat{V_{N}^{\tau ^{\prime }}}(\xi _{2}+\frac{\xi _{1}+\xi
_{2}^{\prime }}{2})}e^{i\phi (\tau ^{\prime },\xi _{1},\xi _{2}^{\prime
})}f(\tau ^{\prime },\xi _{1})\right) d\tau ^{\prime }\right] \bigg| \\
&\leqslant &\int \frac{\left\vert \xi _{1}\right\vert ^{2}d\xi _{2}^{\prime }%
}{\left\vert \xi _{2}^{\prime }\right\vert ^{2}}\int \frac{d\xi _{2}}{%
\left\vert \xi _{2}-\frac{\xi _{1}-\xi _{2}^{\prime }}{2}\right\vert
^{2}\left\vert \xi _{2}+\frac{\xi _{1}-\xi _{2}^{\prime }}{2}\right\vert ^{2}%
}\left\vert \int e^{-i\tau \left\vert \xi _{2}\right\vert ^{2}}\left( 
\widehat{V_{N}^{\tau }}(\xi _{2}+\frac{\xi _{1}+\xi _{2}^{\prime }}{2})%
\overline{e^{i\phi (\tau ,\xi _{1},\xi _{2}^{\prime })}g(\tau )}\right)
d\tau \right\vert \\
&&\left\vert \int e^{i\tau ^{\prime }\left\vert \xi _{2}\right\vert
^{2}}\left( \overline{\widehat{V_{N}^{\tau ^{\prime }}}(\xi _{2}+\frac{\xi
_{1}+\xi _{2}^{\prime }}{2})}e^{i\phi (\tau ^{\prime },\xi _{1},\xi
_{2}^{\prime })}f(\tau ^{\prime },\xi _{1})\right) d\tau ^{\prime
}\right\vert .
\end{eqnarray*}%
To deal with the $d\tau $ and $d\tau ^{\prime }$ integrals, let 
\begin{equation*}
G(\tau )=\widehat{V_{N}^{\tau }}(\xi _{2}+\frac{\xi _{1}+\xi _{2}^{\prime }}{%
2})\overline{e^{i\phi (\tau ,\xi _{1},\xi _{2}^{\prime })}g(\tau )},
\end{equation*}

then%
\begin{equation*}
\int e^{-i\tau \left\vert \xi _{2}\right\vert ^{2}}\left( \widehat{%
V_{N}^{\tau }}(\xi _{2}+\frac{\xi _{1}+\xi _{2}^{\prime }}{2})\overline{%
e^{i\phi (\tau ,\xi _{1},\xi _{2}^{\prime })}g(\tau )}\right) d\tau =\hat{G}%
(\left\vert \xi _{2}\right\vert ^{2}).
\end{equation*}%
This is well-defined since%
\begin{eqnarray*}
\int_{\mathbb{R}}\left\vert G(\tau )\right\vert ^{2}d\tau &=&\int_{\mathbb{R}%
}\left\vert \widehat{V_{N}^{\tau }}(\xi _{2}+\frac{\xi _{1}+\xi _{2}^{\prime
}}{2})\right\vert ^{2}\left\vert \overline{e^{i\phi (\tau ,\xi _{1},\xi
_{2}^{\prime })}g(\tau )}\right\vert ^{2}d\tau \\
&\leqslant &\sup_{\tau ,N}\left\Vert \widehat{V_{N}^{\tau }}(\mathbf{\cdot }%
)\right\Vert _{L_{\xi }^{\infty }}^{2}\int_{\mathbb{R}}\left\vert g(\tau
)\right\vert ^{2}d\tau \\
&\leqslant &b_{0}^{2}\left\Vert g(\cdot )\right\Vert _{L_{\tau }^{2}}^{2}.
\end{eqnarray*}%
Hence, we find%
\begin{eqnarray*}
\left\vert \int J(f)(\tau ,\xi _{1})\overline{g(\tau )}d\tau \right\vert
&\leqslant &\int \frac{\left\vert \xi _{1}\right\vert ^{2}d\xi _{2}^{\prime }%
}{\left\vert \xi _{2}^{\prime }\right\vert ^{2}}\int \frac{\left\vert \hat{G}%
(\left\vert \xi _{2}\right\vert ^{2})\right\vert \left\vert \overline{\hat{F}%
}(\left\vert \xi _{2}\right\vert ^{2},\xi _{1})\right\vert d\xi _{2}}{%
\left\vert \xi _{2}-\frac{\xi _{1}-\xi _{2}^{\prime }}{2}\right\vert
^{2}\left\vert \xi _{2}+\frac{\xi _{1}-\xi _{2}^{\prime }}{2}\right\vert ^{2}%
} \\
&=&\int \frac{\left\vert \xi _{1}\right\vert ^{2}d\xi _{2}^{\prime }}{%
\left\vert \xi _{2}^{\prime }\right\vert ^{2}}\int \frac{\left\vert \hat{G}%
(\rho ^{2})\right\vert \left\vert \hat{F}(\rho ^{2},\xi _{1})\right\vert
\rho ^{2}d\rho d\sigma }{\left\vert \xi _{2}-\frac{\xi _{1}-\xi _{2}^{\prime
}}{2}\right\vert ^{2}\left\vert \xi _{2}+\frac{\xi _{1}-\xi _{2}^{\prime }}{2%
}\right\vert ^{2}}
\end{eqnarray*}%
if we use spherical coordinate in $\xi _{2}.$ Apply H\"{o}lder in $\rho $, 
\begin{eqnarray*}
&\leqslant &\int \frac{\left\vert \xi _{1}\right\vert ^{2}d\xi _{2}^{\prime }%
}{\left\vert \xi _{2}^{\prime }\right\vert ^{2}}\sup_{\rho }\left( \int 
\frac{\rho ^{2}d\sigma }{\rho \left\vert \xi _{2}-\frac{\xi _{1}-\xi
_{2}^{\prime }}{2}\right\vert ^{2}\left\vert \xi _{2}+\frac{\xi _{1}-\xi
_{2}^{\prime }}{2}\right\vert ^{2}}\right) \left( \int \left\vert \hat{F}%
(\rho ^{2},\xi _{1})\right\vert ^{2}\rho d\rho \right) ^{\frac{1}{2}}\left(
\int \left\vert \hat{G}(\rho ^{2})\right\vert ^{2}\rho d\rho \right) ^{\frac{%
1}{2}} \\
&\leqslant &b_{0}^{2}\left\Vert f(\cdot ,\xi _{1})\right\Vert _{L_{\tau
}^{2}}\left\Vert g\right\Vert _{L_{\tau }^{2}}\int \frac{\left\vert \xi
_{1}\right\vert ^{2}}{\left\vert \xi _{2}^{\prime }\right\vert ^{2}}%
\sup_{\rho }\left( \int \frac{\rho ^{2}d\sigma }{\rho \left\vert \xi _{2}-%
\frac{\xi _{1}-\xi _{2}^{\prime }}{2}\right\vert ^{2}\left\vert \xi _{2}+%
\frac{\xi _{1}-\xi _{2}^{\prime }}{2}\right\vert ^{2}}\right) d\xi
_{2}^{\prime }.
\end{eqnarray*}%
Reverse the change of variable in formula \ref{formula:change of variable
3*3d sphere}, we find%
\begin{eqnarray*}
&&\int \frac{\left\vert \xi _{1}\right\vert ^{2}}{\left\vert \xi
_{2}^{\prime }\right\vert ^{2}}\sup_{\rho }\left( \int \frac{\rho
^{2}d\sigma }{\rho \left\vert \xi _{2}-\frac{\xi _{1}-\xi _{2}^{\prime }}{2}%
\right\vert ^{2}\left\vert \xi _{2}+\frac{\xi _{1}-\xi _{2}^{\prime }}{2}%
\right\vert ^{2}}\right) d\xi _{2}^{\prime } \\
&=&\int \frac{\left\vert \xi _{1}\right\vert ^{2}}{\left\vert \xi
_{2}^{\prime }\right\vert ^{2}}\sup_{\rho }\left( \int \frac{\left\vert \xi
_{2}-\frac{\xi _{1}-\xi _{2}^{\prime }}{2}\right\vert ^{2}d\sigma }{%
\left\vert \xi _{2}-\frac{\xi _{1}-\xi _{2}^{\prime }}{2}\right\vert
\left\vert \xi _{1}-\xi _{2}-\xi _{2}^{\prime }\right\vert ^{2}\left\vert
\xi _{2}\right\vert ^{2}}\right) d\xi _{2}^{\prime } \\
&\leqslant &\left\vert \xi _{1}\right\vert ^{2}\int \frac{d\xi _{2}^{\prime }%
}{\left\vert \xi _{2}^{\prime }\right\vert ^{2+2\varepsilon }}\sup_{\rho
}\left( \int \frac{\left\vert \xi _{2}-\frac{\xi _{1}-\xi _{2}^{\prime }}{2}%
\right\vert ^{2}d\sigma }{\left\vert \xi _{2}-\frac{\xi _{1}-\xi
_{2}^{\prime }}{2}\right\vert \left\vert \xi _{1}-\xi _{2}-\xi _{2}^{\prime
}\right\vert ^{2-\varepsilon }\left\vert \xi _{2}\right\vert ^{2-\varepsilon
}}\right) \\
&\leqslant &C\left\vert \xi _{1}\right\vert ^{2}\int \frac{d\xi _{2}^{\prime
}}{\left\vert \xi _{2}^{\prime }\right\vert ^{2+2\varepsilon }\left\vert \xi
_{1}-\xi _{2}^{\prime }\right\vert ^{3-2\varepsilon }}\text{ (Second part of
Lemma \ref{Lemma:MateiLemmaForIntegrals})} \\
&\leqslant &C.
\end{eqnarray*}

In the above calculation, the $\sigma $ in the first line lies on the unit
sphere centered at the origin while the $\sigma $ in the second line is on a
unit sphere centered at $\frac{\xi _{1}-\xi _{2}^{\prime }}{2}$. We use the
same symbol because Lebesgue measure is translation invariant.

Thus, we conclude that%
\begin{equation*}
\left\vert \int J(f)(\tau ,\xi _{1})\overline{g(\tau )}d\tau \right\vert
\leqslant Cb_{0}^{2}\left\Vert f(\cdot ,\xi _{1})\right\Vert _{L_{\tau
}^{2}}\left\Vert g\right\Vert _{L_{\tau }^{2}}.
\end{equation*}
\end{proof}

\subsubsection{Case II: $I(h)$ restricted to the region $\left\vert \protect%
\xi _{2}^{\prime }\right\vert >\left\vert \protect\xi _{2}\right\vert $ with
integration order $d\protect\xi _{2}^{\prime }$ prior to $d\protect\xi _{2}$}

For this case, we express the phase function as%
\begin{eqnarray*}
\frac{\tau }{2}\left( \left\vert \xi _{1}-\xi _{2}-\xi _{2}^{\prime
}\right\vert ^{2}+\left\vert \xi _{2}\right\vert ^{2}-\left\vert \xi
_{2}^{\prime }\right\vert ^{2}\right) &=&\frac{\tau }{2}\left( \left\vert
\xi _{1}-\xi _{2}\right\vert ^{2}-2\left( \xi _{1}-\xi _{2}\right) \cdot \xi
_{2}^{\prime }+\left\vert \xi _{2}\right\vert ^{2}\right) \\
&=&\phi (\tau ,\xi _{1},\xi _{2})-\tau \left( \xi _{1}-\xi _{2}\right) \cdot
\xi _{2}^{\prime }
\end{eqnarray*}%
and let%
\begin{eqnarray*}
J\left( \overline{\hat{h}}\right) (\tau ,\xi _{1}) &=&\int \frac{\left\vert
\xi _{1}\right\vert ^{2}e^{-i\tau \left( \xi _{1}-\xi _{2}\right) \cdot \xi
_{2}^{\prime }}e^{i\tau ^{\prime }\left( \xi _{1}-\xi _{2}\right) \cdot \xi
_{2}^{\prime }}}{\left\vert \xi _{1}-\xi _{2}-\xi _{2}^{\prime }\right\vert
^{2}\left\vert \xi _{2}\right\vert ^{2}\left\vert \xi _{2}^{\prime
}\right\vert ^{2}} \\
&&\widehat{V_{N}^{\tau }}(\xi _{2}+\xi _{2}^{\prime })\overline{\widehat{%
V_{N}^{\tau ^{\prime }}}(\xi _{2}+\xi _{2}^{\prime })}e^{i\phi (\tau
^{\prime },\xi _{1},\xi _{2})}\overline{e^{i\phi (\tau ,\xi _{1},\xi _{2})}}%
\overline{\hat{h}(\tau ^{\prime },\xi _{1})}d\tau ^{\prime }d\xi
_{2}^{\prime }d\xi _{2}.
\end{eqnarray*}%
Again, we want to prove the following property of $J.$

\begin{proposition}
\begin{equation*}
\int \left\vert J(f)(\tau ,\xi _{1})\right\vert ^{2}d\tau \leqslant
Cb_{0}^{2}\left\Vert f(\cdot ,\xi _{1})\right\Vert _{L_{\tau }^{2}}^{2}
\end{equation*}%
where the constant $C$ is independent of $f$ or $\xi _{1}.$
\end{proposition}

\begin{proof}
We calculate%
\begin{eqnarray*}
&&\left\vert \int J(f)(\tau ,\xi _{1})\overline{g(\tau )}d\tau \right\vert \\
&\leqslant &\int \frac{\left\vert \xi _{1}\right\vert ^{2}d\xi _{2}}{%
\left\vert \xi _{2}\right\vert ^{2}}\int \frac{d\xi _{2}^{\prime }}{%
\left\vert \xi _{1}-\xi _{2}-\xi _{2}^{\prime }\right\vert ^{2}\left\vert
\xi _{2}^{\prime }\right\vert ^{2}}\left\vert \int e^{-i\tau \left( \xi
_{1}-\xi _{2}\right) \cdot \xi _{2}^{\prime }}\left( \widehat{V_{N}^{\tau }}%
(\xi _{2}+\xi _{2}^{\prime })\overline{e^{i\phi (\tau ,\xi _{1},\xi
_{2})}g(\tau )}\right) d\tau \right\vert \\
&&\left\vert \int e^{i\tau ^{\prime }\left( \xi _{1}-\xi _{2}\right) \cdot
\xi _{2}^{\prime }}\left( \overline{\widehat{V_{N}^{\tau ^{\prime }}}(\xi
_{2}+\xi _{2}^{\prime })}e^{i\phi (\tau ^{\prime },\xi _{1},\xi _{2})}f(\tau
^{\prime },\xi _{1})\right) d\tau ^{\prime }\right\vert .
\end{eqnarray*}%
For the purpose of the $d\tau $ and $d\tau ^{\prime }$ integrals, let $\xi
_{1}-\xi _{2}=\left\vert \xi _{1}-\xi _{2}\right\vert \omega $, where $%
\omega $ is a unit vector in $\mathbb{R}^{3}$. Without loss of generality,
we may assume that $\omega =(1,0,0)$ and write $\xi _{2}^{\prime }=(x,y,z),$
we then have%
\begin{eqnarray*}
&&\int e^{-i\tau \left( \xi _{1}-\xi _{2}\right) \cdot \xi _{2}^{\prime
}}\left( \widehat{V_{N}^{\tau }}(\xi _{2}+\xi _{2}^{\prime })\overline{%
e^{i\phi (\tau ,\xi _{1},\xi _{2})}g(\tau )}\right) d\tau \\
&=&\int e^{-i\tau \left\vert \xi _{1}-\xi _{2}\right\vert \omega \cdot \xi
_{2}^{\prime }}\left( \widehat{V_{N}^{\tau }}(\xi _{2}+\xi _{2}^{\prime })%
\overline{e^{i\phi (\tau ,\xi _{1},\xi _{2})}g(\tau )}\right) d\tau \\
&=&\int e^{-i\tau \left\vert \xi _{1}-\xi _{2}\right\vert x}\left( \widehat{%
V_{N}^{\tau }}(\xi _{2}+\xi _{2}^{\prime })\overline{e^{i\phi (\tau ,\xi
_{1},\xi _{2})}g(\tau )}\right) du \\
&=&\hat{G}(\left\vert \xi _{1}-\xi _{2}\right\vert x),
\end{eqnarray*}%
where%
\begin{equation*}
G(\tau )=\left( \widehat{V_{N}^{\tau }}(\xi _{2}+\xi _{2}^{\prime })%
\overline{e^{i\phi (\tau ,\xi _{1},\xi _{2})}g(\tau )}\right) ,
\end{equation*}%
which still has the property that 
\begin{equation*}
\int \left\vert G(\tau )\right\vert ^{2}d\tau \leqslant b_{0}^{2}\int
\left\vert g(\tau )\right\vert ^{2}d\tau .
\end{equation*}%
Just as in case 1, this procedure furnishes%
\begin{eqnarray*}
&&\left\vert \int J(f)(\tau ,\xi _{1})\overline{g(\tau )}d\tau \right\vert \\
&\leqslant &\int \frac{\left\vert \xi _{1}\right\vert ^{2}d\xi _{2}}{%
\left\vert \xi _{2}\right\vert ^{2}}\int \frac{dxdydz}{\left\vert \xi
_{1}-\xi _{2}-\xi _{2}^{\prime }\right\vert ^{2}\left\vert \xi _{2}^{\prime
}\right\vert ^{2}}\left\vert \hat{G}(\left\vert \xi _{1}-\xi _{2}\right\vert
x)\right\vert \left\vert \overline{\hat{F}}(\left\vert \xi _{1}-\xi
_{2}\right\vert x,\xi _{1})\right\vert \\
&=&\int \left( \int \frac{dxdydz}{\left\vert \xi _{1}-\xi _{2}-\xi
_{2}^{\prime }\right\vert ^{2}\left\vert \xi _{2}^{\prime }\right\vert ^{2}}%
\left\vert \hat{G}(x)\right\vert \left\vert \hat{F}(x,\xi _{1})\right\vert
\right) \frac{\left\vert \xi _{1}\right\vert ^{2}}{\left\vert \xi _{1}-\xi
_{2}\right\vert \left\vert \xi _{2}\right\vert ^{2}}d\xi _{2}
\end{eqnarray*}%
Apply H\"{o}lder in $x,$ we have 
\begin{eqnarray*}
&\leqslant &\int \frac{\left\vert \xi _{1}\right\vert ^{2}}{\left\vert \xi
_{1}-\xi _{2}\right\vert \left\vert \xi _{2}\right\vert ^{2}}\left(
\sup_{x}\int \frac{dydz}{\left\vert \xi _{1}-\xi _{2}-\xi _{2}^{\prime
}\right\vert ^{2}\left\vert \xi _{2}^{\prime }\right\vert ^{2}}\right)
\left( \int \left\vert \hat{F}(x,\xi _{1})\right\vert ^{2}dx\right) ^{\frac{1%
}{2}}\left( \int \left\vert \hat{G}(x)\right\vert ^{2}dx\right) ^{\frac{1}{2}%
}d\xi _{2} \\
&\leqslant &b_{0}^{2}\left\Vert f(\cdot ,\xi _{1})\right\Vert _{L_{\tau
}^{2}}\left\Vert g\right\Vert _{L_{\tau }^{2}}\int \frac{\left\vert \xi
_{1}\right\vert ^{2}}{\left\vert \xi _{1}-\xi _{2}\right\vert \left\vert \xi
_{2}\right\vert ^{2}}\left( \sup_{x}\int \frac{dydz}{\left\vert \xi _{1}-\xi
_{2}-\xi _{2}^{\prime }\right\vert ^{2}\left\vert \xi _{2}^{\prime
}\right\vert ^{2}}\right) d\xi _{2}.
\end{eqnarray*}%
The first part of Lemma \ref{Lemma:MateiLemmaForIntegrals} along with the
restrictions $\left\vert \xi _{1}-\xi _{2}-\xi _{2}^{\prime }\right\vert
>\left\vert \xi _{2}\right\vert $ and $\left\vert \xi _{2}^{\prime
}\right\vert <\left\vert \xi _{2}\right\vert $ entail 
\begin{eqnarray*}
&&\int \frac{\left\vert \xi _{1}\right\vert ^{2}}{\left\vert \xi _{1}-\xi
_{2}\right\vert \left\vert \xi _{2}\right\vert ^{2}}\left( \sup_{x}\int 
\frac{dydz}{\left\vert \xi _{1}-\xi _{2}-\xi _{2}^{\prime }\right\vert
^{2}\left\vert \xi _{2}^{\prime }\right\vert ^{2}}\right) d\xi _{2} \\
&\leqslant &\int \frac{\left\vert \xi _{1}\right\vert ^{2}}{\left\vert \xi
_{1}-\xi _{2}\right\vert \left\vert \xi _{2}\right\vert ^{2+2\varepsilon }}%
\left( \sup_{x}\int \frac{dydz}{\left\vert \xi _{1}-\xi _{2}-\xi
_{2}^{\prime }\right\vert ^{2-\varepsilon }\left\vert \xi _{2}^{\prime
}\right\vert ^{2-\varepsilon }}\right) d\xi _{2} \\
&\leqslant &C\int \frac{\left\vert \xi _{1}\right\vert ^{2}d\xi _{2}}{%
\left\vert \xi _{1}-\xi _{2}\right\vert ^{3-2\varepsilon }\left\vert \xi
_{2}\right\vert ^{2+2\varepsilon }} \\
&\leqslant &C,
\end{eqnarray*}%
which finishes the proposition.
\end{proof}

\end{document}